\documentclass[12pt]{article}
\usepackage{amsmath}
\usepackage{amssymb}
\usepackage{amsthm}
\usepackage{tikz}
\usepackage{float}
\usetikzlibrary{decorations.markings,decorations.pathreplacing,calc,automata,quantikz}
\usetikzlibrary{arrows.meta,shadows, fadings,shapes.arrows}

\tikzset{arrow1/.style={single arrow, fill=blue!20, anchor=base, align=center,text width=2.8cm}}

\makeatletter
\newtheorem*{rep@theorem}{\rep@title}
\newcommand{\newreptheorem}[2]{%
	\newenvironment{rep#1}[1]{%
		\def\rep@title{#2 \ref{##1}}%
		\begin{rep@theorem}}%
		{\end{rep@theorem}}}
\makeatother
\newreptheorem{theorem}{Theorem}
\newreptheorem{lemma}{Lemma}
\newreptheorem{problem}{Problem}

\usepackage{algorithm}
\usepackage{algpseudocode}

\newcommand*\Let[2]{\State #1 $\gets$ #2}

\usepackage{thm-restate}
\usepackage{verbatim}
\usepackage[framemethod=tikz]{mdframed}
\usepackage{complexity}

\usepackage[margin=1in]{geometry}
\usepackage[utf8]{inputenc}

\usepackage[colorlinks = true]{hyperref}
\hypersetup{pdftitle={Classical algorithms for Forrelation}}
\usepackage{subcaption}
\usepackage{xcolor}
\definecolor{darkred}  {rgb}{0.5,0,0}
\definecolor{darkblue} {rgb}{0,0,0.5}
\definecolor{darkgreen}{rgb}{0,0.5,0}

\hypersetup{
  urlcolor   = black,         
  linkcolor  = darkblue,     
  citecolor  = darkgreen,    
  filecolor  = darkred       
}

\newtheorem{dfn}{Definition}
\newtheorem{prop}{Proposition}

\newtheorem{claim}{Claim}
\newtheorem{lemma}{Lemma}
\newtheorem{corol}{Corollary}
\newtheorem{fact}{Fact}

\newtheorem{theorem}{Theorem}
\newtheorem*{theorem*}{Theorem}
\newtheorem{problem}{Problem}

\newtheorem*{forrproblem}{Forrelation problem}

\newcommand{\be}{\begin{equation}}
\newcommand{\ee}{\end{equation}}
\newcommand{\ra}{\rangle}
\newcommand{\la}{\langle}
\usepackage{braket}

\newcommand{\calA}{{\cal A }}
\newcommand{\calN}{{\cal N }}

\newcommand{\footremember}[2]{%
    \footnote{#2}
    \newcounter{#1}
    \setcounter{#1}{\value{footnote}}%
}
\newcommand{\footrecall}[1]{%
    \footnotemark[\value{#1}]%
}

\usepackage{todonotes}
\DeclareMathOperator{\Tr}{Tr}
\DeclareMathOperator{\opvec}{vec}

\title{Classical algorithms for Forrelation}
\author{Sergey Bravyi\footremember{ibm}{IBM Quantum, IBM T.J. Watson Research Center}
\and
David Gosset\footremember{iqc}{Institute for Quantum Computing, University of Waterloo, Canada}\footremember{co}{Department of Combinatorics and Optimization, University of Waterloo, Canada}
\and Daniel Grier\footrecall{iqc} \footremember{cs}{Cheriton School of Computer Science, University of Waterloo, Canada}
\and Luke Schaeffer\footrecall{iqc} \footrecall{co}
}
\date{}
\begin{document}
\maketitle
\begin{abstract}
We study the forrelation problem: given a pair of  $n$-bit Boolean functions
$f$ and $g$, estimate the correlation between $f$
and the Fourier transform of $g$. This problem is known to provide the largest possible quantum
speedup in terms of its query complexity and achieves the landmark oracle separation between the 
complexity class BQP and the Polynomial Hierarchy. 
Our first result is a classical algorithm for the forrelation problem which has runtime $O(n2^{n/2})$.
This is a nearly quadratic improvement over the best previously known algorithm.
Secondly, we show that quantum query algorithm that makes $t$ queries to an $n$-bit oracle can be simulated by classical query algorithm making only $O(2^{n(1-1/2t)})$ queries. 
This fixes a gap in the literature arising from a recently discovered critical error in a previous proof; it matches recently established lower bounds (up to $\mathrm{poly}(n,t))$ factors) and thus characterizes the maximal separation in query complexity between quantum and classical algorithms. Finally, we introduce a graph-based forrelation problem where $n$ binary variables
live at vertices of some fixed graph and the functions $f,g$  are products of terms describing interactions between nearest-neighbor variables.
We show that the graph-based forrelation problem can be solved
on a classical computer in time $O(n)$ for any bipartite graph, any planar graph,
or, more generally, any graph which can be partitioned into two subgraphs of constant treewidth.  
The graph-based forrelation is simply related to the variational energy achieved by 
the Quantum Approximate Optimization Algorithm (QAOA) with two entangling layers and Ising-type cost functions.
By exploiting the connection between QAOA and the graph-based forrelation we were able to simulate the 
recently proposed Recursive QAOA with two entangling layers
and $225$ qubits on a laptop computer.

\end{abstract}

\section{Introduction}

Much of what is known about the power of quantum computers has been learned by studying the black-box model of \textit{quantum query complexity}. Here one considers a kind of computational problem where the input is a binary string $x\in \{-1,1\}^N$ and the goal is to compute some property of the input by accessing or querying as few bits of $x$ as possible. In a classical algorithm we may query the input bits one at a time, whereas in a quantum algorithm each query is performed by applying a unitary oracle $O_x$ satisfying $O_x|i\rangle=x_i|i\rangle$ for $i\in [N]$. Many computational problems have been shown to admit quantum speedups as measured by query complexity, including e.g.,  quantum search \cite{grover1997quantum}, period finding \cite{shor1999polynomial}, and the hidden subgroup problem \cite{mosca1998hidden, ettinger2004quantum}. These and other provable speedups in quantum query complexity inform the search for real-world quantum advantage.  For example, Shor's algorithm for integer factorization is based on the quantum algorithm for period finding \cite{shor1999polynomial}. The conjectured real-world speedup for this problem is predicated on the belief that the additional structure present in the integer factorization problem does not make it significantly easier for classical computers.

In this paper we study the forrelation problem which was introduced in Ref. \cite{aaronson2010bqp}.  Forrelation is a powerful computational primitive: it achieves an almost maximal quantum speedup as measured by query complexity \cite{aaronson2018forrelation, tal2019towards, bansal2020kforrelation, sherstov2020optimal} and it underlies the dramatic oracle separation between the complexity classes BQP and PH \cite{raz2019oracle}. Moreover, a generalized $k$-fold forrelation problem is BQP-complete \cite{aaronson2018forrelation}. We note that the primary focus of these previous works was to lower bound the classical query complexity of forrelation and related problems. In this paper we develop new classical algorithms for forrelation. Firstly in the black-box setting, and then in an explicit graph-based setting which is related to the performance of certain quantum algorithms for optimization \cite{farhi2014quantum}.

To describe our results in more detail, let us now define the forrelation problem, following Aaronson and Ambainis \cite{aaronson2010bqp, aaronson2018forrelation}.  The input to the problem is a pair of Boolean functions $f,g\, \colon \, \{0,1\}^n \to \{1,-1\}$. The forrelation between $f$ and $g$ is defined as
\[
\Phi(f,g)=2^{-3n/2} \sum_{x,y \in \{0,1\}^n}\; f(x) g(y) (-1)^{x\cdot y}.
\]
The forrelation quantifies the correlations between $f$ and the binary Fourier transform of $g$. Approximating $\Phi(f,g)$ with a small additive error is known as the forrelation problem.

\begin{forrproblem}
Given $f,g\, \colon \, \{0,1\}^n \to \{1,-1\}$ and $\epsilon>0$, compute an estimate $\mu$ such that $|\mu-\Phi(f,g)|\leq \epsilon$. 
\end{forrproblem}

A simple quantum
algorithm can approximate $\Phi(f,g)$ by querying oracles $U_f$ and $U_g$ implementing diagonal unitary operators $U_f|x\rangle = f(x)|x\rangle$ and $U_g|x\rangle=g(x)|x\rangle$, as can be
seen from the identity 
\[
\Phi(f,g)=\langle 0^n| H^{\otimes n} U_f H^{\otimes n} U_g H^{\otimes n} |0^n\rangle.
\]
Using this identity, one can express the forrelation as $\Phi(f,g)=1-2p$ where $p$ is the output probability of a quantum circuit that makes only $1$ query to an $(n+1)$-bit oracle that computes either $f$ or $g$ depending on the value of an ancilla bit. This output probability can be estimated to within additive error $\epsilon$ using amplitude estimation \cite{brassard2002quantum}. In this way we can solve the forrelation problem on a quantum computer using only $O(\epsilon^{-1})$ queries to the oracles $U_f$ and $U_g$.

While this algorithm approximates $\Phi(f,g)$ to within a small constant error using $O(1)$ quantum queries, Aaronson and Ambainis have shown that any classical algorithm achieving the same requires $\Omega(n^{-1} 2^{n/2})$ queries. 
It is also known that the output probability of \textit{any} $1$-query quantum algorithm can be approximated to a given constant error by a classical randomized algorithm using $O(2^{n/2})$ queries \cite{aaronson2018forrelation,aaab21}. In other words---among tasks that can be solved with a bounded error using $1$ quantum query---forrelation has an almost maximal randomized query complexity.

In fact, a very simple classical algorithm suggested by Aaronson in Ref.~\cite{aaronson2010bqp} (see Section~5 of that paper) is almost optimal in terms of query complexity. One first samples random $n$-bit strings $x_1,\ldots, x_L$ and $y_1,\ldots, y_L$ and then computes the sample forrelation 
\begin{equation}
Z=L^{-2}2^{n/2}\sum_{i,j=1}^{L} f(x_i)g(y_j) (-1)^{x_i \cdot y_j}
\label{eq:samplebased}
\end{equation}
It is easy to see that $\mathbb{E}[Z]=\Phi(f,g)$ and one can also show that using  $L= O(\epsilon^{-1}2^{n/2})$ samples suffices to ensure that $Z$ approximates the forrelation within an additive error $\epsilon$ with high probability. The total number of queries required by this algorithm is $2L$ which almost matches the lower bound on classical query complexity.  However, computing  Eq.~\eqref{eq:samplebased} seems to require summing $L^2$ terms, and so the runtime of this algorithm is quadratically worse than the query complexity. Our first contribution provides the following improvement:
\begin{restatable}{theorem}{forr}
There is a classical randomized algorithm that solves the forrelation problem with probability at least $99\%$. The algorithm has query complexity $O(\epsilon^{-1} 2^{n/2})$ and runtime $O(n \epsilon^{-1}  2^{n/2})$. 
\label{thm:fastfor}
\end{restatable}
To prove the theorem, we approximate $\Phi(f,g)$ using a variant of the sample-based forrelation Eq.~\eqref{eq:samplebased} in which the sets $\{x_1,\ldots, x_L\}\subseteq \mathbb{F}_2^{n}$ and $\{y_1,\ldots, y_L\}\subseteq \mathbb{F}_2^{n}$ are chosen to be random affine subspaces with $L=2^\ell$ for suitably chosen $\ell\approxeq n/2$. To establish the claimed runtime bound we show that this estimator can equivalently be expressed as an amplitude of a quantum circuit acting on  $\ell$ qubits and computed using a runtime $\sim 2^\ell$ using sparse matrix-vector multiplication.

Ref.~\cite{aaronson2018forrelation} also introduces a more general $k$-fold forrelation of Boolean functions  $f_1,f_2,\ldots, f_k:\{0,1\}^n\rightarrow \{-1,1\}$ defined as 
\begin{equation}
\Phi(f_1, f_2,\ldots,f_k)=\langle 0^n|H^{\otimes n} U_{f_1}H^{\otimes  n}U_{f_2}H^{\otimes n}\ldots U_{f_k}H^{\otimes n}|0^n\rangle.
\label{eq:kfold}
\end{equation}
As in the case $k=2$ considered above, the $k$-fold forrelation can be expressed as $1-2p$ where $p$ is the output probability of a quantum circuit that makes $\left\lceil k/2\right\rceil$ queries to the oracles. This output probability, and the forrelation, can be additively approximated by a quantum computation that makes $O(k\epsilon^{-1})$ queries. Using the well-known connection between quantum query algorithms and multilinear polynomials over $\mathbb{F}_2$ \cite{beals2001quantum}, Aaronson and Ambainis claimed that any amplitude of a $k$-query quantum circuit---and consequently, the $k$-fold forrelation---can be additively approximated by a classical randomized algorithm which uses $O(2^k \epsilon^{-2+\frac{2}{k}} 2^{n(1-\frac{1}{k})})$ queries. It was recently discovered, however, that the proof of this statement given in Ref.  \cite{aaronson2018forrelation} contained a critical error as discussed in the recent note \cite{aaab21} and in the blog post \cite{aaronsonblogpost}.\footnote{The authors of Ref. \cite{aaab21} provided a new proof of the upper bound for the special case $k=1$.}

In the same work, Aaronson and Ambainis conjectured that a matching lower bound holds for each even $k=O(1)$ and constant error $\epsilon$. That is, their conjecture asserts that the $k$-fold forrelation should have maximal classical query complexity among quantities that can be expressed as the output probability of a $k/2$ query quantum algorithm. This conjecture was recently established in two breakthrough works Refs.~\cite{bansal2020kforrelation,sherstov2020optimal}. These works aimed to establish that indeed $k$-fold forrelation \cite{bansal2020kforrelation} and variants thereof \cite{sherstov2020optimal} achieve an almost maximal separation in quantum query complexity. For example, Bansal and Sinha show that for every positive integer $k$, approximating the $k$-fold forrelation to within an error $\epsilon$ satisfying $\epsilon=2^{-\Theta(k)}$ requires $\Omega(2^{n(1-\frac{1}{k})}/\mathrm{poly}(n,k))$ oracle queries using a classical computer \cite{bansal2020kforrelation}. 

Here we provide a simple classical algorithm that provides a matching upper bound. In other words we prove the result claimed by Aaronson and Ambainis in Ref.~\cite{aaronson2018forrelation}.\footnote{In fact, our upper bound achieves an exponential improvement in terms of the scaling with $k$ and improves the dependence on $\epsilon$, as compared with the one claimed in Ref.~\cite{aaronson2018forrelation}.}

\begin{theorem}[Informal]
Any amplitude of a quantum circuit that makes $k$ quantum queries to an $n$-bit oracle and may include non-unitary gates, can be approximated to within additive error $\epsilon$ by a classical randomized algorithm that makes $O(k \epsilon^{-2/k} 2^{n(1-\frac{1}{k})})$ oracle queries.
\label{thm:kamplitude}
\end{theorem}
Since the acceptance probability of a $t$-query quantum algorithm can be expressed as such an amplitude with $k=2t$ we directly obtain the following.
\begin{theorem}[\textbf{Classical simulation of quantum query algorithms}]
\label{thm:informalquery}
Let $p$ be the acceptance probability of a quantum query algorithm that makes $t$ quantum queries to an $n$-bit oracle. There is a classical randomized algorithm that outputs an estimate $\tilde{p}$ satisfying $|p-\tilde{p}|\leq \epsilon$ and uses $O(t\epsilon^{-1/t}2^{n(1-1/2t)})$ oracle queries.
\end{theorem}
We may also directly apply Theorem \ref{thm:kamplitude} to the forrelation problem.
\begin{theorem}
 There exists a classical randomized algorithm which approximates $\Phi(f_1, f_2,\ldots,f_k)$ to within additive error $\epsilon$ with probability at least $99\%$, and uses  $O(k \epsilon^{-2/k} 2^{n(1-\frac{1}{k})})$ oracle queries.
\label{thm:tquery}
\end{theorem}

Note that if we choose $\epsilon=2^{-\Theta(k)}$ as prescribed in Refs. \cite{bansal2020kforrelation,sherstov2020optimal}, then for each $k$ our randomized simulation matches their lower bound up to factors polynomial in $n$ and $k$. In this way our Theorem \ref{thm:informalquery} together with the results of Refs. \cite{bansal2020kforrelation,sherstov2020optimal} completes the characterization of the maximal separation in query complexity between quantum and classical computers.

Next we consider a graph-based forrelation problem 
with explicitly specified functions $f$ and $g$ that possess certain local features.
Suppose $G=(V,E)$ is a fixed $n$-vertex graph.
We place a binary variable $x_j$ at each vertex $j\in V$.
We say that 
a function  $f\, \colon \{0,1\}^n\to \mathbb S^1$ is two-local on $G$
 if
it is a product of terms describing interactions between nearest-neighbor variables, as well as terms that depend on a single variable, that is,
\[
f(x)=\prod_{\{u,v\}\in E} f_{uv} (x_u,x_v)  \prod_{u\in V} f_{u}(x_u)
\]
for some functions $f_{uv} \colon \{0,1\}^2\rightarrow \mathbb S^1$ and 
$f_u \colon \{0,1\} \rightarrow \mathbb S^1$. 
Here $\mathbb S^1$ is the unit circle in the complex plane. 
We shall use the term two-local 
without specifying the underlying graph whenever it is clear from the context.
A function $f(x)$ is said to be $1$-local if it is a product of terms that depend on a single variable. 

Suppose  $O_1,\ldots,O_n$ are single-qubit operators normalized such that $\|O_j\|\le 1$ for all $j$. 
A graph-based forrelation associated with
the graph $G$ and two-local functions $f,g$ is defined as a complex number
\be
\label{graph_based}
\Phi = \la 0^n|H^{\otimes n} U_g (O_1\otimes O_2 \otimes \cdots \otimes O_n)U_f H^{\otimes n}|0^n\ra.
\ee
As before, $U_f$ and $U_g$ are diagonal $n$-qubit unitary operators
such that $U_f|x\ra=f(x)|x\ra$ and $U_g|x\ra=g(x)|x\ra$ for all $x$.
Specializing Eq.~(\ref{graph_based}) to $O_1=\ldots=O_n=H$
and functions $f,g$ taking values $\pm 1$
one obtains the forrelation quantity $\Phi(f,g)$ defined above. 
The  consideration of complex-valued functions $f,g$
and more general operators $O_j$ is motivated by 
applications of the graph-based forrelation  that we discuss below. 
Note that the operators $U_f, U_g$ can be implemented efficiently by
quantum circuits for any two-local functions $f,g$. Furthermore, 
the operators $O_j$ can be extended
to two-qubit unitary operators using a block encoding, see e.g.\ Lemma~5 of~\cite{bravyi2019classical}.
Thus the graph-based forrelation 
can be efficiently approximated to a given additive error using a quantum computer.   A natural question is whether or not it is classically hard to approximate a graph-based forrelation to within a given additive error. Unfortunately the exponential lower bound on the query complexity established in the black-box setting does not rule out an efficient classical algorithm that can ``look inside" the black box.

Our interest in graph-based forrelation is partly motivated by its connection with algorithms for near-term quantum computers and restricted models of quantum computation. 
For example, one can show that the expected value of any tensor product operator on the output state of an $n$-qubit Clifford circuit can be expressed as a graph-based forrelation\footnote{This claim follows from the well-known fact~\cite{van2004graphical}
that the output state of any Clifford circuit acting on $n$-qubits is locally equivalent
to a graph state $\prod_{(u,v)\in E} \mathsf{CZ}_{u,v} 
H^{\otimes n}|0^n\ra$ for a suitable $n$-vertex graph $G=(V,E)$.}
for a suitable $n$-vertex graph $G$ and two-local functions $f=g$
(in particular, an output probability of a measurement based quantum computation~\cite{raussendorf2001one} can be expressed as a graph-based forrelation).
Similarly, an amplitude of any IQP circuit \cite{bremner2011classical} composed of one- and two-qubit gates can be expressed as a graph-based forrelation with $O_j=H$ for all $j$. In Appendix \ref{sec:relative_error} we provide more details and ascertain as a simple consequence that approximating a graph-based forrelation with a given \textit{relative} error is $\#$P-hard, even when restricted to planar graphs. 

A graph-based forrelation appears naturally in the study
of Quantum Approximate Optimization Algorithm (QAOA)~\cite{farhi2014quantum}.
Namely,
consider $n$ qubits located at vertices of the graph $G$
and a cost function Hamiltonian
$C=\sum_{\{u,v\}\in E} J_{u,v} Z_u Z_v$, where $J_{p,q}$
are arbitrary coefficients.
Level-$k$ QAOA maximizes the expected energy 
$\la \psi|C|\psi\ra$ over variational states 
\[
|\psi\ra = e^{-i \beta_k B}e^{-i \gamma_k C}
\cdots e^{-i \beta_1 B}e^{-i \gamma_1 C}H^{\otimes n}|0^n\ra,
\]
where $B=X_1+\ldots+X_n$ and $\beta,\gamma\in {\mathbb R}^k$ are variational parameters. In Section~\ref{sec:QAOA} 
we show that the variational energy 
$\la \psi |C|\psi\ra$ 
of level-$2$ QAOA states can be expressed as a linear combination of a few graph-based forrelations with
suitable single-qubit operators $O_j$
and $2$-local functions $f=g$ simply related to $C$.
Thus a polynomial-time classical algorithm for graph-based forrelation can be used to efficiently estimate the variational energy of level-$2$ QAOA states. Prior to our work such an algorithm was only known for the case of constant-degree graphs, or for level-$1$ QAOA states on general graphs~\cite{wang2018quantum}.

Our main result is a classical algorithm for approximating graph-based forrelation with a given additive error. The algorithm is efficient whenever the  vertices of $G$ can be partitioned into two disjoint subsets $A$ and $B$ such that the subgraphs $G_A$ and $G_B$ induced by $A$ and $B$ have a small treewidth. Let $w$ be the maximum width of these tree decompositions. Our algorithm is described in the following theorem.
\begin{theorem}
\label{thm:graph_based}
There is a classical randomized algorithm that outputs an estimate $\mu$ satisfying
$|\mu-\Phi|\le \epsilon$ with probability at least $99\%$. The algorithm
has runtime $O(n 4^w \epsilon^{-2})$.
\end{theorem} 

For example, suppose $G$ is a bipartite graph.
Then one can partition the vertices  of $G$ as  $V=A\cup B$  such that 
every edge of $G$ connects a vertex in $A$ and a vertex in $B$. 
Note that the induced subgraphs $G_A$ and $G_B$ are trivial in this case
(each subgraph has an empty set of edges). Such graphs
have treewidth $w=1$. Thus our algorithm has a runtime $O(n \epsilon^{-2})$. A more sophisticated application is obtained by considering planar graphs $G$. As discussed in Section~\ref{sec:treewidth}, there is an efficiently computable partition of vertices such that $w=2$ for any planar graph $G$ and thus our algorithm has runtime $O(n \epsilon^{-2})$.
In fact, the same runtime scaling holds for any family of graphs 
with a fixed forbidden minor~\cite{devos2004excluding}.

Thus Theorem~\ref{thm:graph_based} provides an efficient classical method for calculating variational energies of the level-2 QAOA for any bipartite or any planar graph.
As a demonstration, we report a classical simulation\footnote{Our simulation will actually use a simpler version of the algorithm in Theorem~\ref{graph_based} whose runtime suffers from an additional factor of $n$.} of the recently
proposed Recursive QAOA~\cite{bravyi2019obstacles} on planar graphs
with the level-2 variational states, see Section~\ref{sec:QAOA} for details. For example, simulating the Recursive QAOA with $n=225$ qubits on a laptop computer
took less than one day. To accomplish this simulation we had
to solve about $300,000$ instances of the graph-based forrelation problem.

In order to convey the main ideas, let us now describe how the
algorithm claimed in Theorem~\ref{thm:graph_based} works in the simplest case where  $G$ is a bipartite graph and all operators $O_j$
are unitary. Partition the vertices  of $G$ as  $V=A\cup B$  such that 
the induced subgraphs $G_A$ and $G_B$ have no edges. 
Let $O_A$ and $O_B$ be the tensor products of operators $O_j$ over all qubits $j\in A$ and $j\in B$ respectively. Then
$O_1\otimes O_2 \otimes \cdots \otimes O_n=O_A \otimes O_B$.
Define normalized $n$-qubit states
\[
|\alpha\ra = (O_A\otimes I_B) U_f H^{\otimes n}|0^n\ra
\quad \mbox{and} \quad
|\beta\ra = (I_A\otimes O_B^\dag) U_g^\dag H^{\otimes n}|0^n\ra
\]
such that $\Phi = \la \beta|\alpha\ra$. 
We claim that the states $|\alpha\ra$ and $|\beta\ra$ are computationally tractable~\cite{nest2009simulating}
in the sense that their amplitudes in the standard basis are easy to compute classically.
Indeed, consider some amplitude $\la x|\alpha\ra$. Write $x=x_A x_B$ 
where $x_A$ and $x_B$ are the restrictions of $x$ onto $A$ and $B$. Then 
\[
\la x|\alpha\ra=
\la x_A x_B|\alpha\ra = 2^{-n/2} \sum_{y_A} f(y_A x_B) \la x_A|O_A|y_A\ra,
\] 
where the sum ranges over $y_A\in  \{0,1\}^{|A|}$. Note that 
fixing the variables $x_B$ transforms  $f(y_Ax_B)$ into a $1$-local
function. Indeed, each two-local term in $f$ depends on some variable
in $A$  and some variable in $B$. Since the latter is fixed, $f$ becomes a product of terms
depending on a single variable each, that is, $f$ is $1$-local. Likewise,  for a fixed $x_A$ the matrix element
$\la x_A|O_A|y_A\ra$ is a $1$-local function of $y_A$. Thus 
$\la x|\alpha\ra =\sum_{y_A} h(y_A)$, where $h(y_A)$ is a $1$-local function.
Such a sum can be computed exactly in linear time. By the same
reasoning, any  amplitude $\la x|\beta\ra$ is easy to compute classically.   A slight generalization
of the above argument shows that  a classical computer can sample
the probability distribution $P(x)\equiv |\la x|\alpha\ra|^2$
in linear time. 
Now the forrelation $\Phi=\la \beta|\alpha\ra$
can be approximated by a Monte Carlo method due to Van den Nest~\cite{nest2009simulating}
which is based on the identity
\[
\Phi = \la \beta|\alpha\ra = \sum_{x\in \{0,1\}^n} P(x) R(x), \qquad R(x)\equiv \frac{\la \beta|x\ra}{\la \alpha|x\ra}.
\]
Therefore it suffices to approximate the mean value of a random variable $R(x)$ with $x$ sampled from $P(x)$.
Fix an integer $S\gg 1$ and let  $x^1,\ldots,x^S\in \{0,1\}^n$ be independent samples from $P(x)$.
Define the desired estimate of $\Phi$ as an empirical mean value of $R(x)$ over the observed samples,
$\mu=(R(x^1)+R(x^2)+ \ldots + R(x^S))/S$.
The random variable $R(x)$ has the mean value $\Phi$ and its variance is bounded as
\be
\mathrm{Var}(R)\le \sum_{x\in \{0,1\}^n} P(x) |R(x)|^2 = \sum_{x\in \{0,1\}^n} |\la x|\beta\ra|^2 = 1.
\ee
By Chebyshev's inequality, $|\mu - \Phi|\le \epsilon$
with probability at least $0.99$  if we choose  $S=100 \epsilon^{-2}$.

The above algorithm is efficient whenever efficient subroutines are available
for computing  amplitudes of the states $|\alpha\ra$, $|\beta\ra$
and for sampling the probability distribution $|\la x|\alpha\ra|^2$. 
In Section~\ref{sec:graph_based} we prove Theorem \ref{thm:graph_based} by constructing the desired subroutines
for more general  (non-bipartite) graphs that admit a partition into two subgraphs with a small treewidth. We leave as an open question complexity of the graph-based forrelation problem on general graphs.

To sample from $|\la x|\alpha\ra|^2$ in linear time for (non-bipartite) planar graphs, we will prove that there is an efficient subroutine for the following general problem on tensor networks that may be of independent interest:
\begin{problem}
	\label{prob:forrsample}
	Given an initial state $\chi = \chi_{1} \otimes \cdots \otimes \chi_{n}$ of $n$ $d$-dimensional qudits, a collection of $m$ diagonal gates $U_1, \ldots, U_m$ (not necessarily unitary, possibly multi-qudit), and single-qudit operators
	$O_1,\ldots,O_n$, sample a string $x\in \{1,2,\ldots,d\}^n$ from the distribution\footnote{We note that there are choices for the operators $\{ O_a \}_{a\in[n]}$ for which $P$ cannot be normalized to a distribution (i.e., is always 0).  In such cases we simply report that there are no valid outputs $x$.} $P(x) \propto \bra{x} O U \chi U^{\dag} O^{\dag} \ket{x}$, where
	\begin{align*}
		U &= \prod_{j=1}^{m} U_j, & O &= \bigotimes_{i=1}^{n} O_i.
	\end{align*} 
\end{problem}
A notable special case of this problem is sampling a Markov random field
which is a ubiquitous problem in statistics~\cite{clifford1990markov}.
Indeed, any Markov random field $P(x)$ with $n$ discrete variables $x_i\in \{1,2,\ldots,d\}$ can be written as
$P(x)\propto \la x|U\chi U^\dag|x\ra$,
where $U_j$ are nonunitary diagonal multi-qudit gates with real non-negative entries on the diagonal and $\chi$ is the maximally mixed state. 

Connectivity of the gates $U_j$ can be described by a graph $G$ with $n$ vertices
such that each vertex represents a qudit
and a pair of vertices is connected by an edge if the corresponding pair of qudits participates in some gate $U_j$.
Let $w$ be the threewidth of $G$. 
In Appendix~\ref{sec:linear_time_sampling}
we give an algorithm that solves Problem~\ref{prob:forrsample}
in time $O(nd^{2w} + m d^{w})$.

The remainder of the paper is organized as follows. Our results concerning oracle-based forrelation, i.e., the proofs of Theorems \ref{thm:fastfor} and \ref{thm:kamplitude}, are provided in Sections \ref{sec:oraclebased} and \ref{sec:kfold_forrelation}, respectively. 
Section~\ref{sec:graph_based}  
describes our classical algorithm for graph-based forrelation.
We first analyze a simplified version of the algorithm 
which has runtime scaling quadratically with the number of qubits,
see Section~\ref{sec:graph_based}. To achieve the runtime stated
in Theorem~\ref{thm:graph_based}, we augment the algorithm
of Section~\ref{sec:graph_based}
by a linear-time subroutine solving Problem~\ref{prob:forrsample}.
This subroutine is described in Appendix~\ref{sec:linear_time_sampling} and
illustrated by an example in Appendix~\ref{app:example}.
We apply the algorithm for graph-based forrelation  to estimate the variational energy of level-2 QAOA and report a numerical simulation
RQAOA in Section~\ref{sec:QAOA}. 
Further details concerning
this simulation can be found in Appendices~\ref{app:eliminate_beta_2},\ref{all:brute_force},\ref{app:tree}.
We conclude by discussing implications of our work
and some open questions concerning graph-based forrelation
in Section~\ref{sec:conclusions}.
Finally, the problem of approximating the graph-based forrelation
with a small relative error is discussed in 
Appendix~\ref{sec:relative_error}.

\section{Oracle-based forrelation}
\label{sec:oraclebased}

In this section we consider the forrelation problem in which functions $f,g$ are accessed via oracle queries.  We begin by proving Theorem \ref{thm:fastfor}, restated below. 

\forr*

\begin{proof}
Let $\epsilon>0$ be given. If $\epsilon \leq 2^{-n/2}$ then we use a very naive algorithm to compute $\Phi(f,g)$ exactly using the desired total runtime $O(n2^n)$ and $O(2^n)$ queries. Indeed, in this case it suffices to exactly compute a length $2^n$ vector $H^{\otimes n}U_g H^{\otimes n} U_fH^{\otimes n}|0^n\rangle$ using matrix vector multiplication. 
Applying each single-qubit Hadamard gate using sparse matrix-vector multiplication requires a runtime $O(2^n)$ \footnote{To see this, note that for any $n$-qubit state $|v\rangle$, each entry of $\left(H\otimes I_{n-1}\right) |v\rangle$ is a sum of two entries of $|v\rangle$.} while applying each gate $U_f$ or $U_g$ uses $O(2^n)$ queries. The forrelation $\Phi(f,g)$ is then obtained as the $0^n$ entry of this vector.

In the following we consider the case 
\begin{equation}
\epsilon\geq 2^{-n/2}.
\label{eq:eps}
\end{equation}

Let $\mathcal{A}_{n,k}$ be the set of all affine spaces $S\subseteq \mathbb{F}_2^{n}$ with $|S|=2^k$. For each $S\in \mathcal{A}_{n,k}$ define a normalized $n$-qubit stabilizer state
\[
|S\rangle=2^{-k/2} \sum_{x\in S}|x\rangle.
\]
For any pair $S,T\in \mathcal{A}_{n,k}$, define
\begin{equation}
\mu(S,T)=2^{n-k} \langle S|U_f H^{\otimes n} U_g|T\rangle.
\label{eq:must}
\end{equation}
\begin{claim}
Suppose $S,T\in \mathcal{A}_{n,k}$ are drawn uniformly at random. Then
\[
\mathbb{E} (\mu(S,T))=\Phi(f,g) \quad \text{ and } \quad \mathrm{Var}(\mu(S,T))\leq 2^{n-2k}+\frac{2}{2^k}.
\]
\label{claim:var}
\end{claim}
\begin{proof}
For any $x\in \mathbb{F}_2^n$, let $\chi_S(x)=1$ if $x\in S$ and $\chi_S(x)=0$ otherwise. Then
\begin{equation}
\mathbb{E}_S(\chi_S(x))=\frac{|S|}{2^n}=2^{k-n}.
\label{eq:mean}
\end{equation}
Likewise, for any pair of binary strings $x,y\in \mathbb{F}_2^n$ such that $x\neq y$ we have
\begin{equation}
\sigma_2\equiv \mathbb{E}_S(\chi_S(x)\chi_S(y))=\frac{|S|\left(|S|-1\right)}{2^n(2^n-1)}=\frac{4^{k-n}(1-2^{-k})}{1-2^{-n}}\leq 4^{k-n}.
\label{eq:var}
\end{equation}
Using Eq.~\eqref{eq:mean} we get 
\[
\mathbb{E}_S(|S\rangle)=2^{(k-n)/2} |+^n\rangle
\]
where $|+^n\rangle=H^{\otimes n} |0^n\rangle$, and consequently
\begin{equation}
\mathbb{E}_S\mathbb{E}_T \mu(S,T)=\Phi(f,g).
\label{eq:exp}
\end{equation}
Using Eq.~\eqref{eq:var} we get
\begin{equation}
\mathbb{E}_S(|S\rangle\langle S|)= a\cdot I +b\cdot |+^n\rangle\langle +^n| \qquad a\equiv 2^{-n}-2^{-k}\sigma_2 \qquad b\equiv 2^{n-k}\sigma_2.
\label{eq:braket}
\end{equation}
and 
\begin{align}
\mathbb{E}_S\mathbb{E}_T (\mu(S,T))^2 &=4^{n-k}\mathbb{E}_S\mathbb{E}_T \mathrm{Tr}\left(|S\rangle\langle S| U_f H^{\otimes n} U_g |T\rangle \langle T|U_g^{\dagger} H^{\otimes n} U_f^{\dagger} \right)\nonumber\\
&=2^{3n-2k} a^2+4^{n-k}b^2\Phi(f,g)^2+2\cdot 4^{n-k} ab\nonumber\\
&\leq 2^{n-2k}+\Phi(f,g)^2+\frac{2}{2^k}.
\label{eq:exp2}
\end{align}
where we used the upper bounds $a\leq 2^{-n}$ and $b\leq 2^{k-n}$. Combining Eqs.~(\ref{eq:exp}, \ref{eq:exp2}) gives
\[
\mathrm{Var}(\mu(S,T))= \mathbb{E}_S\mathbb{E}_T (\mu(S,T))^2-\left(\mathbb{E}_S\mathbb{E}_T \mu(S,T)\right)^2\leq 2^{n-2k}+\frac{2}{2^k}.
\]
\end{proof}

Now let us fix $k$ to be an integer satisfying
\begin{equation}
2^{18}\cdot 2^n\epsilon^{-2} \geq 2^{2k}\geq 2^{16}\cdot 2^n\epsilon^{-2}.
\label{eq:choosek}
\end{equation}
We shall choose our estimate $\mu\equiv\mu(S,T)$ where the affine subspaces $S,T\in \mathcal{A}_{n,k}$ are drawn uniformly at random. Note that the random selection of these affine spaces can be performed using a runtime $O(\mathrm{poly}(n))$, resulting in a parameterization
\begin{equation}
S=\{ Ax+b: x\in \{0,1\}^{k}\} \qquad T=\{ Cx+d: x\in \{0,1\}^{k}\} 
\label{eq:st}
\end{equation}
where $A$ and $C$ are $n\times k$ binary matrices and $b,d\in \{0,1\}^n$, and addition is performed over $\mathbb{F}_2$. 

Using Claim \ref{claim:var} and Eq.~\eqref{eq:choosek} we get
\[
\mathrm{Var}(\mu(S,T))\leq \frac{\epsilon^2}{2^{16}}+\frac{\epsilon}{2^{7}}2^{-n/2}\leq \epsilon^2 \left(\frac{1}{2^{16}}+\frac{1}{2^7}\right)\leq \frac{\epsilon^2}{100},
\]
where in the second inequality we used  Eq.~\eqref{eq:eps}. Applying Chebyshev's inequality we see that 
\[
\mathrm{Pr}\left(|\mu(S,T)-\Phi(f,g)|\geq \epsilon \right)\leq 1/100,
\]
as desired. Below we show that there is a classical algorithm which, given $S,T\in \mathcal{A}_{n,k}$ (as in Eq.~\eqref{eq:st}), computes $\mu(S,T)$ exactly using a runtime $O(k^2n+n 2^k)$ and $O(2^{k})$ queries.  Using Eq.~\eqref{eq:choosek} we see that the total runtime is $O(n2^{n/2}\epsilon^{-1})$ and the total number of queries is $O(2^{n/2}\epsilon^{-1})$.

Let us now describe how to compute $\mu(S,T)$ using the stated runtime and number of queries. Using Eqs.~(\ref{eq:st},\ref{eq:must}) we may write
\[
\mu(S,T)=2^{n/2-2k}\sum_{x,y\in \{0,1\}^{k}} f(Ax+b)g(Cy+d)(-1)^{y^T C^T Ax+x^{T}A^{T} d+b^{T}Cy+b^Td}
\]
To compute this sum on a classical computer we first initialize the following vector of length $2^{k}$ 
\[
|\psi_1\rangle=\sum_{x\in \{0,1\}^{k}}f(Ax+b)(-1)^{x^TA^Td} |x\rangle.
\]
We compute the entries of this vector one at a time using a Gray code ordering $x_1,x_2,\ldots, x_{2^k}$ of the set of binary strings $\{0,1\}^{k}$, so that $x_{i+1}$ differs from $x_i$ in only one bit. With this choice, we can compute $Ax_{i+1}+b$ from $Ax_{i}+b$ using $O(n)$ binary operations. We need only compute $A^Td$ once (using a runtime $O(nk)$) and then each phase factor $(-1)^{x^TA^Td}$ can also be computed using $O(k)$ runtime. From this we see that initializing $\psi_1$ can be performed using $O(nk+n2^{k})$ runtime and $O(2^{k})$ queries.

 In the next step we compute a vector 
\[
|\psi_2\rangle= \sum_{x\in \{0,1\}^{k}}f(Ax+b)(-1)^{x^TA^Td} |C^TA x\rangle.
\]
We first compute the matrix $C^TA$ which takes a runtime $O(k^2 n)$. We initialize $|\psi_2\rangle$ as a vector of length $2^k$ with every entry equal to zero. Then,  using our Gray code order we compute $C^TAx_1, C^TAx_2, \ldots, C^TAx_{2^k}$. Each time we compute a vector $z_i=A^TCx_i$ we set
\[
\langle z_i|\psi_2\rangle\leftarrow \langle z_i|\psi_2\rangle+\langle x_i|\psi_1\rangle.
\]
The total runtime for this step is $O(k^2n+k2^k)$ since $C^TA$ is a $k\times k$ matrix.

Viewing $\psi_2$ as a $k$-qubit state, we then apply a Hadamard gate on each qubit to obtain 
\[
|\psi_3\rangle=H^{\otimes k} |\psi_2\rangle= 2^{-k/2}\sum_{x,y\in \{0,1\}^{k}}f(Ax+b)(-1)^{x^TA^Td}(-1)^{y^TC^TAx} |y\rangle.
\]
Applying a single Hadamard gate to a $k$-qubit state stored in classical computer memory takes a runtime $O(2^k)$ using sparse matrix-vector multiplication. The total runtime to compute $\psi_3$ using a sequential application of the $k$ Hadamard gates is then upper bounded as $O(k2^k)$.

Finally, we compute a vector $\psi_4$ defined by
\[
\langle y|\psi_4\rangle=\langle y|\psi_3\rangle\cdot 2^{n/2-3k/2}(-1)^{b^Td}   g(Cy+d)(-1)^{b^TCy} \qquad y\in \{0,1\}^k.
\]
Performing this step using a Gray code ordering can be done using  $O(nk+n2^{k})$ total runtime and $O(2^{k})$ queries (similar to the procedure used to initialize $\psi_1$).  Here
\[
|\psi_4\rangle= 2^{n/2-2k}\sum_{x,y\in \{0,1\}^{k}}f(Ax+b)(-1)^{x^TA^Td}(-1)^{y^TC^TAx} g(Cy+d)(-1)^{b^TCy} (-1)^{b^Td} |y\rangle.
\]
Finally, $\mu(S,T)$ is obtained by summing all entries of the vector $\psi_4$ using a runtime $O(2^k)$.

In total, this algorithm uses $O(2^k)$ queries and $O(k^2n+n2^k)$ runtime, as claimed.
\end{proof}

\section{Quantum query complexity and \texorpdfstring{$k$-fold}{k-fold} Forrelation}
\label{sec:kfold_forrelation}
Now let us turn our attention to quantum query algorithms and the $k$-fold forrelation $\Phi(f_1,f_2,\ldots, f_k)$ defined in Eq.~\eqref{eq:kfold}. Here $f_1,f_2,\ldots, f_k:\{0,1\}^n\rightarrow \{-1,1\}^n$ are $n$-bit boolean functions. In this Section we prove Theorem \ref{thm:kamplitude}. 

Letting $m\geq n$ we shall describe a randomized classical algorithm which is capable of estimating an $m$-qubit amplitude of the form
\begin{equation}
q=\langle 0^m|M_1U_{f_1} M_2 U_{f_2}\ldots U_{f_k}M_{k+1}|0^m\rangle.
\label{eq:p2}
\end{equation}
where $M_1,\ldots, M_{k+1}$ are $m$-qubit operators such that $\|M_j\|\leq 1$ for all $1\leq j\leq k+1$, and $U_{f_j}|x\rangle|y\rangle=f_j(x)|x\rangle|y\rangle$ for each $x\in \{0,1\}^n, y\in \{0,1\}^{m-n}$.  Clearly the $k$-fold forrelation $\Phi(f_1,f_2,\ldots, f_k)$ can be expressed as in Eq.~\eqref{eq:p2} with $m=n$. Amplitudes of the form Eq.~\eqref{eq:p2} can also express output probabilities of quantum query algorithms. Indeed, suppose $p$ is the probability of measuring an output qubit in the state $|1\rangle$ at the output of a quantum algorithm that makes $t$ quantum queries to an oracle. Then $p$ can be expressed as as in Eq.~\eqref{eq:p2} with $k=2t$.

\begin{theorem}
There exists a classical randomized algorithm which, given $\epsilon>0$,  outputs an estimate $\tilde{q}$ such that
\[
|\tilde{q}-q|\leq \epsilon
\]
with probability at least $99\%$. The algorithm uses  $O(k \epsilon^{-2/k} 2^{n(1-\frac{1}{k})})$ oracle queries.
\label{thm:tforr}
\end{theorem}

\begin{proof}
In the following we shall write 
\[
\Pi(z)=|z\rangle\langle z|\otimes I \quad \qquad z\in \{0,1\}^n.
\]
where the identity acts on $m-n$ qubits. For any nonzero complex vector $\alpha$ of length $2^{m}$, define a probability distribution over $n$-bit strings
\begin{equation}
p_{\alpha}(z)=\|\alpha\|^{-2}\| \Pi(z)|\alpha\rangle \|^2 \quad \qquad z\in \{0,1\}^n.
\label{eq:p}
\end{equation}

Define $|\phi_0\rangle=|0^m\rangle$ and then for each $1\leq j\leq k$ construct a (possibly un-normalized) state $\phi_j$ as follows. First sample binary strings $z_1,z_2,\ldots, z_L$ independently from the distribution $p_{M^{\dagger}_{j}\phi_{j-1}}$ (see Eq.~\eqref{eq:p}). We will fix $L$ later.  Then define
\begin{equation}
|\phi_{j}\rangle=\frac{1}{L}\sum_{i=1}^{L} \frac{f_j(z_i)\Pi(z_i)M^{\dagger}_j|\phi_{j-1}\rangle}{p_{M^{\dagger}_{j}\phi_{j-1}}(z_i)}
\label{eq:phipsi}
\end{equation}
Note that
\begin{equation}
\mathbb{E}(|\phi_{j}\rangle)=U_{f_j}M^{\dagger}_{j}|\phi_{j-1}\rangle 
\label{eq:ephipsi}
\end{equation}
and we have the operator inequality
\begin{align}
 &\mathbb{E}(|\phi_{j}\rangle\langle \phi_j|)\\
 &=\left(\frac{L-1}{L}\right)U_{f_j}M^{\dagger}_{j}|\phi_{j-1}\rangle\langle \phi_{j-1}|M_{j} U_{f_j}+\frac{\|\phi_{j-1}\|^2}{L}\sum_{z\in \{0,1\}^n}\frac{\Pi(z)M_j^{\dagger}|\phi_{j-1}\rangle\langle \phi_{j-1}|M_j \Pi(z) }{\langle \phi_{j-1}|M_j\Pi(z)M_j^{\dagger}|\phi_{j-1}\rangle}\label{eq:opineq}\\
&\leq U_{f_j}M^{\dagger}_{j}|\phi_{j-1}\rangle\langle \phi_{j-1}|M_{j} U_{f_j} +\frac{\|\phi_{j-1}\|^2}{L}\cdot I.
\label{eq:phiphi}
\end{align}
Using Eq.~\eqref{eq:opineq} we see that 
\begin{equation}
\mathbb{E}(\|\phi_{j}\|^2)=\mathrm{Tr}\left(\mathbb{E}(|\phi_j\rangle\langle \phi_j|)\right)\leq  \left(1+\frac{2^{n}}{L}\right)\|\phi_{j-1}\|^{2}
\label{eq:normj}
\end{equation}
where to bound the first term we used the fact that $\|M_j\|\leq 1$.  For each $0\leq r\leq k$ define a complex-valued random variable 
\[
X_{r}= \langle \phi_{r}|M_{r+1}U_{f_{r+1}}M_{r+2}\ldots U_{f_{k}}M_{k+1}|0^m\rangle.
\]
Note that $X_0=q$ is the desired amplitude. Using Eq.~\eqref{eq:ephipsi} we see that 
\[
\mathbb{E}(X_k)=\mathbb{E}(X_{k-1})=\ldots =\mathbb{E}(X_0)=q.
\]
Note that $X_k$ can be computed using $L$ queries to each of the given oracles $f_1,f_2,\ldots, f_k$.
\begin{lemma}

\begin{equation}
\mathbb{E}(|X_k|^2)-|q|^2\leq 2^{n(k-1)}L^{-k}\left(1+\frac{L}{2^n}\right)^{k}.
\label{eq:xtvar}
\end{equation}
\label{lem:xtvar}
\end{lemma}
\begin{proof}
Let us write $\mathbb{E}_j$ for the expectation over all random variables used in generating the states $\phi_1, \phi_2,\ldots,  \phi_j$. Then, for all $1\leq r\leq k$ we have
\begin{align*}
\mathbb{E}(|X_r|^2)=\mathbb{E}_{r} (|X_r|^2)&=\mathbb{E}_{r-1}\mathbb{E}_{\phi_{r}}\left( |\langle \phi_{r}|M_{r+1}U_{f_{r+1}}M_{r+2}\ldots U_{f_{k}}M_{k+1}|0^m\rangle|^2\right)\\
&\leq \mathbb{E}_{r-1}\left(|X_{r-1}|^2\right)+\frac{1}{L}\mathbb{E}_{r-1}\left( \|\phi_{r-1}\|^2\right)\\
\end{align*}
where in the last line we used Eq.~\eqref{eq:phiphi}. The second term above can be computed using a repeated application of Eq.~\eqref{eq:normj}, which gives, for each $1\leq r\leq k$,
\[
\mathbb{E}(|X_r|^2)-\mathbb{E}\left(|X_{r-1}|^2\right)\leq \frac{1}{L}\left(1+\frac{2^{n}}{L}\right)^{r-1}
\]
 Since $X_0=q$ we get
\[
\mathbb{E}(|X_k|^2)-|q|^2=\sum_{r=1}^{k} \left(\mathbb{E}(|X_r|^2)- \mathbb{E}\left(|X_{r-1}|^2\right)\right)\leq 2^{-n}\left(\left(1+\frac{2^n}{L}\right)^{k}-1\right)
\]
Discarding the subtracted term in the parentheses, we arrive at Eq.~\eqref{eq:xtvar}.
\end{proof}

Now let us fix $B$ as follows
\begin{equation}
 B=2k\cdot \left\lceil2^{n(1-1/k)}(\epsilon^{2}/{400})^{-1/k}\right\rceil
 \label{eq:Bk}
\end{equation}
This will be our query budget. There are two cases to consider. If $B\geq k\cdot 2^n$ then we can simply query each oracle at all $2^n$ possible binary strings and output an exact computation of $q$. If instead $B\leq k\cdot 2^n$ then we compute the estimator $\tilde{q}=X_k$ with the choice $L=B/k$. In this case we have $L\leq 2^n$ and therefore Eq.~\eqref{eq:xtvar} gives
\[
\mathbb{E}(|X_k|^2)-|q|^2\leq 2^{n(k-1)}L^{-k}2^{k}\leq \epsilon^2/400
\]
where in the second inequality we lower bounded $L=B/k$ using Eq.~\eqref{eq:Bk}. Therefore
\[
\mathrm{Var}(\mathrm{Re}(\tilde{q}))+\mathrm{Var}(\mathrm{Im}(\tilde{q}))=\mathbb{E}(|X_k|^2)-|q|^2\leq \frac{\epsilon^2}{400},
\]
By Chebyshev's inequality we see that $|\mathrm{Re}(\tilde{q})-\mathrm{Re}(q)|\leq \epsilon/\sqrt{2}$ with probability at least $0.995$. The same arguments show that $|\mathrm{Im}(\tilde{q})-\mathrm{Im}(q)|\leq \epsilon/\sqrt{2}$ with probability at least $0.995$. By a union bound, we have $|q-\tilde{q}|\leq \epsilon$ with probability at least $0.99$.
\end{proof}

\section{Graph-based forrelation}
\label{sec:graph_based}

The algorithm for estimating a graph-based forrelation $\Phi$ described in the Introduction has a polynomial runtime whenever efficient subroutines are available
for computing  amplitudes of the states 
$|\alpha\ra$ and $|\beta\ra$ such that 
$\Phi=\la \beta|\alpha\ra$
and for sampling the probability distribution $|\la x|\alpha\ra|^2$.

In this section we will construct such subroutines for more general (non-bipartite) graphs that admit a partition into two subgraphs with a small treewidth.  We will actually give two separate algorithms for sampling the probability distribution $|\la x|\alpha\ra|^2$.  We give the first such algorithm in this section because it is conceptually simpler, and contained within the same framework as the algorithm for computing amplitudes.  However, it will suffer from a quadratically worse dependence on $n$.  The second algorithm, which solves Problem~\ref{prob:forrsample}, will be described in full in Appendix~\ref{sec:linear_time_sampling} and will lead to the runtime given in Theorem~\ref{thm:graph_based}.

Section~\ref{sec:treewidth} provides relevant background concerning tree decompositions.
Section~\ref{sec:two_local} describes an efficient subroutine for computing a
sum $\sum_x h(x)$, where $h(x)$ is a two-local function on a graph with a small treewidth.
Section~\ref{sec:graph_based_proof} combines these ingredients and completes the proof
of Theorem~\ref{thm:graph_based}.

\subsection{Low treewidth partitions of planar graphs}
\label{sec:treewidth}

\begin{dfn}[Tree decomposition]
Let $G=(V,E)$ be a graph. A tree decomposition of $G$ consists of a tree $T=(W,F)$, along with a \textit{bag} $B_i\subseteq V$ for each node $i\in W$, such that the following properties hold:
\begin{itemize}
\item For each $v\in V$, there exists a node $i\in W$ such that $v\in B_i$.
\item For each $\{u,v\}\in E$, there exists a node $i\in W$ such that $u\in B_i$ and $v\in B_i$.
\item For each $v\in V$, the subgraph $T_v$ of $T$ induced by the set of nodes $i\in W$ such that $v\in B_i$ is a tree.
\end{itemize}
\end{dfn}
We shall use the shorthand $T,\mathcal{B}$ to denote a tree decomposition, where $\mathcal{B}=\{B_i:i\in W\}$ is the set of bags. The \textit{width} $w$ of a tree decomposition is defined as 
\[
w= \max_{i\in W} |B_i| -1.
\]
The \textit{treewidth} of a graph $G$ is defined as the minimal width of any tree decomposition of it.

\begin{dfn}
A graph $G$ is said to be outerplanar if it admits a planar embedding in which all vertices are incident to the outer face.
\end{dfn}
Boedlaender has shown that outerplanar graphs have tree decompositions of width at most $2$ that can be computed in linear time.
\begin{theorem}[Bodlaender \cite{bodlaender88, bodlaender1996linear}]
If $G$ is an outerplanar graph then its treewidth is $\leq 2$ and a tree decomposition of width $\leq 2$ can be computed in linear time.
\end{theorem}
We refer the reader to Ref. \cite{bodlaender88} which provides a straightforward inductive proof that the treewidth is at most $2$, and Ref .\cite{bodlaender1996linear} which gives a linear time algorithm to compute such a tree decomposition. In this paper it will suffice for us to use the slower quadratic time Algorithm \ref{alg:outerplanar} (implicit in Ref. \cite{bodlaender88}).

The following theorem states that the vertices of any planar graph can be partitioned into two subsets, each of which induce a graph of low treewidth. For completeness, we include the simple proof below, following Ref. \cite{chartrand1971graphs}. An example is depicted in Fig.~\ref{fig:partitionexample}.
\begin{theorem}[Chartrand, Geller, Hedetniemi 1971 \cite{chartrand1971graphs}]
If $G=(V,E)$ is a planar graph then there exists a partition $V=A\cup B$ such that the induced subgraphs $G_A$ and $G_B$ are outerplanar graphs. Moreover, the partition can be computed in linear time from a planar embedding of $G$.
\label{thm:cgh}
\end{theorem}
\begin{proof}
We begin by defining a sequence of planar graphs $G_k=(V_k, E_k)$ for $0\leq k\leq p$, where $p\leq |V|$ depends on $G$.  Let $V_0=V, E_0=E$, and $G_0=G$. For each $k$ let us write $V^{\mathrm{out}}_k \subseteq V_k$ for the vertices incident to the outer face of $G_k$. If $V_{k}\neq V_{k}^{\mathrm{out}}$ we define
\[
V_{k+1}=V_{k}\setminus V^{\mathrm{out}}_{k}
\]
and let $G_{k+1}$ be the subgraph of $G_k$ induced by $V_{k+1}$. The process terminates at some graph $G_p$ for which $V_{p}=V_{p}^{\mathrm{out}}$ (an outerplanar graph).  Note that $G_{k+1}$ has fewer vertices than $G_k$, so the process stops after $p\leq |V|$ steps.

The vertex set of $G$ can be partitioned as $V=A\cup B$ where
\[
A=V_0^{\mathrm{out}}\cup V_2^{\mathrm{out}}\cup V_4^{\mathrm{out}}\ldots \qquad \qquad B=V_1^{\mathrm{out}}\cup V_3^{\mathrm{out}}\cup V_5^{\mathrm{out}}\ldots.
\]
Finally, we show that the induced subgraphs $G_A$ and $G_B$ are outerplanar. For each $j=0,1,\ldots,p$ let $H_j$ be the induced subgraph of $G$ on vertices $V_j^{\mathrm{out}}$. By construction, $G_A$ consists of disconnected copies of $H_0,H_2,H_4,\ldots$ and $G_B$ consists of disconnected copies of $H_1,H_3,H_5,\ldots$. To complete the proof we show that each graph $H_j$ is outerplanar. Observe that $H_j$ is the induced subgraph of the graph $G_j$ on the vertices incident to its outer face.  Therefore all vertices of $H_j$ are on its outer face and we are done.
\end{proof}
\def\factor{1.155}
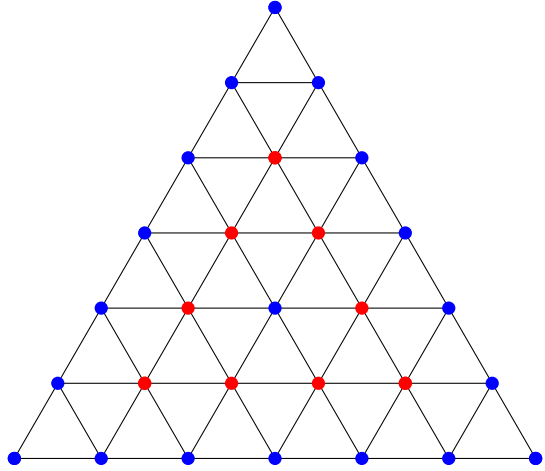
\begin{figure}
\centering
\begin{tikzpicture}[auto]
\pgftransformyscale{-1}
\pgftransformxscale{1.155}
\pgftransformxslant{-0.5}
\tikzstyle{vertexred}=[circle,fill=red,minimum size=5pt, inner sep=0pt]
\tikzstyle{vertexblue}=[circle,fill=blue,minimum size=5pt, inner sep=0pt]
\foreach \j in {1,...,7}
{
	\foreach \i in {1,...,\j}
	{
	\ifnum \i>1
		\draw (\i-1,\j)--(\i,\j);
	\fi
	\ifnum \j>1
		\ifnum \i<\j
		\draw (\i,\j-1)--(\i,\j);
		\fi
	 	\ifnum \i >1 
		\draw (\i,\j)--(\i-1,\j-1);
		\fi
	\fi
	}
}
\foreach \j in {1,...,7}
{
		\node[vertexblue] at (1,\j){};
		\node[vertexblue] at (\j,\j){};
		\node[vertexblue] at (\j,7){};
}

		\node[vertexblue] at (3,5){};
\foreach \j in {3,...,6}
{
		\node[vertexred] at (2,\j){};
		\node[vertexred] at (\j-1,\j){};
		\node[vertexred] at (\j-1,6){};
}
\end{tikzpicture}
\caption{Illustration of Theorem \ref{thm:cgh}. A planar graph and a partition of its vertices into subsets $A$ and $B$ shown in red and blue. The induced subgraphs $G_A$ and $G_B$ are outerplanar graphs.\label{fig:partitionexample}}
\end{figure}

\subsection{Summation of two-local functions}
\label{sec:two_local}

Suppose $G=(V,E)$ is a graph with $m=|V|$ vertices and $\Gamma$ is a finite set. We shall say that $h:\Gamma^m\rightarrow \mathbb C$ is a two-local function on $G$ if it can be written as
\begin{equation}
h(x)=\prod_{\{u,v\}\in E} h_{uv} (x_u,x_v)  \prod_{u\in V} h_{u}(x_u)
\label{eq:twoloc}
\end{equation}
for some functions $h_{uv}: \Gamma^2\rightarrow \mathbb C$ and $h_u:\Gamma \rightarrow \mathbb C$. Here $x=x_1x_2\ldots x_m$ with $x_j\in \Gamma$ for $1\leq j\leq m$. Note that if $G$ has no isolated vertices then the vertex terms $\{h_u\}$ can be folded into the edge terms $\{h_{uv}\}$ and need not explicitly appear in Eq.~\eqref{eq:twoloc}.

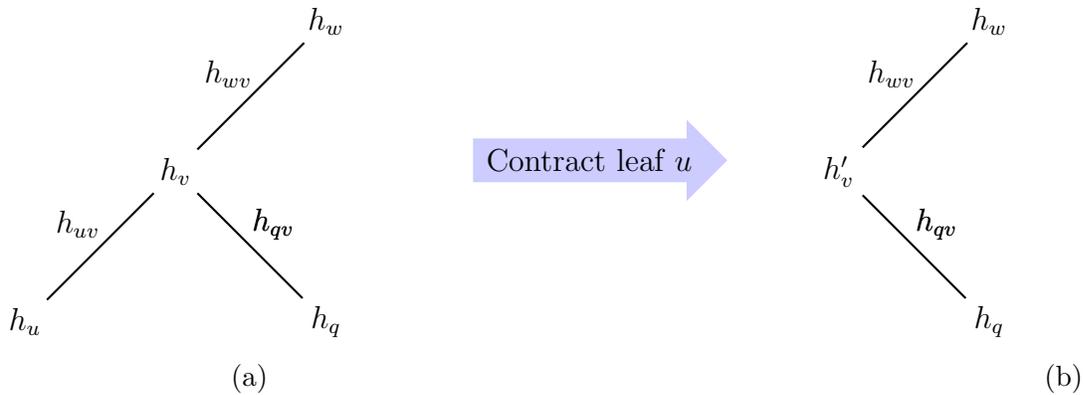
\begin{figure}

\begin{subfigure}[t]{0.4\textwidth}
\begin{tikzpicture}
\node[ inner sep=3pt, minimum size=6pt](A) at (0,0){$h_u$};
\node[ inner sep=3pt, minimum size=6pt](B) at (2,2){$h_v$};
\node[ inner sep=3pt, minimum size=6pt](C) at (4,0){$h_{q}$};
\node[ inner sep=3pt, minimum size=6pt](E) at (4,4){$h_{w}$};

\node[ inner sep=3pt, minimum size=6pt] at (0.7,1.3){$h_{uv}$};
\node[ inner sep=3pt, minimum size=6pt] at (3.3,1.3){$h_{qv}$};
\node[ inner sep=3pt, minimum size=6pt] at (3.3,1.3){$h_{qv}$};

\node[ inner sep=3pt, minimum size=6pt] at (2.7,3.3){$h_{wv}$};
\draw[thick] (A)--(B);
\draw[thick] (B)--(C);
\draw[thick] (B)--(E);

\node[arrow1] at (7.5,2){Contract leaf $u$};
\end{tikzpicture}
\caption{}
\end{subfigure}
\qquad \qquad \qquad \qquad \qquad
\begin{subfigure}[t]{0.4\textwidth}
\begin{tikzpicture}
\node[ inner sep=3pt, minimum size=6pt](B) at (2,2){$h'_v$};
\node[ inner sep=3pt, minimum size=6pt](C) at (4,0){$h_{q}$};
\node[ inner sep=3pt, minimum size=6pt](E) at (4,4){$h_{w}$};

\node[ inner sep=3pt, minimum size=6pt] at (3.3,1.3){$h_{qv}$};
\node[ inner sep=3pt, minimum size=6pt] at (3.3,1.3){$h_{qv}$};

\node[ inner sep=3pt, minimum size=6pt] at (2.7,3.3){$h_{wv}$};

\draw[thick] (B)--(C);
\draw[thick] (B)--(E);
\end{tikzpicture}
\caption{}
\end{subfigure}
\caption{(a) Depiction of a two-local function $h$ on a tree with $4$ vertices. (b) The two-local function resulting from contraction of a leaf vertex $u$. Here $h'_v(x_v)=\sum_{x_u} h_{uv}(x_u,x_v)h_v(x_v)$ can be computed in time $O(|\Gamma|^2)$.\label{fig:contract}}
\end{figure}
\begin{lemma}
Suppose we are given a two-local function $h \colon \Gamma^{m}\rightarrow \mathbb C$ on an $m$-vertex tree $T$.  There is a classical algorithm which computes $\sum_{x\in \Gamma^m} h(x)$ using a runtime $O(m|\Gamma|^2)$.
\label{lem:treecontract}
\end{lemma}
\begin{proof}
Write $T=(W,F)$. Suppose $u\in W$ is a leaf of the tree and is incident to $v\in W$. Let $T'=(W',F')$ where $W'=W\setminus \{u\}$ and $F'=F\setminus \{u,v\}$ be the tree obtained by removing $u$ and its incident edge from $T$. We may compute a new two-local function $h'$ on $T'$ such that $\sum_{y\in \Gamma^{m-1}} h'(y)=\sum_{x\in \Gamma^{m}} h(x)$. This is depicted in Fig. \ref{fig:contract}. Indeed, define $h'_{r,s}=h_{r,s}$ for all edges $\{r,s\}\in F'$,  $h'_{r}=h_r$ for all vertices $r\in W'\setminus \{v\}$, and 
\begin{equation}
h'_v(x_v)=\sum_{x_u\in \Gamma} h_{uv}(x_u,x_v)h_v(x_v).
\label{eq:hprime}
\end{equation}
We can compute the new function $h'$ from $h$ by evaluating the sum Eq.~\eqref{eq:hprime} for each $x_v\in \Gamma$. Each such evaluation requires  $O(|\Gamma|)$ additions and multiplications, so the total runtime is $O(|\Gamma|^2)$.  To evaluate the sum $\sum_{x}h(x)$ we repeat this process $m$ times until we have removed all vertices of $T$. The total runtime is $O(m|\Gamma|^2)$.
\end{proof}

\begin{lemma}
Let $G=(V,E)$ be an $m$-vertex graph. Suppose we are given a two-local function $h \colon \{0,1\}^m\rightarrow \mathbb{C}$ on $G$, along with a tree decomposition of $G$ of width $w$ and containing $N$ nodes. There is a classical algorithm which computes 
\begin{equation}
\sum_{x\in \{0,1\}^m} h(x)
\end{equation}
using a runtime $O(N w 2^w)$. 
\label{lem:sum2local}
\end{lemma}
\begin{proof}
Let us write $T=(W,F)$ and $\mathcal{B}=\{B_i: i\in W\}$ for the given tree decomposition. With this notation we have $N=|W|$ and $w+1=\max_{i} |B_i|$. To simplify notation in the following we redefine $w\leftarrow w+1$ to be the width of the tree decomposition plus one.


For each bag $B_i$ with $i\in W$ we define a variable $y_i\in \{0,1\}^{w}$. Here $y_i$ contains a bit labeled by each vertex of $G$ in bag $i$ (i.e., each element of $B_i$), and also contains $w-|B_i|$ ancilla bits. These ancilla bits are only present to ensure, for convenience, that all variables $y_i$ are of the same length. We say that $y_i$ and $y_j$ are consistent if, (a) for each $v\in B_i\cap B_j$ the corresponding bit of $y_i$ agrees with that of $y_j$, and (b) all ancilla bits of $y_i$ are $0$ and all ancilla bits of $y_j$ are $0$.

Let $\Gamma=\{0,1\}^{w}$. We now show how to compute a function $Q \colon \Gamma^{N} \rightarrow \mathbb C$ which is two-local on $T$ and such that
\begin{equation}
\sum_{y\in \Gamma^{N} } Q(y)=\sum_{x\in \{0,1\}^m} h(x).
\label{eq:hw}
\end{equation}
By definition, such a two-local function is of the form 
\[
Q(y)=\prod_{\{i,j\}\in F} Q_{ij}(y_i, y_j)\prod_{i\in W} Q_i(y_i).
\]
We choose 
\[
Q_{ij}(y_i,y_j)=\begin{cases}1 &  y_i \text{ and } y_j \text{ are consistent }\\ 0 & \text{otherwise}.\end{cases}
\]
With this choice, the binary strings $y$ such that 
\[
\prod_{\{i,j\}\in F} Q_{ij}(y_i, y_j)=1
\]
are in one-to-one correspondence with assignments $x\in \{0,1\}^m$ of bits to each vertex of $G$. Let us write $x\sim y$ for this correspondence. We choose the functions at each node so that 
\begin{equation}
\prod_{i\in W} Q_i(y_i)=h(x)=\prod_{\{u,v\}\in E} h_{uv}(x_u,x_v)\prod_{u\in V} h_u(x_u) \qquad   \text{ for } \qquad x\sim y.
\label{eq:choicew}
\end{equation}
which then ensures Eq.~\eqref{eq:hw}. To satisfy Eq.~\eqref{eq:choicew}, note that (by definition of a tree decomposition) for each edge $\{u,v\}\in E$, there is a node $i$ such that $u\in B_i$ and $v\in B_i$.  The weight $h_{uv}(x_u,x_v)$ can then be incoporated into the function $Q_i$. Likewise, the weight $h_u(x_u)$ can be incorporated into any function $Q_i$ such that $u\in B_i$.  

Let's now upper bound the time needed to compute the truth table of each $Q_i$.  First, we can compute $Q_i(0\ldots0)$ in $O(w^2)$ time since there are at most $O(w^2)$ terms of $h$ on $w$-many vertices.  Notice that if we compute $Q_i(0\ldots01)$, there are $O(w)$ terms of $h$ whose contribution to $Q_i$ changes.  That is, given $Q_i(0\cdots0)$, computing $Q_i(0\ldots01)$ can be done in $O(w)$ time.  Of course, this argument generalizes for any two inputs which differ in a single input bit.  Therefore, we sequentially compute $Q_i(y_i)$ for $y_i \in \{0,1\}^w$ in Gray code order (i.e., each $y_i$ differs in at most 1 bit from the previously computed bit), leading to an $O(N w 2^w)$ time algorithm to compute $Q_i(y_i)$ for all $i \in W$ and $y_i \in \{0,1\}^w$. 

Since $Q$ is two-local on the tree $T$, we may then apply Lemma~\ref{lem:treecontract} which gives a classical algorithm to evaluate the sum Eq.~\eqref{eq:hw} in time $O(N4^w)$.  However, we can speed up this algorithm using the fact that $Q_{ij}(y_i,y_j)$ is particularly simple.  Using the same algorithm from Lemma~\ref{lem:treecontract}, we must compute
$$
Q'_v(y_v) = Q(y_v) \sum_{y_u \in \{0,1\}^w} Q_{uv}(y_u, y_v) Q_u(y_u)
$$
when we are contracting a child $u$ to its parent $v$ in the tree.  Suppose that node $v$ and node $u$ share $k$ many vertices, and for simplicity, let's assume that these are the first $k$ bits of both $y_u$ and $y_v$.  These $k$ bits partition $\sum_{y_u \in \{0,1\}^w} Q_{uv}(y_u, y_v) Q_u(y_u)$ into $2^k$ non-overlapping sums.   For example, if the first $k$ bits of $y_v$ are $0^k$, then 
$$
\sum_{y_u \in \{0,1\}^w} Q_{uv}(y_u, y_v) Q_u(y_u) = \sum_{x \in \{0,1\}^{w-k}} Q_u(0^k x) 
$$
where we have ignored the ancilla bits since we can always just set those bits to be 0.  If the first $k$ bits of $y_v$ are $z \neq 0^k$, then we get the sum $\sum_{x \in \{0,1\}^{w-k}} Q_u(z x)$ instead.  Therefore, computing $\sum_{y_u \in \{0,1\}^w} Q_{uv}(y_u, y_v) Q_u(y_u)$ for all the possible $k$-bit prefixes of $y_v$ can be done in time $O(2^w)$.  Finally, to compute $Q'(y_v)$ we must also multiply by the factor $Q(y_v)$.  Once again, we iterate over all $y_v$ in time $O(2^w)$.  Repeating this process for each node, we get a total runtime of $O(N 2^w)$, which completes the proof.
\end{proof}

\subsection{Classical algorithm for graph-based forrelation}
\label{sec:graph_based_proof}

In this section we prove Theorem~\ref{thm:graph_based}.
We first consider the special case in which all operators $O_j$ are unitary; at the end of this section we show how to handle the more general case. As discussed in Section~\ref{sec:graph_based}, 
it suffices to give efficient algorithms for computing amplitudes of a state
\[
|\alpha\ra= (O_A\otimes I_B)U_fH^{\otimes n}|0^n\ra.
\]
and for sampling the probability distribution $|\la x|\alpha\ra|^2$.
This is accomplished in Lemmas~\ref{lemma:amplitudes},\ref{lemma:sampling} below.
Let $G_A=(A,E_A)$  be the subgraph of $G$ induced by $A$.
By assumption, we are given a tree decomposition $T,\mathcal{B}$ of $G_A$ of width $w$ such that $T$ has  $O(n)$ nodes. 
\begin{lemma}
\label{lemma:amplitudes}
There is a classical algorithm which, given $x\in \{0,1\}^n$, computes an amplitude $\langle x|\alpha\rangle$ using a runtime $O(n w 2^w)$.
\end{lemma}
\begin{proof}
Below we use notation introduced in Section~\ref{sec:treewidth}.
Let us begin by writing
\begin{equation}
\langle x|\alpha\rangle=2^{-n/2} \sum_{y_A} 
\la x_A|O_A|y_A\ra
f(y_Ax_B).
\label{eq:px}
\end{equation}
Partition the edges of $G$ as $E=E_A\cup E_{\mathrm{cut}}\cup  E_B $ where $E_A$ are the edges with both endpoints in $A$, $E_B$ are the edges with no endpoints in $A$, and $E_{\mathrm{cut}}$ are the edges with exactly one endpoint in $A$. Since $f$ is a two-local function on $G$ we may write
\[
f(y_A x_B)=\prod_{\{u,v\}\in E_A} f_{uv}(y_u,y_v) \prod_{\{u,v\}\in E_{\mathrm{cut}}} f_{uv}(y_u, x_v) \prod_{\{u,v\}\in E_B} f_{uv}(x_u,x_v) \prod_{u\in A} f_u(y_u)\prod_{u\in B} f_u(x_u).
\]
We claim that for any fixed $x=x_Ax_B$, 
\[
h(y_A)\equiv 2^{-n/2} \la x_A|O_A|y_A\ra f(y_Ax_B), \qquad y_A\in \{0,1\}^{|A|}
\]
defines a two-local function on $G_A$. Indeed, we may write
\begin{equation}
h(y_A)=\omega \prod_{\{u,v\}\in E_A}h_{uv} (y_u,y_v) \prod_{u\in A} h_{u}(y_u)
\label{eq:F}
\end{equation}
where 
\[
h_{uv}(y_u,y_v)=  f_{uv}(y_u,y_v),
\]
\[
h_{u}(y_u)=\la x_u|O_u|y_u\ra
f_u(y_u)\prod_{v: \{u,v\}\in E_{\mathrm{cut}}} f_{uv}(y_u, x_v) 
\]
and $\omega=2^{-n/2}\prod_{\{u,v\}\in E_B} f_{uv}(x_u,x_v)\prod_{u\in B} f_u(x_u) \in \mathbb{C}$ can be absorbed into (say) one of the functions $h_u$ to make Eq.~\eqref{eq:F} of the form Eq.~\eqref{eq:twoloc}. We then apply Lemma \ref{lem:sum2local} which completes the proof.
\end{proof}
\begin{lemma}
\label{lemma:sampling}
There is a classical algorithm which samples a binary string $x\in \{0,1\}^n$ from the distribution $P(x)=|\langle x|\alpha\rangle|^2$ using a runtime $O(n^2 w 4^w)$.
\end{lemma}
\begin{proof}
It will be convenient to order the bits of $x$ as $x_Bx_A$, that is, all bits of $B$ appear first. We shall use the well-known linear time reduction from sampling to computing marginal probabilities.
That is, for each $\ell\in [n]$ and binary string $x_1x_2\ldots x_\ell$, we define a marginal probability
\begin{equation}
P_\ell(x_1x_2\ldots x_\ell)=\sum_{x_{\ell+1}x_{\ell+2}\ldots x_n} P(x).
\label{eq:marg}
\end{equation}
To sample from $P(x)$ we first sample $x_1\in \{0,1\}$ from $P_1(x_1)$, then we sample $x_2$ from $P_2(x_1 x_2)/P_1(x_1)$, then $x_3$  from $P_3(x_1 x_2 x_3)/P_2(x_1x_2)$, and so on.  To complete the proof, below we show that each marginal Eq.~\eqref{eq:marg} can be computed using a runtime $O(nw4^w)$.

If $\ell\leq |B|$, then 
\[
P_\ell(x_1x_2\ldots x_\ell)=
\la 0^n|H^{\otimes n} U_f^\dag
(|x_1x_2\ldots x_\ell\rangle \langle x_1x_2\ldots x_\ell|\otimes I_{|B|-\ell}\otimes O_A^\dag O_A) U_f H^{\otimes n}|0^n\ra.
\]
Since all operators $O_j$ are unitary, one can replace $O_A^\dag O_A$
with the identity. Then
$U_f^\dag$ and $U_f$ cancel each other since
$U_f$ commutes with any diagonal operator. Thus 
\[
P_\ell(x_1x_2\ldots x_\ell)=
\la 0^n|H^{\otimes n}(|x_1x_2\ldots x_\ell\rangle \langle x_1x_2\ldots x_\ell|\otimes I_{n-\ell})H^{\otimes n}|0^n\ra
=\frac{1}{2^\ell}.
\]

Suppose now that $\ell\geq |B|+1$.
Define a partition $A=CD$, where $C=A\cap \{1,2,\ldots,\ell\}$
includes all bits of $A$ that contribute to the marginal probability and
$D$ includes all bits of $A$ that are traced out. 
Then $x_1x_2\ldots x_\ell=x_Bx_C$.
 Using Eq.~\eqref{eq:px} we may write
\begin{align}
P_\ell(x_B x_C)&=2^{-n}\sum_{x_D} \left| \sum_{y_{A}} 
\la x_C x_D|O_A|y_A\ra
f(x_By_A)\right|^2\nonumber\\
&
=2^{-n}\sum_{x_D,y_A,z_A}
\la x_C x_D|O_A|y_A\ra^*\la x_C x_D|O_A|z_A\ra
f^*(x_By_A) f(x_Bz_A).
\end{align}
For any subset of qubits $M$ we shall write $O_M$ for 
the tensor product of $O_j$ over $j\in M$. Then
$O_A = O_C\otimes O_D$ and thus
\[
\sum_{x_D}
\la x_C x_D|O_A|y_A\ra^*\la x_C x_D|O_A|z_A\ra=
\left(\la x_C|O_C|y_C\ra^* \la x_C|O_C|z_C\ra \right)
\sum_{x_D} \la x_D|O_D|y_D\ra^* \la x_D|O_D|z_D\ra.
\]
The unitarity of $O_j$ implies
\[
\sum_{x_D} \la x_D|O_D|y_D\ra^* \la x_D|O_D|z_D\ra = 
\la y_D|O_D^\dag O_D|z_D\ra = \la y_D|z_D\ra = \delta_{y_D,z_D}.
\]
We arrive at 
\begin{align}
P_\ell(x_B x_C)&=2^{-n} \sum_{y_C,y_C',y_D}
\la x_C|O_C|y_C\ra^* \la x_C|O_C|y_C'\ra f^*(x_By_Cy_D) f(x_By_C' y_D) \nonumber \\
&= \sum_{y_C,y_C',y_D} h(y_C,y_C',y_D),\label{eq:probsum}
\end{align}
where
\[
h(y_C,y_C',y_D) = 2^{-n} \la x_C|O_C|y_C\ra^* \la x_C|O_C|y_C'\ra f^*(x_By_Cy_D) f(x_By_C' y_D).
\]
We claim that $h$ is a two-local function on a certain augmented
graph
$\tilde{G}=(\tilde{V},\tilde{E})$ with vertex set $\tilde{V}=C\cup C'\cup D$, where $C'$ is a second copy of $C$. To construct this graph, start with a disjoint union of $G_A$ and $G_{C'}$ (the induced subgraphs of $G$ on vertices in $A,C'$ respectively). For each edge $\{u,v\}\in E_A$ with 
$u\in C$ and $v\in D$, we then add a new edge $\{u',v\}$ 
between $v$ and the copy $u'\in C'$. Since $f$ is two-local on $G$, a direct inspection
shows that each two-local term in $f^*(x_By_C y_D)$
or $f(x_B y_C' y_D)$ 
that depends on the variables $y$ or $y'$
involves nearest-neighbor variables
in the augmented graph $\tilde{G}$. 
The remaining two-local terms
in $f^*(x_By_C y_D)$
or $f(x_B y_C' y_D)$ 
that depend on at least one variable $x$ become $1$-local 
since $x_B$ is fixed. 
Furthermore,
the matrix elements $\la x_C|O_C|y_C\ra$
and $\la x_C|O_C|y_C'\ra$ are $1$-local functions
of $y_C$ and $y_C'$. Thus $h$ is a two-local function on the augmented graph $\tilde{G}$.

We claim that the treewidth $\tilde{w}$ of $\tilde{G}$ satisfies $\tilde{w}\leq 2w+1$. Indeed, let $T, \mathcal{B}$ be the given tree decomposition of $G_A$ of width $w$. Define a tree decomposition $\tilde{T}, \tilde{\mathcal{B}}$ of $\tilde{G}$ as follows. First set $\tilde{T}=T$. Then, for each node $i$ of $T$,  let $S_i=B_i\cap C$ be the set of vertices of $C$ that appear in $B_i$.  Let $S'_i$ contain the corresponding vertices of $C'$, and define 
\[
\tilde{B}_{i}=B_i\cup S'_i.
\]
One can straightforwardly check that this defines a tree decomposition $\tilde{T}, \tilde{\mathcal{B}}$ of $\tilde{G}$. The upper bound  $\tilde{w}\leq 2w+1$ follows from the fact that each bag can at most double in size.

To complete the proof, we apply Lemma \ref{lem:sum2local} which allows us to compute the sum Eq.~\eqref{eq:probsum} using a runtime $O(n (2w+1) 2^{2w+1})=O(n w 4^w)$. Since sampling $x$ from $P(x)$
requires computation of $n$ marginal probabilities
the overall runtime of the algorithm is $O(n^2 w 4^w)$.
\end{proof}

Using the algorithm from Lemma~\ref{lemma:sampling} suffices to give an efficient algorithm for the graph-based forrelation problem as we will describe below.  Indeed, this is the algorithm that will be used in our numerical calculations of Section~\ref{sec:QAOA_numerical}.  Nevertheless, an asymptotically more efficient algorithm exists which we prove in Appendix~\ref{sec:linear_time_sampling}:
\begin{lemma}
\label{lem:linear_sampling_simple}
There is a classical algorithm which samples a binary string $x\in \{0,1\}^n$ from the distribution $P(x)=|\langle x|\alpha\rangle|^2$ using a runtime $O(n 4^w)$.
\end{lemma}

In either case, let us now describe how to obtain Theorem~\ref{thm:graph_based} from the lemmas above.  
Let $\cal A$ be a randomized algorithm that samples
a bit string $x\in \{0,1\}^n$ from the distribution $|\la x|\alpha\ra|^2$ and outputs the quantity $\mu = \la \beta |x\ra/\la \alpha|x\ra$. We have
${\mathbb E}(\mu)=\la \beta|\alpha\ra = \Phi$
and ${\mathbb E}(|\mu|^2)=\la \beta|\beta\ra \le 1$.
Thus $\mu$ is an unbiased estimator of $\Phi$ with variance at most one. Lemmas~\ref{lemma:amplitudes} and \ref{lem:linear_sampling_simple} will imply that
imply that $\cal A$ has runtime $O(n 4^w)$.
By generating $S=100\epsilon^{-2}$ independent samples of $\mu$ using $S$ calls to the algorithm $\cal A$ and
computing the sample mean value one can approximate
$\Phi$ with an additive error $\epsilon$ in time
$O(n 4^w \epsilon^{-2})$, as claimed. 

So far we have assumed that $O_j$ are unitary operators.\footnote{While this is not technically required for Lemma~\ref{lem:linear_sampling_simple}, our running will not suffer from making this assumption, and the following argument can be made for both Lemma~\ref{lemma:sampling} and Lemma~\ref{lem:linear_sampling_simple}.} Now let us consider the more general case where we only require $\|O_j\|\le 1$ for all $j$. Using the singular value
decomposition one can write any $2\times 2$ matrix
$M$ as
\[
M=\|M\| \cdot U \left[\begin{array}{cc} 1 & 0 \\ 0 & s \\ \end{array} \right] V
\]
for some $0\le s\le 1$ and some $2\times 2$ unitary matrices $U,V$.
The identity 
\[
\left[\begin{array}{cc} 1 & 0 \\ 0 & s \\ \end{array} \right]
= \frac{(1+s)}2 I + \frac{(1-s)}2 Z
\]
gives  $M=\|M\|(q_0 M_0 + q_1M_1)$,
where $M_0=UV$ and $M_1=UZV$ are unitary, $q_0=(1+s)/2$, and
$q_1=(1-s)/2$.
Applying this decomposition to each operator $O_j$ one gets
\[
\Phi = \omega \sum_{z \in \{0,1\}^n} q(z) \Phi(z),
\]
where $\omega=\prod_{j=1}^n \|O_j\|\le 1$ is the normalization constant,
$q(z)$ is a product distribution, and 
\[
\Phi(z)=\la 0^n|H^{\otimes n} U_g (O_{1,z_1}\otimes O_{2,z_2} \otimes \cdots \otimes O_{n,z_n}) U_f H^{\otimes n}|0^n\ra,
\]
for some unitary operators $O_{j,0}$ and $O_{j,1}$.
As shown above, there exists 
a randomized algorithm
${\cal A}$ that takes as input a bit string $z\in \{0,1\}^n$
and outputs a random variable $\mu(z)\in \mathbb C$ such that
${\mathbb E}(\mu(z))=\Phi(z)$ and ${\mathbb E}(|\mu(z)|^2)\le 1$.
Here the mean values  are taken over the internal randomness of ${\cal A}$.
We can now estimate $\Phi$ by Monte Carlo method.
Generate $S=100\epsilon^{-2}$ independent samples
$z^1,\ldots,z^S$ from the distribution $q(z)$.
For each sample $z^j$ call the algorithm $\cal A$
to obtain the estimator $\mu(z)$. Finally, approximate
$\Phi$ by the quantity
\[
\eta = \frac{\omega}S \sum_{j=1}^S \mu(z^j).
\]
One can straightforwardly check that ${\mathbb E}(\eta) = \Phi$.
Furthermore, 
\[
\mathrm{Var}(\eta)=\frac{\omega^2}S
\mathrm{Var}(\mu(z^1))
\le
\frac1S \sum_z q(z) 
{\mathbb E}(|\mu(z)|^2) \le  \frac1S=\frac{\epsilon^2}{100}.
\]
By the Chebyshev inequality, $|\eta - \Phi|\le \epsilon$
with probability at least $0.99$.

\section{Applications to quantum approximate optimization} 
\label{sec:QAOA}

In this section we discuss the connection between the classical algorithm for graph-based forrelation and certain quantum approximate optimization algorithms (QAOA) \cite{farhi2014quantum}.  In particular, in Section~\ref{sec:QAOA_map} we show how graph-based forrelation can be be used to calculate the variational energies for the level-2 QAOA algorithm on planar and bipartite graphs.  In Section~\ref{sec:QAOA_numerical}, we describe an implementation of this algorithm that can accurately estimate these energies.  Finally, in Section~\ref{sec:RQAOA}, we discuss the recursive QAOA algorithm \cite{bravyi2019obstacles} for which our graph-based forrelation approach can give a complete classical simulation.  We then combine our implementation for computing level-2 QAOA energies with the recursive QAOA algorithm to show that our algorithm can be leveraged for solving large optimization problems.

\subsection{Quantum mean value problem for level-2 QAOA}
\label{sec:QAOA_map}
Suppose $G=(V,E)$ is a graph with $n$ vertices. We place a qubit at each vertex of $G$.
Consider a diagonal $n$-qubit  Hamiltonian
\[
C=\sum_{\{p,q\}\in E} J_{p,q} Z_p Z_q
\]
where $J_{p,q}$ are real coefficients.
It describes a classical Ising-type cost function
defined on the graph $G$.
Note that $\langle z|C |z\rangle$ is a $2$-local function on $G$. 
The Quantum Approximate Optimization Algorithm (QAOA)
introduced in~\cite{farhi2014quantum}
maximizes the expected value $\la \psi|C|\psi\ra$
over a suitable class of $n$-qubit variational states $\psi$
that depend on a few parameters.
Once the optimal variational
state $\psi$ is found, a bit string $x\in \{0,1\}^n$
is sampled from the distribution $|\la x|\psi\ra|^2$ obtained
by measuring every qubit of $\psi$. The expected value of the
cost function $\la x|C|x\ra$ coincides with the optimal
variational energy $\la \psi|C|\psi\ra$.

Here we focus on the level-$2$ version of QAOA.
The corresponding variational states are
generated by a  quantum circuit with two entangling layers,
\[
|\psi\ra = e^{-i\beta_2 B} e^{-i\gamma_2 C} e^{-i\beta_1 B} e^{-i\gamma_1 C}|+^n\ra.
\]
Here, $\beta_\ell,\gamma_\ell\in {\mathbb{R}}$ are variational parameters, $|+^n\ra=H^{\otimes n}|0^n\ra$,
and $B=X_1+X_2+\ldots+X_n$.
Consider the first step of QAOA -- computing
the expected value $\la \psi|C|\psi\ra$.
By linearity, it suffices to compute 
the expected value 
$\la \psi|Z_s Z_t|\psi\ra$ for some fixed edge $\{s,t\}\in E$.
The standard lightcone argument shows that $\la \psi|Z_s Z_t|\psi\ra$ only depends
on coefficients $J_{p,q}$ such that the edge $\{p,q\}$ is incident to $s$, or $t$, or one of nearest
neighbors of $s,t$. The remaining edges are irrelevant and can be removed from the graph. 
Below we assume that the graph has been truncated such that all irrelevant edges are removed.

Inserting the identity decompositions on qubits $s,t$ between the first $B$-layer and the second
$C$-layer one gets
\[
\la \psi|Z_sZ_t|\psi\ra = \sum_{x\in \{0,1\}^4} \mu(x)
\]
where
\[
\mu(x)=
\la +^n| e^{i\gamma_1 C} e^{i\beta_1 B} |x_1x_2\ra\la x_1x_2|_{s,t} e^{i\gamma_2 C} e^{i\beta_2 B}
Z_s Z_t  e^{-i\beta_2 B}e^{-i\gamma_2 C}|x_3x_4\ra\la x_3x_4|_{s,t} e^{-i\beta_1 B} e^{-i\gamma_1 C}|+^n\ra
\]
Using the fact that $X^{\otimes n}$ commutes with $B,C$, $Z_sZ_t$ and $X^{\otimes n}|+^n\ra=|+^n\ra$
one can easily check that $\mu(x)$ has symmetries
$\mu(x\oplus 1111)=\mu(x)$ and $\mu(x_3x_4x_1x_2)=\mu^*(x_1x_2x_3x_4)$.
Since $\la \psi|Z_sZ_t|\psi\ra$ is real, one can write
\be
\label{mu_eq1}
\la \psi|Z_sZ_t|\psi\ra = \mathrm{Re}
\left[ 2\mu(0000) + 4\mu(0010) + 4\mu(0001)+2\mu(0011)+2\mu(0110)+2\mu(0101)\right].
\ee
Let us  show that the  quantity $\mu(x)$ can be expressed as a graph-based forrelation, namely,
\be
\label{mu_eq2}
\mu(x)=\la x_1 x_2 |U|x_3 x_4\ra \cdot \la +^n|e^{i\gamma_1C} O_1(x)\otimes O_2(x)\otimes \cdots \otimes O_n(x)e^{-i\gamma_1 C}|+^n\ra
\ee
for  some single-qubit operators $O_j(x)$ such that $\|O_j(x)\|\le 1$ and
a two-qubit operator
\be
\label{mu_eq3}
U=e^{i\gamma_2 J_{s,t} Z\otimes Z} e^{i\beta_2(X\otimes I + I\otimes X)}
Z\otimes Z e^{-i\beta_2(X\otimes I + I\otimes X)}e^{-i\gamma_2 J_{s,t} Z\otimes Z}.
\ee
Indeed, a simple algebra shows that 
\[
\mu(x)=\la x_1 x_2 |U|x_3 x_4\ra \cdot \la +^n |e^{i\gamma_1 C} e^{i\beta_1 B} e^{i\gamma_2 C'} |x_1x_2\ra\la x_3 x_4|_{s,t}
e^{-i\gamma_2 C'} e^{-i\beta_1 B} e^{-i\gamma_1 C}|+^n\ra
\]
where $C'$ is obtained from $C$ by retaining only the edges incident to exactly one of the vertices $s,t$.
In other words,
\[
C'=  \sum_{p\in V\setminus \{s,t\}} J_{s,p} Z_s Z_p + J_{t,p}Z_t Z_p.
\]
Replacing
$Z_s$ and $Z_t$ in $C'$ by their eigenvalues leads to Eq.~(\ref{mu_eq2}) where
\be
\label{mu_eq4}
O_s(x) = e^{i\beta_1 X} |x_1\ra\la x_3| e^{-i \beta_1 X},
\ee
\be
\label{mu_eq5}
O_t(x) = e^{i\beta_1 X}|x_2\ra\la x_4| e^{-i \beta_1 X},
\ee
and
\be
\label{mu_eq6}
O_p(x) = e^{i\beta_1 X} e^{i\gamma_2 (J_{s,p} (-1)^{x_1}  + J_{t,p} (-1)^{x_2} - J_{s,p} (-1)^{x_3} - J_{t,p}(-1)^{x_4})Z}
e^{-i\beta_1 X}
\ee
for all $p\notin \{s,t\}$.
Combining Eqs.~(\ref{mu_eq1},\ref{mu_eq2},\ref{mu_eq3},\ref{mu_eq4},\ref{mu_eq5},\ref{mu_eq6})
shows that computing the mean value $\la \psi|Z_s Z_t|\psi\ra$ can be reduced
to solving six instances of the graph-based forrelation problem.

Finally, although there are efficient classical algorithms for maximizing Ising cost functions over planar graphs \cite{schraudolph2008efficient}, we note that for Ising cost functions with non-zero magnetic fields (i.e., the Hamiltonian contains $J_a Z_a$ terms) no such polynomial-time algorithm is known.  Furthermore, the above reduction from
level-$2$ QAOA to graph-based forrelation 
can be easily extended to these more general cost
functions $C$ that contain both quadratic and linear terms.

As a consequence, our classical algorithm for graph-based forrelation problems can be deployed to efficiently compute mean values $\langle \psi|C|\psi\rangle$ for level-$2$ QAOA on planar and bipartite graphs. This allows one to study the \textit{performance} of these QAOA algorithms using a classical computer.  Note that it does not provide a classical simulation of the full QAOA algorithm---a quantum computer is still needed to sample from the variational state $\ket{\psi}$ once a suitable choice of variational parameters has been found.

\subsection{Numerical estimation of quantum mean values}
\label{sec:QAOA_numerical}

We implemented our graph-based forrelation algorithm\footnote{To be clear, the running time of our algorithm will scale quadratically in $n$ due to the fact that we use the sampling algorithm from Lemma~\ref{lemma:sampling}} and used it to calculate level-2 QAOA variational energies.  To test our implementation, we calculated $\bra{\psi} C \ket{\psi}$ where $C$ is the 2-local Hamiltonian with randomly-selected $\pm 1$ couplings over the 10-qubit triangular lattice shown below:
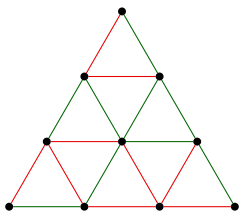
\begin{center}
\begin{tikzpicture}[scale=1]
\tikzstyle{vertex}=[draw, circle, minimum size=10pt, inner sep=0pt, outer sep=0pt]

\node[fill, circle, minimum size=3pt, inner sep=0pt, outer sep=0pt] (0) at (0, -0) {};
\node[fill, circle, minimum size=3pt, inner sep=0pt, outer sep=0pt] (1) at (-0.5, -0.866025) {};
\node[fill, circle, minimum size=3pt, inner sep=0pt, outer sep=0pt] (2) at (0.5, -0.866025) {};
\node[fill, circle, minimum size=3pt, inner sep=0pt, outer sep=0pt] (3) at (-1, -1.73205) {};
\node[fill, circle, minimum size=3pt, inner sep=0pt, outer sep=0pt] (4) at (0, -1.73205) {};
\node[fill, circle, minimum size=3pt, inner sep=0pt, outer sep=0pt] (5) at (1, -1.73205) {};
\node[fill, circle, minimum size=3pt, inner sep=0pt, outer sep=0pt] (6) at (-1.5, -2.59808) {};
\node[fill, circle, minimum size=3pt, inner sep=0pt, outer sep=0pt] (7) at (-0.5, -2.59808) {};
\node[fill, circle, minimum size=3pt, inner sep=0pt, outer sep=0pt] (8) at (0.5, -2.59808) {};
\node[fill, circle, minimum size=3pt, inner sep=0pt, outer sep=0pt] (9) at (1.5, -2.59808) {};

\path[color=black!60!green] (2) edge node {} (0);
\path[color=black!10!red] (1) edge node {} (0);
\path[color=black!10!red] (2) edge node {} (1);
\path[color=black!60!green] (4) edge node {} (1);
\path[color=black!60!green] (3) edge node {} (1);
\path[color=black!60!green] (5) edge node {} (2);
\path[color=black!60!green] (4) edge node {} (2);
\path[color=black!10!red] (4) edge node {} (3);
\path[color=black!10!red] (7) edge node {} (3);
\path[color=black!10!red] (6) edge node {} (3);
\path[color=black!60!green] (5) edge node {} (4);
\path[color=black!10!red] (8) edge node {} (4);
\path[color=black!60!green] (7) edge node {} (4);
\path[color=black!60!green] (9) edge node {} (5);
\path[color=black!10!red] (8) edge node {} (5);
\path[color=black!60!green] (7) edge node {} (6);
\path[color=black!10!red] (8) edge node {} (7);
\path[color=black!10!red] (9) edge node {} (8);

\end{tikzpicture}

\end{center}

A green edge corresponds to a $+1$ coupling, while a red edge corresponds to a $-1$ coupling.
The triangular lattice was chosen over the usual square lattice due to the fact that it is not bipartite, and therefore requires the more sophisticated partitioning of vertices used in the graph-based forrelation algorithm for general planar graphs.

In Figure~\ref{fig:epsilon_graph}, we graph the accuracy of the estimate as a function of $\epsilon$ where $\epsilon^{-2}$ samples were taken for each unitary graph-based forrelation instance.  Recall, however, that  we estimate each local term of the Hamiltonian separately (of which there are 18 in the example).  We randomly chose the variational angles in the range $[0, 2\pi)$ ($\beta_1 = 1.44433$, $\beta_2 = 3.56786$, $\gamma_1 = 0.937498$, and $\gamma_2 = 4.93861$).  These choices were fixed for all $\epsilon$.  The exact energy ($-4.35917$) was calculated explicitly using Mathematica.

\begin{figure}
\centering
\input{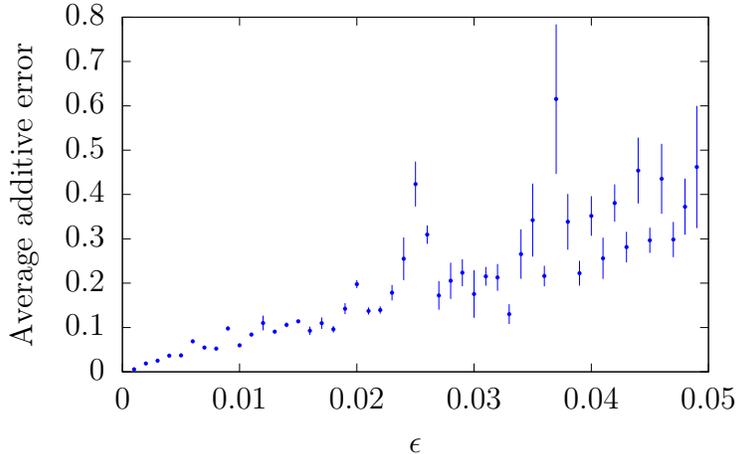}
\caption{Additive estimation error for the energy of the triangular lattice as a function of accuracy parameter $\epsilon$.  Each data point is the average of 10 trials with variance shown by the error bars.  The parameters $\beta_1 = 1.44433$, $\beta_2 = 3.56786$, $\gamma_1 = 0.937498$, and $\gamma_2 = 4.93861$ were fixed for all $\epsilon$, and the exact energy $-4.35917$ was calculated by brute force.  We note that we expect linear scaling with $\epsilon$.}
\label{fig:epsilon_graph}
\end{figure}

\subsection{The level-2 RQAOA algorithm and implementation}
\label{sec:RQAOA}

Recall that the ability to optimize the variational parameters of the QAOA algorithm does not immediately confer the ability to carry out the entire QAOA heuristic due to the fact that a quantum computer is needed to sample from the variational state $\ket{\psi}$.  Instead, we turn to the Recursive QAOA algorithm (RQAOA) of Bravyi, Kliesch, K\"onig, and Tang \cite{bravyi2019obstacles} which can be simulated entirely given an algorithm to compute expected values and optimize variational angles.  The algorithm has also been shown to perform well in certain cases in which the usual QAOA is provably suboptimal.

Let us briefly describe this algorithm.  Suppose we have some cost function $C \colon \{-1,1\}^n \to \mathbb R$ we would like to maximize.  For our purposes, let us also suppose that $C$ is an Ising cost function of the form
$$
C(z) = \sum_{\{p,q\} \in E} J_{p,q} z_p z_q
$$
where $E$ is the edge set of some graph, and the $J_{p,q}$ are real coefficients.  As in the QAOA algorithm, we promote the cost function to a Hamiltonian, i.e., $C = \sum_{\{p,q\} \in E} J_{p,q} Z_p Z_q$, and optimize the variational parameters of the state
$$
\ket{\psi} = e^{-i\beta_2 B} e^{-i\gamma_2 C} e^{-i\beta_1 B} e^{-i\gamma_1 C} \ket{+^n}
$$
to maximize $\bra{\psi} C \ket{\psi}$.  Given these parameters, we compute the values
$$
M_{p,q} = \bra{\psi} Z_p Z_q \ket{\psi}
$$
for each edge $\{p,q\} \in E$.  We then choose whichever edge $\{p,q\}$ showed the strongest correlation, that is,
$$
\{p,q\} = \operatorname*{arg\,max}_{\{a,b\} \in E} |M_{a,b}|.
$$
Finally, we impose the constraint $z_p = \mathrm{sgn}(M_{p,q}) z_q$ and eliminate $z_p$ from the cost function so that we obtain a new Ising cost function on $n-1$ variables.  That is, after eliminating variable $z_p$, the term $J_{r,p} z_r z_p$ in the cost function becomes $J_{r,p}  \mathrm{sgn}(M_{p,q}) z_r z_q$.  We recurse on this smaller instance (finding the new optimal parameters, choosing an edge, etc.) until some threshold number of variables is met, after which we simply brute force the optimal solution.

The RQAOA algorithm has an important property 
of preserving planarity of the underlying graph $G$. Namely, the 
variable elimination step 
of the algorithm, i.e., constraining $z_p$ and $z_q$ to be (anti)correlated, corresponds to a contraction 
of the edge $\{p,q\}$ in the graph $G$
---the term $J_{p,q} z_p z_q$ in the cost function disappears and becomes a constant energy shift, and all edges incident to those two vertices become incident to the merged node.  It is well-known that edge contractions preserve planarity.\footnote{For example, one can appeal to Wagner's theorem which states that a finite graph is planar iff it excludes $K_5$ and $K_{3,3}$ as minors.  By definition, taking minors involves contracting edges, and so if the minor was excluded before then it is clearly excluded after contracting an edge.}
Thus, recursive steps in RQAOA generate a family of Ising cost functions
defined on planar graphs, as long as the initial cost function is defined on a planar graph.

Let us now describe our implementation of the level-2 RQAOA algorithm on 2D grid graphs.  We optimize the variational angles in two phases.  First, we set $\beta_2 = \gamma_2 = 0$ and optimize $\beta_1$ and $\gamma_1$ using the solution for level-1 QAOA \cite{bravyi2019obstacles}.  We fix the optimal $\beta_1$ and $\gamma_1$ obtained at this step.  Second, we search over a grid of possible $\beta_2$ and $\gamma_2$ values and choose whichever angles yield the largest energy using our algorithm for calculating 2-level variational energies.  We note that this approach ensures that the energy computed at each step is at least the energy computed by the level-1 QAOA algorithm. Indeed, we observe empirically that the energy we obtain is almost always strictly greater than the level-1 energy.  That said, our approach does not guarantee we obtain the largest possible level-2 variational energy.

The second step above is aided by the following observation:  for fixed angles $\beta_1, \gamma_1, \gamma_2$, the value of $\beta_2$ which maximizes the variational energy can be computed by calculating the energy at only three different values of $\beta_2$.  Therefore, the two-dimensional search in the second step above essentially reduces to a one-dimensional search.  We describe this reduction in Appendix~\ref{app:eliminate_beta_2}.

Finally, in order to more efficiently probe larger and more complex optimization problems using level-2 RQAOA, we supplement our graph-based forrelation algorithm for calculating the variational energies with an explicit brute force solution.  It is simple to see that the expected value of a particular $Z_s Z_t$ term only depends on neighbors of the neighbors of the edge $\{s,t\}$---that is, the lightcone of the $Z_s Z_t$ term.  When this set of qubits is small, we can brute force the expected values by explicitly writing out the state of the system.  On the other hand, when this set of qubits is large, the exponential running time of the brute force algorithm becomes prohibitive and we must turn to our graph-based forrelation algorithm.  We describe our exact brute force algorithm (which is more sophisticated than the one hinted at above) in Appendix~\ref{all:brute_force}.

Notice that although the 2D grid graph has low degree, the graph obtained by contracting various edges does not in general remain low degree.  Therefore, while the earlier steps of the RQAOA algorithm can be computed by brute force, the later steps must be computed by our graph-based forrelation algorithm. This phenomenon can be observed in Figure~\ref{fig:15x15_rqaoa}.

We test this implementation of the level-2 RQAOA on $10 \times 10$ and $15 \times 15$ grid graphs with random $\pm 1$ couplings.  In both tests, we set $\epsilon = .03$ (once again, taking $\epsilon^{-2}$ many samples), brute force the solution when $10$ vertices remain, and search over a grid of 30 points to optimize the level-2 variational angles.  In ten random instances on the $10 \times 10$ grid, there was only one for which our implementation did not yield an optimal solution (and even in this case, the solution obtained was 98\% of the optimal value).  While impressive, we note that the level-1 RQAOA algorithm also performs very well on such instances, and so we do not attempt to make a qualitative comparison between the two.  That said, we stress that we know of no faster method for calculating the level-2 variational energies and that the energies given by our level-2 algorithm are greater than those given by the level-1 algorithm.  Finally, we tested a single random instance on the $15 \times 15$ grid and give a detailed breakdown of its runtime in Figure~\ref{fig:15x15_rqaoa}.  The algorithm found an optimal assignment to the vertices to maximizes the cost function, which we computed explicitly using Gurobi's integer quadratic programming software \cite{gurobi}.

\begin{figure}
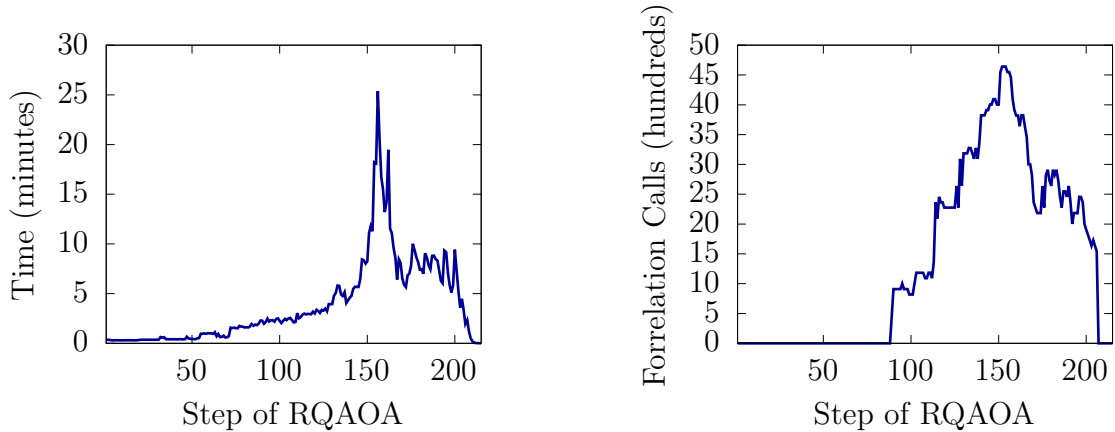

\centering
\parbox{.40\textwidth}{
    \input{15x15_timee}
} \qquad\qquad
\begin{minipage}{.40\textwidth}
    \input{15x15_ZZss}
\end{minipage}
\caption{The graphs above depict a single run of the level-2 RQAOA algorithm on a $15 \times 15$ grid with random $\pm 1$ couplings.  We set the parameters $\epsilon = .03$, brute force the solution when $10$ vertices remain, and search over 30 possible values of $\gamma_2$ to optimize the level-2 variational angles.  The algorithm returned an optimal assignment to the vertices to maximize the cost function. \emph{Left:}  The running time of the algorithm at each step of the RQOA algorithm. \emph{Right:} The number of calls we make to our graph-based forrelation algorithm (e.g., a value of 0 implies that only the brute-force algorithm was used).}
\label{fig:15x15_rqaoa}
\end{figure}

\section{Conclusions}
\label{sec:conclusions}

Quantum computers available today may already be capable of certain specialized tasks that are too hard for existing classical computers. A convincing demonstration of a quantum computational advantage requires evidence that the task is hard for classical computers.  Even if we focus on problems which admit asymptotic quantum speedups, one would like to ensure that the problem size is large enough to defeat existing classical machines.  This requires an accurate estimate of the runtime required by classical algorithms.
Our work provides  such estimates for the forrelation problem -- a powerful computational primitive
that achieves nearly maximal quantum speedup, as measured by query complexity. 
We developed classical algorithms both for the oracle-based and for an explicitly defined graph-based
forrelation problems. In the oracular setting our algorithm is almost optimal as its runtime nearly matches
the query complexity lower bound. This provides a nearly quadratic speedup ($2^{n/2}$ vs $2^n$) compared with the best previously known
algorithms reported in Refs.~\cite{aaronson2010bqp,aaronson2018forrelation}.  In the graph-based setting our algorithm has a
linear runtime for a large family of graphs including any bipartite or any planar graph.
This extends a line of work studying classical algorithms for problems that promise a quantum advantage such as
boson sampling~\cite{clifford2018classical},
random circuit sampling~\cite{aaronson2016complexity},
or IQP circuits~\cite{kahanamoku2019forging}.
Our work also provides a new tool to study the performance of quantum approximate optimization algorithms. The ability to calculate the 
variational energy of QAOA states classically eliminates the need to prepare the state on a quantum device while optimizing the QAOA parameters. Alternatively, our classical simulator can be used to benchmark QAOA on large problem sizes that are not yet accessible
in the experiment. 

While we have provided strong evidence that a relative-error version of graph-based forrelation is intractable, an intriguing open question is whether or not the standard additive-error version is hard for classical machines. Alternatively, one may attempt to improve or extend the reach of our efficient classical algorithms. A modest goal would be to improve the error dependence of the classical algorithm for graph-based forrelation on planar graphs from $\epsilon^{-2}$ to $\epsilon^{-1}$, matching the error scaling of our algorithm for oracle-based forrelation. Such an improvement would likely provide a substantial enhancement of the practical performance of QAOA simulation.

A related open question is whether some special instances of the graph-based forrelation problem
 can be used to demonstrate a computational quantum advantage with an efficient classical verification. 
 For example, suppose $G=(V,E)$ is an $n$-vertex graph
 that admits a partition of vertices $V=AB$ such that the subgraphs $G_A$ and $G_B$ induced by $A$ and $B$ have treewidth $O(\log{n})$. Suppose also that finding a low-treewidth partition as above 
 from scratch is a classically hard problem. The forrelation $\Phi$ based on such graph $G$
 can be efficiently (in time polynomial in $n$) approximated by the classical algorithm of Section~\ref{sec:graph_based} only if 
 a low-treewidth partition $AB$ is explicitly specified. Meanwhile, a quantum computer can 
 efficiently approximate the forrelation $\Phi$ in all cases. Imagine a classical verifier in possession of a secret low-treewidth partition $V=AB$ and a prover who claims to have a quantum computer. The verifier can 
challenge the prover to approximate the forrelation $\Phi$ without access to the secret partition $AB$.
Once the prover provides an estimate of $\Phi$, the verifier can efficiently test whether this
estimate is accurate enough using the algorithm of Section~\ref{sec:graph_based}.
The ability to pass this test may serve as a proof of quantumness, assuming that
graphs $G$ with the desired properties indeed exist.

\section{Acknowledgments}
D. Gosset and D. Grier are supported in part by IBM Research. D. Gosset also acknowledges the support of the Natural Sciences and Engineering Research Council of Canada and the Canadian Institute for Advanced Research.
This work is supported in part by the Army Research Office under Grant Number W911NF-20-1-0014. The views and conclusions contained in this document are those of the authors and should not be interpreted as representing the official policies, either expressed or implied, of the Army Research Office or the U.S. Government. The U.S. Government is authorized to reproduce and distribute reprints for Government purposes notwithstanding any copyright notation herein.


\appendix

\section{Linear-time graph-based forrelation algorithm}
\label{sec:linear_time_sampling}

In this appendix, we prove Lemma~\ref{lem:linear_sampling_simple} which improves the running time of Lemma~\ref{lemma:sampling} to linear in $n$. The key to the proof is Theorem~\ref{thm:appendixmain} stated below, which solves a more general problem than is necessary for Lemma~\ref{lem:linear_sampling_simple}:
\begin{itemize}
	\item Qubits are generalized to qudits of dimension $d$.
	\item The starting state is a tensor product of mixed states, $\chi_{1}\otimes \cdots \otimes \chi_{n}$, rather than $H^{\otimes n} \ket{0^n}$. 
	\item There are operators $O_a$ for all qudits, and they need not have bounded norm. 
	\item The gates $U_1, \ldots, U_m$ may act on a set $S$ of qudits, and induce edges $\{i,j\}$ in $G$ for all $i, j \in S$. We do not require the product of the gates to be unitary.
\end{itemize}
We recall the problem from the introduction.
\begin{repproblem}{prob:forrsample}
	Given an initial state $\chi = \chi_{1} \otimes \cdots \otimes \chi_{n}$, a collection of diagonal gates $U_1, \ldots, U_m$, and single-qudit operators $\{ O_a \}_{a\in[n]}$ on the qudits, sample a string $x$ from the distribution $P$ where $P(x) \propto \bra{x} O U \chi U^{\dag} O^{\dag} \ket{x}$ and 
	\begin{align*}
		U &= \prod_{j=1}^{m} U_j, & O &= \bigotimes_{i=1}^{n} O_i.
	\end{align*} 
\end{repproblem}
As before, there is graph $G$ associated with an instance of this problem. 
\begin{dfn}
	We define the \emph{connectivity graph} $G = (V, E)$ where $V$ is the set of all qudits and
	\begin{align*}
	E &:= \{ (i,j) \in V \times V : \text{some $U_k$ acts on both $i$ and $j$} \}.
	\end{align*}
\end{dfn}

Our main result for this appendix is the following improvement of Lemma~\ref{lemma:sampling}:
\begin{theorem}
    \label{thm:appendixmain}
	Given an instance of Problem~\ref{prob:forrsample}, and a rooted tree decomposition $T$ of width $w$ for the connectivity graph $G=(V,E)$, there is a classical algorithm sampling a solution in time $O(nd^{2w} + m d^{w} + \ell)$ where $\ell$ is the length of the input, $m$ is the number of gates, and $n$ is the number of qudits.
\end{theorem}

Before we go any further, we note that the algorithm shares some similarities with a result of Gosset et al. \cite{kerzner}. That result describes how to use a tree decomposition to help sample Pauli measurements on a graph state. Since a graph state is just the tensor product state $\ket{+}^{\otimes n}$ with diagonal CZ gates applied on the edges, the \emph{graph state sampling} problem they introduce is actually a special case of Problem~\ref{prob:forrsample}. However, Theorem~\ref{thm:appendixmain} gives a \emph{much} worse bound (exponential in treewidth, rather than polynomial) because it does not assume the state is Clifford. 

We divide the proof into three lemmas. First, we use a tree decomposition for $G$ to build a tensor network for $O U \chi U^{\dag} O^{\dag}$, and prove correctness of the network in Lemma~\ref{lem:networkinvariant}. Second, we present an algorithm in Lemma~\ref{lem:algorithm} to sample measurement outcomes from the tensor network, exploiting its tree structure. Third, we separate out the runtime analysis of the algorithm into Lemma~\ref{lem:runtime}.

We begin with the tensor network, $\mathcal{N}$, for the state $O U \chi U^{\dag} O^{\dag}$. Actually, it is easier if $O U \chi U^{\dag} O^{\dag}$ is a vector, so let $\opvec$ be the linear map which flattens operators to vectors with respect to the standard basis:
$$
\opvec( \ket{i} \bra{j} ) = \ket{i} \ket{j} \text{ for all $i, j$.}
$$
That is, operators on a Hilbert space $\mathcal{H}_1 \otimes \cdots \otimes \mathcal{H}_k$ map to vectors in $\bigotimes_{i=1}^{k} (\mathcal{H}_i \otimes \mathcal{H}_i)$. Then the goal is for $\mathcal{N}$ to be a tensor network for
\begin{equation}
\opvec(O U \chi U^{\dag} O^{\dag}) = (O \otimes \overline{O})(U \otimes \overline{U}) \opvec(\chi) = \bigotimes_i (O_i \otimes \overline{O_i}) \prod_j (U_j \otimes \overline{U_j}) \bigotimes_k \opvec \chi_k. \label{eqn:network}
\end{equation}
For example, one possible tensor network for $\opvec(O U \chi U^{\dag} O^{\dag})$ looks like a normal circuit: an initial state followed by a sequence of \emph{gate tensors} $U_j \otimes \overline{U_j}$ for each gate $U_j$, and finally $O_i \otimes \overline{O_i}$ \emph{operator tensors} for the operators. The wires are $d^2$-dimensional, the states are flattened density matrices, gates and operators are doubled as $U_j \otimes \overline{U_j}$ or $O_i \otimes \overline{O_i}$, but otherwise it is just like a normal circuit. 

Of course, we could readily simulate a circuit on a mixed state without the formalism of $\opvec$ and tensor networks, but we will also be using \emph{merge tensors} $M_b$ to combine two wires representing the same qudit $b$. If $\opvec(\chi_b) = \sum_{i} \alpha_i \ket{i}$ (note: this is a classical basis of size $d^2$) then we define the merge tensor $M_b$ as 
$$
M_b := \sum_{i : \alpha_i \neq 0} \alpha_i^{-1} \ket{i} (\bra{i} \otimes \bra{i}).
$$
We need three key properties of these tensors. See Figure~\ref{fig:identities} for a diagram of the first two. 

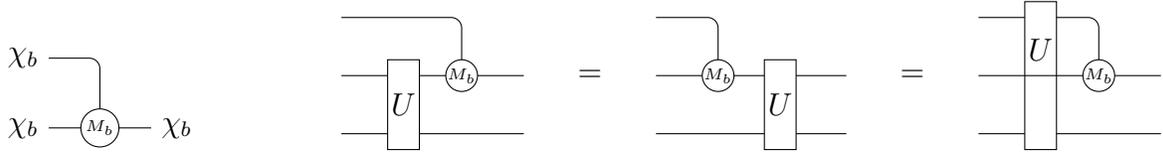
\begin{figure}
	\centering
		\begin{tikzpicture}[scale=1.200000,x=1pt,y=1pt]
\filldraw[color=white] (0.0, -11.0) rectangle (72.0, 33.0);
\draw[color=black, rounded corners] (12.500000,22.0) -- (36.0,22.0) -- (36,0);
\draw[color=black] (12.500000,0.0) -- (59.500000,0.0);
\draw[color=black] (20.0,22.0) node[fill=white,left,minimum height=22.0pt,minimum width=15.0pt,inner sep=0pt] {\phantom{$\chi_b$}};
\draw[color=black] (20.0,22.0) node[left] {$\chi_b$};
\draw[color=black] (20.0,0.0) node[fill=white,left,minimum height=22.0pt,minimum width=15.0pt,inner sep=0pt] {\phantom{$\chi_b$}};
\draw[color=black] (20.0,0.0) node[left] {$\chi_b$};
\begin{scope}
\draw[fill=white] (36.0, 0.0) circle(6.0pt);
\clip (36.0, 0.0) circle(6.0pt);
\draw (36.0, 0.0) node {\tiny $M_b$};
\end{scope}
\draw[color=black] (52.0,0.0) node[fill=white,right,minimum height=22.0pt,minimum width=15.0pt,inner sep=0pt] {\phantom{$\chi_b$}};
\draw[color=black] (52.0,0.0) node[right] {$\chi_b$};
\end{tikzpicture}
		\vspace{1em}\hspace{2em}
		\begin{tikzpicture}[scale=1,x=1pt,y=1pt]
\filldraw[color=white] (0, -11) rectangle (335, 55);
\draw[color=black,rounded corners] (12.500000,44) -- (58,44) -- (58,22);
\draw[color=black,rounded corners] (131.500000,44) -- (155,44) -- (155,22);
\draw[color=black,rounded corners] (253.500000,44) -- (299,44) -- (299,22);
\draw[color=black] (12.500000,22) -- (81.500000,22);
\draw[color=black] (131.500000,22) -- (203.500000,22);
\draw[color=black] (253.500000,22) -- (322.500000,22);
\draw[color=black] (12.500000,0) -- (81.500000,0);
\draw[color=black] (131.500000,0) -- (203.500000,0);
\draw[color=black] (253.500000,0) -- (322.500000,0);
\draw (36,22) -- (36,0);
\begin{scope}
\draw[fill=white] (36, 11) +(-45:8.485281pt and 24.041631pt) -- +(45:8.485281pt and 24.041631pt) -- +(135:8.485281pt and 24.041631pt) -- +(225:8.485281pt and 24.041631pt) -- cycle;
\clip (36, 11) +(-45:8.485281pt and 24.041631pt) -- +(45:8.485281pt and 24.041631pt) -- +(135:8.485281pt and 24.041631pt) -- +(225:8.485281pt and 24.041631pt) -- cycle;
\draw (36, 11) node {$U$};
\end{scope}
\begin{scope}
\draw[fill=white] (58, 22) circle(6pt);
\clip (58, 22) circle(6pt);
\draw (58, 22) node {\tiny $M_b$};
\end{scope}
\draw[fill=white,color=white] (99, -6) rectangle (114, 50);
\draw (106.500000, 22) node {$=$};
\begin{scope}
\draw[fill=white] (155, 22) circle(6pt);
\clip (155, 22) circle(6pt);
\draw (155, 22) node {\tiny $M_b$};
\end{scope}
\draw (178.500000,22) -- (178.500000,0);
\begin{scope}
\draw[fill=white] (178.500000, 11) +(-45:8.485281pt and 24.041631pt) -- +(45:8.485281pt and 24.041631pt) -- +(135:8.485281pt and 24.041631pt) -- +(225:8.485281pt and 24.041631pt) -- cycle;
\clip (178.500000, 11) +(-45:8.485281pt and 24.041631pt) -- +(45:8.485281pt and 24.041631pt) -- +(135:8.485281pt and 24.041631pt) -- +(225:8.485281pt and 24.041631pt) -- cycle;
\draw (178.500000, 11) node {$U$};
\end{scope}
\draw[fill=white,color=white] (221, -6) rectangle (236, 50);
\draw (228.500000, 22) node {$=$};
\draw (277,44) -- (277,0);
\begin{scope}
\draw[fill=white] (277, 22) +(-45:8.485281pt and 39.597980pt) -- +(45:8.485281pt and 39.597980pt) -- +(135:8.485281pt and 39.597980pt) -- +(225:8.485281pt and 39.597980pt) -- cycle;
\clip (277, 22) +(-45:8.485281pt and 39.597980pt) -- +(45:8.485281pt and 39.597980pt) -- +(135:8.485281pt and 39.597980pt) -- +(225:8.485281pt and 39.597980pt) -- cycle;
\draw (277, 32) node {$U$};
\end{scope}
\draw[color=black] (271, 22) -- (283, 22);
\begin{scope}
\draw[fill=white] (299, 22) circle(6pt);
\clip (299, 22) circle(6pt);
\draw (299, 22) node {\tiny $M_b$};
\end{scope}
\end{tikzpicture}
	\caption{Identities with the merge tensor. (1) Two copies of $\chi_b$ are merged into one. (2) A gate tensor for diagonal $U$ can be moved to any one of the three arms.}
	\label{fig:identities}
\end{figure}

\begin{prop}
	\label{prop:mergetensor}
	For any qudit $b$, the merge tensor $M_b$ has the following properties:
	\begin{enumerate}
	\item Applying $M_b$ to two copies of $\opvec(\chi_b)$ produces $\opvec(\chi_b)$. 
	$$M_b(\opvec(\chi_b) \otimes \opvec(\chi_b)) = \opvec(\chi_b).$$
	\item Any diagonal gate $U$ commutes through from either input of $M_b$ to the output.
	$$(M_b \otimes I) (I \otimes U) = U (M_b \otimes I).$$
	\item $M_b(\ket{u} \otimes \ket{v}) = \ket{c} \circ \ket{u} \circ \ket{v}$ where $c$ is a constant vector and $\circ$ denotes entrywise multiplication. 
	\end{enumerate}
\end{prop}
\begin{proof}
	For the first property, we have the calculation
	\begin{align*}
	M_b(\opvec(\chi_b) \otimes \opvec(\chi_b)) &= \sum_{i : \alpha_i \neq 0} \alpha_i^{-1} (\bra{i} \opvec(\chi_b))^2 \ket{i}
	 = \sum_{i : \alpha_i \neq 0} \alpha_i \ket{i} = \opvec(\chi_b).
	\end{align*}
	Now consider a diagonal gate $U$ and decompose it as $U = \sum_{j} \ket{j} \bra{j} \otimes U_j$. Then 
	\begin{align*}
	(M_b \otimes I) (I \otimes U) &= \left( \sum_{i : \alpha_i \neq 0} \alpha_i^{-1} \ket{i} (\bra{i} \otimes \bra{i}) \otimes I \right) \left( \sum_{j} I \otimes \ket{j} \bra{j} \otimes U_j \right) \\
	 &= \left( \sum_{i : \alpha_i \neq 0} \sum_j \delta_{ij} \alpha_i^{-1} \ket{i} (\bra{i} \otimes \bra{j}) \otimes AUj \right) \\
	 &= \left( \sum_{j} \ket{j} \bra{j} \otimes U_j \right) \left( \sum_{i : \alpha_i \neq 0} \alpha_i^{-1} \ket{i} (\bra{i} \otimes \bra{i}) \right) \\
	 &= U (M_b \otimes I).
	\end{align*}
	The final property is straightforward when we take $\ket{c} = \sum_{i}c_i \ket{i}$ where
	$$
	c_i = \begin{cases} \alpha_i^{-1} & \text{if $\alpha_i \neq 0$} \\
	0 & \text{if $\alpha_i= 0$.}
	\end{cases}
	$$
\end{proof}

Finally, we can define the tensor network $\mathcal{N}$. We construct $\mathcal{N}$ from the tree decomposition $T$ by creating a subnetwork $\mathcal{N}_B$ for each node $B$ in $T$, and connecting subnetworks of neighbouring nodes in $T$. For organizational purposes, each wire in the network is directed, and labeled by a qudit. Then for each child $C$ of $B$ in the tree, $\mathcal{N}_B$ has incoming wires from $\mathcal{N}_C$ for each qudit in $B \cap C$. There are outgoing wires for each qudit in $B$; the parent $A$ of $B$ takes wires $A \cap B$ and the remaining $B \backslash A$ are free. Since each vertex ``leaves'' a tree decomposition (appears in some node, but not its parent) exactly once, $\mathcal{N}$ has one free leg for each vertex (qudit). See Figure~\ref{fig:network1} for an example tensor network. 

Within each subnetwork $\mathcal{N}_B$, there are four kinds of tensors that can appear: state tensors (i.e., $\opvec(\chi_i)$), gate tensors (i.e., $U_j \otimes \overline{U_j}$), operator tensors (i.e., $O_i \otimes \overline{O_i}$), and merge tensors ($M_b$). These are split across four stages appearing in the following order. 
\begin{description}	
	\item[Merge stage:] We have incoming wires from child subnetworks $\mathcal{N}_{C_i}$, and there may be more than one wire for some qudit $b \in B$. Use the merge tensor $M_b$ to combine wires for $b$ in pairs until there is only one wire for qudit $b$. 
	\item[Introduction stage:] If some qudit $b$ in $B$ does not appear in any child of $B$, then add the tensor $\opvec(\chi_b)$, so there is a wire for that qudit leading into the next stage. 
	\item[Gate stage:] At this point, there is exactly one incoming wire per qudit in $B$. Apply gate tensors for all gates which (i) act only on qudits in $B$, and (ii) touch some qudit in $B \backslash A$ where $A$ is $B$'s parent (or $A = \emptyset$ if $B$ is the root). That is, apply $\prod_{j \in S} U_j \otimes \overline{U_j}$ for the prescribed set $S$.
	\item[Operator stage:] Wires for qudits in $A \cap B$ go to the parent; the remaining $B \backslash A$ qudits are due to be measured, so we apply operator tensors to these qudits. 
\end{description}

\begin{lemma}
	\label{lem:networkinvariant}
	The tensor network $\mathcal{N}$ contracts to $\opvec(O U \chi U^{\dag} O^{\dag})$. 
\end{lemma}
\begin{proof}
	We have already established that $\mathcal{N}$ has a free wire for each qudit, so the tensor network has the right dimension.  
	
	Next, note that the edges between subnetworks are directed towards the root of the tree, and edges within subnetworks are directed through the sequence of stages ($\to$ merge $\to$ introduction $\to$ gate $\to$ operator $\to$), so the full network is a directed acyclic graph. As a result, individual merge tensors are partially ordered. If there are any merge tensors in $\mathcal{N}$, then there is an ``earliest'' merge tensor $M_b$, preceded only by gate tensors and state tensors (operator tensors are last under this ordering, so they do not interfere). By the second property of Proposition~\ref{prop:mergetensor}, we can commute the gate tensors past $M_b$, until both input wires to $M_b$ come directly from state tensors. 
	
	Recall that all wires are labeled, and note that all the wires of a merge tensor $M_b$ are labeled $b$ in the original network. This is preserved when we commute gates through with Proposition~\ref{prop:mergetensor}. It follows that the inputs to $M_b$ are labeled $b$, and therefore are $\opvec(\chi_b)$ state tensors. By the first property of Proposition~\ref{prop:mergetensor}, we replace $M_b$ and its inputs with $\opvec(\chi_b)$. Repeat the whole process with the next earliest merge tensor, until we eliminate all merge tensors. 
	
	The remaining network has only state tensors, gate tensors, and operator tensors, but it contracts to the same tensor as the original network, $\mathcal{N}$. The new network is essentially a circuit: we can follow each output wire back through the circuit to the input, a state tensor. There is a free wire and operator tensor for each qudit $i$ at the \emph{end} of the circuit, so there must be a matching state tensor $\opvec(\chi_i)$ at the \emph{start} of the circuit. It only remains to check that we have the right set of gate tensors. Since the gate tensors commute, it will follow that the network is correct by (\ref{eqn:network}). 
	
	Let $U_k$ be an arbitrary gate, on some subset of qudits $S \subseteq V$. For each $i$, let $T_i$ be the non-empty subtree of $T$ with all nodes containing $i$. Note that any $i, j \in S$ define an edge in $G$, so $i,j$ appear together somewhere in $T$, i.e., $T_i \cap T_j \neq \emptyset$. The intersection of a collection of pairwise intersecting subtrees is non-empty, so there is some highest node $B$ in $\bigcap_{s \in S} T_s$. We conclude that the gate tensor $U_k \otimes \overline{U_k}$ appears in the gate stage of $\mathcal{N}_B$ because it contains all the necessary vertices $S$ (it is in $T_i$ for all $i \in S$), but its parent does not (or $B$ would not be highest). 
	
	Finally, we claim this occurrence is unique. $U_k$ cannot appear in some $B'$ below $B$ because there is no vertex $b$ that leaves (i.e., isn't in the parent of $B'$) since $b$ is in $B$, $B'$, and all the nodes in between by tree decomposition definitions. On the other hand, $U_k$ cannot appear outside the subtree rooted at $B$, since there \emph{is} some vertex $b \in S$ missing from $B$'s parent, and hence in the rest of the tree.
\end{proof}
We encourage the reader to work through the steps on the example in Figure~\ref{fig:network1} to help visualize the process. 

It is perhaps not surprising that a tensor network $\mathcal{N}$ derived from a bounded-treewidth graph can be sampled efficiently, but we could not find a blackbox way to get the runtime we want from existing algorithms.\footnote{We depend too much on the gates being diagonal, among other things.} We describe a new contraction-based algorithm, while arguing its correctness, in Lemma~\ref{lem:algorithm}, and then analyze its running time in Lemma~\ref{lem:runtime}. 

The algorithm works with subtrees of $T$ and corresponding subnetworks of $\mathcal{N}$, so let us introduce some notation for this. For a node $B$, let $T_B$ be the subtree of $T$ rooted at $B$. Given an arbitrary subtree $T'$ of $T$, let $\mathcal{N}(T')$ be the restriction of $\mathcal{N}$ to nodes appearing in $T'$. 

\begin{lemma}
	\label{lem:algorithm}
	There is a classical algorithm (described within) which solves Problem~\ref{prob:forrsample} given an instance and an accompanying tree decomposition $T$.
\end{lemma}
\begin{proof}
    Assume the tree decomposition $T$ is rooted and ordered, and we have constructed the tensor network $\mathcal{N}$. The algorithm traverses $T$ (and the corresponding pieces of the network) twice: first to compute tensors $\rho_B$ for all nodes $B$, then again to compute tensors $\sigma_B$, $\Pi_B$, and the sampled outcomes $\ket{x_b}$.
	\begin{enumerate}
		\item Tensor $\rho_{B}$ is the contraction of $\mathcal{N}(T_B)$. This represents the initial state of the subtree network before any sampling. 
		\item Tensor $\sigma_{B}$ is the contraction of $\mathcal{N}(T_B)$ with the sampled outcomes $\opvec(\ket{x_b} \bra{x_b})$ for all qudits measured in the subtree $T_B$. This represents the final state of the subtree as the second traversal leaves it.
		\item Tensor $\Pi_B$ is the contraction of $\mathcal{N}(T \backslash T_B)$ with sampled outcomes $\opvec(\ket{x_b} \bra{x_b})$ for all qudits measured in nodes visited before $B$ in the second traversal. $\Pi_B$ represents the rest of the network during the (second) traversal of $T_B$: the contraction everything outside $T_B$, postselected on the current sampled outcomes at the moment the traversal descends into $B$ and $T_B$. 
	\end{enumerate}
	These three tensors are explicit, i.e., stored as vectors of complex numbers at some point in the algorithm. As an \emph{essential} performance optimization, we restrict $\rho_{B}$ to qudits $B$, and restrict $\sigma_{B}$ and $\Pi_{B}$ to qudits $A \cap B$ where $A$ is the parent of $B$---we trace out all other qudits in these tensors. Since tracing out a qudit is itself a tensor contraction, and the order of contractions in a tensor network does not affect the result, we can reason about the \emph{unrestricted} versions of $\rho_{B}$, $\sigma_{B}$, and $\Pi_{B}$, with the understanding that we still have to trace out the qudits at the end. 
	
	In the first traversal, the algorithm computes $\rho_{B}$ for each $B$ from the bottom up. We simply take $\rho_{C_1}, \ldots, \rho_{C_k}$ for $B$'s children $C_1, \ldots, C_k$ (if they exist), contract them with $\mathcal{N}_B$, call the result $\rho_B$ and pass it up the tree.

	The second traversal works from the top down. See Figure~\ref{fig:complicatedtraversal} for a diagram of the steps in node $B$. In this pass, we assume the algorithm gets $\Pi_B$ from the parent when it first visits $B$, and returns $\sigma_B$ to the parent after traversing the subtree ($T_B$). The base case is the root $R$, where $\Pi_{R}$ is an empty tensor with no wires. Then for an arbitrary node $B$ in the tree, we notice that $\Pi_{B}$ and $\rho_{B}$ represent two halves of the network, and both are available to the algorithm when it first arrives at $B$. Thus, the algorithm contracts $\rho_{B}$ and $\Pi_B$ to get $\mathcal{N}$ (restricted to qudits $B$), and this is enough to sample measurement outcomes for qudits $B \backslash A$ in $B$ which do not connect to the parent, $A$. 
	
	Next, the algorithm visits the children $C_1, \ldots, C_k$ of $B$. Since we need to provide $\Pi_{C_i}$ to traverse $C_i$, we contract $\Pi_{B}$, $\mathcal{N}_B$, projectors onto the measurement outcomes for $B$, and either $\rho_{C_j}$ (if $j > i$) or $\sigma_{C_j}$ (if $j < i$) for each sibling of $C_i$ to get $\Pi_{C_i}$. Then the algorithm recursively traverses $T_{C_i}$, samples measurement outcomes within the subtree, and returns $\sigma_{C_i}$.
	
	Finally, once we have traversed all the subtrees of $B$, we visit $B$ once more to compute $\sigma_{B}$ by contracting $\mathcal{N}_B$, projections onto the sampled outcomes for qudits measured in $B$, and $\sigma_{C_i}$ for each child of $B$. The result is $\sigma_{B}$, and we pass it to the parent. 
	
	Let us consider correctness. Consider an arbitrary qudit $b$ (measured in some node $B$), and suppose inductively that the qudits sampled before it were sampled with the correct distribution. The algorithm samples $b$ by contracting $\mathcal{N}_B$ with $\Pi_B$ and possibly $\rho_{C_1}, \ldots, \rho_{C_k}$. We have already argued that the tensors $\Pi_B$, $\rho_{C_1}, \ldots, \rho_{C_k}$ represent the contractions of $\mathcal{N}(T\backslash T_B)$, $\mathcal{N}(T_{C_1}), \ldots, \mathcal{N}(T_{C_k})$ claimed at the start of this proof, so the main contraction gives the state $O U \chi U^{\dag} O^{\dag}$ post-selected on the measurement outcomes so far. Unless this state is zero\footnote{If the post-selected state is zero then $O U \chi U^{\dag} O^{\dag}$ must have been zero from the start, since we have only post-selected on outcomes with positive probability.}, the algorithm can renormalize the state and sample $b$ from it. Hence, $b$ has the correct distribution, and therefore the algorithm correctly samples. 
\end{proof}

A significant portion of the runtime analysis depends on the cost of contracting tensors together. We will need the following fact about tensor contraction.
\begin{fact}
	\label{fact:basiccontraction}
	Given two tensors $X$ and $Y$ explicitly, we can combine them in $O(D^{c + f})$ arithmetic ops where $c$ is the number of wires between $X$ and $Y$, $f$ is the number of wires in the combined tensor, and $D$ is the dimension of each wire. 
\end{fact}
We will take $D = d^2$ when we apply this fact, since recall that all our wires have dimension $d^2$. Note that we are counting the number of \emph{arithmetic operations} (arithmetic ops), meaning complex additions, subtractions, multiplications and divisions.

\begin{lemma}
	\label{lem:runtime}
	Given gates $U_1, \ldots, U_m$ (as dense vectors of their diagonal elements) on a system of $n$ qudits, operators $\{ O_i \}_{i \in [n]}$ (as dense matrices), and a tree decomposition $T$ of width $w$, then the algorithm described above (Lemma~\ref{lem:algorithm}) is dominated by the cost of $O(nd^{2w} + m d^{w} + \ell)$ arithmetic ops, where $\ell$ is the length of the input. 
\end{lemma}
\begin{proof}
    We now show that we may assume $T$ is a binary tree with $O(n)$ nodes. First, we contract any edge where the parent contains all the vertices of the child. In the remaining tree, every edge has some qudit $b$ that appears in the child but not the parent. However, there is only \emph{one} such edge for each qudit $b$ (the least common ancestor of all nodes containing $b$) by tree decomposition axioms. Hence, there are at most $n$ edges and $n-1$ nodes after contraction. Next, we replace any node $B$ with three or more children $C_1, \ldots, C_k$ by a binary tree having $C_1, \ldots, C_k$ at the leaves and $B$ at all internal nodes, at most doubling the number of nodes. The resulting tree decomposition has the same treewidth as the original, at most $2n-2$ nodes.
	
	The main algorithm computes tensors $\rho_{B}$, $\sigma_{B}$, and $\Pi_{B}$ for each node $B$ in the tree decomposition. In Lemma~\ref{lem:algorithm}, we presented the operations on $\rho_{B}$, $\sigma_{B}$, and $\Pi_{B}$ abstractly as tensor contractions, mostly between $\mathcal{N}_B$ and $\rho$, $\sigma$, and $\Pi$ for various nodes. We explicitly store $\rho$, $\sigma$, and $\Pi$ tensors as vectors of complex numbers, but \emph{not} $\mathcal{N}_B$ because it would be too large (up to $w$ wires to the parent and $w$ from each child is up to $3w$). Instead, every time we contract tensors to $\mathcal{N}_B$, we expand it as a network of individual state, gate, operator, and merge tensors, and contract the network strategically.  

	Note that there are up to four tensor computations that involve $\mathcal{N}_B$, which compute $\rho_{B}$ (in the first traversal), $\Pi_{C_1}$, $\Pi_{C_2}$, $\sigma_{B}$. Each one contracts all the individual tensors of $\mathcal{N}_B$ (for essentially the same cost in each contraction). We analyze the cost of contraction for each stage in the forward and reverse directions (i.e., whether we're given a tensor on the input wires or output wires to contract with). 
	\begin{description}
		\item[Merge tensors:] We wish to combine tensors from the two children of $B$, namely $C_1$ and $C_2$, which overlap on qudits $B \cap (C_1 \cup C_2)$. We apply the merge tensor
		$M_B := \otimes_{b \in B \cap (C_1 \cup C_2)} M_b$ on the duplicates. Fortunately, each individual $M_b$ is an entry-wise product of the two input tensors and a constant vector (by  Proposition~\ref{prop:mergetensor}), and so is the full merge tensor $M_B$. On $w$ qudits, the vectors all have dimension $d^{2w}$, and therefore we use at most $O(d^{2w})$ arithmetic ops. It is similarly $O(d^{2w})$ ops to contract the tensor given a tensor for the output and one of the inputs. 
		
		We can only merge as many times as there are internal nodes of $T$, i.e., at most $O(n)$. Hence, $O(nd^{2w})$ arithmetic ops are spent on merge tensors. 
		\item[State tensors:] The tensor product of $\opvec(\chi_i)$ with an existing tensor $\opvec(\rho)$ takes $O(d^{2f})$ arithmetic ops (by Fact~\ref{fact:basiccontraction}) where $f$ is the final number of wires. If we build up a large product $\opvec(\chi_1) \otimes \cdot \otimes \opvec(\chi_k)$ iteratively then since the final number of wires \emph{increases} in each step, the number of arithmetic ops is at worst a geometric series,
		$$O(d^{2}) + O(d^4) + \cdots + O(d^{2w-4}) + O(d^{2w-2}) + O(d^{2w}) = O(d^{2w}).$$
		The series dominated by its largest term, which is at most $O(d^{2w})$. In the other direction, contracting state tensor $\opvec(\chi_i)$ with an arbitrary tensor is $O(d^{2f})$ (again by Fact~\ref{fact:basiccontraction}) where $f$ is the number of wires in the arbitrary tensor. Since this \emph{decreases} after the contraction, the number of arithmetic ops required for a sequence of these operations is again at worst a geometric series totalling $O(d^{2w})$. 
		
		The cost for the whole tree is $O(d^{2w})$ per node times $O(n)$ nodes: $O(nd^{2w})$ arithmetic ops. 
		\item[Gate tensors:] The gate tensors are linear maps $U_j \otimes \overline{U_j}$, where $U_j$ is diagonal. To update an explicit input tensor, it suffices to multiply entry-wise by the diagonal of $U_j \otimes \overline{U_j}$. Clearly this is $O(d^{2w})$ arithmetic ops, and it is similar in reverse. 
		
		Now let $U_B$ be the product of all the gates applied in $B$. Rather than apply gate tensors for the individual gates, we compute the product in time $O(d^{w})$ times the number of gates, and then use the gate tensor $U_{B} \otimes \overline{U_B}$ in time $O(d^{2w})$. This way there are $O(n)$ gate tensor contractions using $O(d^{2w})$ arithmetic ops each, plus $O(d^{w})$ operations per gate times $O(m)$ gates. Thus, gate tensors cost us $O(nd^{2w} + m d^{w})$. 
		\item[Operator tensors:] Unlike gates, operators may be arbitrary dense matrices, $O_i$. For this reason, it is cheaper to apply the operators one at a time rather than grouping them like gates. From Fact~\ref{fact:basiccontraction}, we need only $O(d^{2w+2}) = O(d^{2w})$ arithmetic ops to contract one wire of the operator tensor the main $w$ wire tensor. There are at most $O(n)$ operator tensors, one per qudit, so the total cost is $O(d^{2w} n)$ arithmetic ops across the tree.
	\end{description}

	There is one final step which does not involve $\mathcal{N}_B$: contracting $\rho_{B}$ and $\Pi_B$ and then sampling outcomes. The two tensors have $A \cap B$ qudits in common, and $B \backslash A$ wires are left at the end, so again the cost of the contraction is $O(d^{2w})$ arithmetic ops by Fact~\ref{fact:basiccontraction}. The result is a tensor for the qudits to be measured, and since we are measuring in the classical basis the distribution can be read off of the tensor. Hence we can sample in $O(d^w)$ time per node, which is negligible in comparison to contraction. Hence, the cost is $O(n d^{2w})$ arithmetic ops across the entire tree. 
\end{proof}

We present a way to tweak the algorithm to compute amplitudes as well, although it does not improve on the earlier algorithm given in Lemma~\ref{lemma:amplitudes}. 
\begin{corol}
    Given an initial pure state $\ket{\chi} = \ket{\chi_1} \otimes \cdots \otimes \ket{\chi_n}$, a collection of diagonal gates $U_1, \ldots, U_m$, single-qudit operators $\{ O_a \}_{a \in [n]}$, and classical state $\ket{x}$, there is a classical algorithm computing $\langle x | \alpha \rangle$ where
	$$
	\ket{\alpha} = \bigotimes_{i} O_i \prod_{j} U_j \ket{\chi}.
	$$
	The algorithm runs in time $O((n+m)d^{w} + \ell)$ where $\ell$ is the length of the input.
\end{corol}
\begin{proof}
    Re-use the same tensor network, but specializing it to pure states. That is, wires are dimension $d$, gate tensors become $U_j$, operator tensors become $O_i$, state tensors are $\ket{\chi_i}$. Merge tensors are still 
    $$
    M_b = \sum_{i : \beta_i \neq 0} \beta_i^{-1} \ket{i}(\bra{i} \otimes \bra{i}),
    $$
    but over $d$ dimensions for $i$, and where $\ket{\chi_b} = \sum_{i} \beta_i \ket{i}$. We claim that this tensor network contracts to the vector $\ket{\alpha}$, so if we attach $\bra{x_b}$ to the free wire $b$ for all qudits, then we can compute $\langle x | \alpha \rangle$.
    
    We leave it as an exercise to check that we can contract this network in such a way that all intermediate tensors have at most $O(d^{w})$ entries, and hence we have the desired runtime. 
\end{proof}

Recall that Theorem~\ref{thm:graph_based} partitions the graph into two halves, and splits the operators $O_a$ between them. We sample from one half, and evaluate an amplitude for the other half, so in each calculation, some of the operators are the identity.  Crucially, we get to drop vertices where $O_a = I$ from the graph, and the treewidth is computed for the remaining vertices. Since we have generalized the problem, we need to verify that this is still true: we can sample the easy qudits and reduce to an instance on the rest. However, we require the product of the gates to be unitary.
\begin{lemma}
	\label{lem:removeeasy}
	Given an instance of Problem~\ref{prob:forrsample} where the product $U$ is unitary. Let 
	$$
	V_{easy} = \{ a \in V : O_a = I \},
	$$
	be the set of vertices/qudits with trivial operators $O_a$. Then we can reduce to an instance on $V \backslash V_{easy}$ in linear time, i.e., sample outcomes for $V_{easy}$ and simplify the instance to vertices $V \backslash V_{easy}$. 
\end{lemma}
\begin{proof}
	Take an arbitrary qudit $i \in V_{easy}$. Ordinarily, the probability of measuring qudit $i$ in state $\ket{x_i} \bra{x_i}$ would be $\Tr( (\ket{x_i} \bra{x_i} \otimes I) O_i U \chi U^{\dag} O_i^{\dag})$, but since $O_i$ is the identity, we have
	$$
	\Tr((\ket{x_i} \bra{x_i} \otimes I) U \chi U^{\dag}) = \Tr(U^{\dag} U (\ket{x_i} \bra{x_i} \otimes I) \chi) = \Tr((\ket{x_i} \bra{x_i} \otimes I) \chi) = \bra{x_i} \chi_i \ket{x_i},
	$$
	where we used the fact that $U$ is unitary (so $U^{\dag} U = I$) and diagonal (so it commutes with other diagonal matrices, e.g., $\ket{x} \bra{x} \otimes I$). It is easy to sample an outcome $\ket{x_i}$ proportional to $\bra{x_i} \chi_i \ket{x_i}$; no calculation is necessary since the distribution is on the diagonal of $\chi_i$. Once we have $\ket{x_i}$, we are interested in the state $\Tr_{i}( (\ket{x_i} \bra{x_i} \otimes I)U \chi U^{\dag})$, where we have projected onto the outcome and traced out qudit $i$. Since each factor of $U$ commutes with the projector, this is equivalent to replacing each $U_j$ that touches qudit $i$ with $U_j' := \Tr_i((\ket{x_i}\bra{x_i} \otimes I)U_j)$. We leave it as an exercise to check that this is a new instance of Problem~\ref{prob:forrsample}. 

	Let us sample all qudits $V_{easy}$ first, then project onto $\ket{x_{V_{easy}}}$ and trace out $V_{easy}$. Sampling from a given discrete distribution on $d$ elements requires at most $O(d)$ arithmetic ops. Projecting onto a standard basis state $\ket{x_i}$ and tracing out simply restricts $U_j$ to a submatrix $U_j'$ with rows and columns matching $\ket{x_i}$. We pick out this submatrix in linear time, completing the reduction to an instance on $V \backslash V_{easy}$. 
\end{proof}

Finally, we recall Lemma~\ref{lem:linear_sampling_simple} from Section~\ref{sec:graph_based_proof}:
\begin{replemma}{lem:linear_sampling_simple}
There is a classical algorithm which samples a binary string $x\in \{0,1\}^n$ from the distribution $P(x)=|\langle x|\alpha\rangle|^2$ in time $O(n 4^w)$. Recall that $\ket{\alpha} = (O_A\otimes I_B)U_fH^{\otimes n}\ket{0^n}$, and $w$ is the treewidth of the graph restricted to qubits $A$ (i.e., to the qubits having non-trivial operators). We additionally assume that $U_f$ is unitary, and that we are given a tree decomposition of size $O(n)$ and width $w$. 
\end{replemma}
\begin{proof}
    First, we use Lemma~\ref{lem:removeeasy} to sample the qubits from $B$ (i.e., qubits without operators) in linear time. The graph for the remaining qubits $A$ and gates has treewidth $w$, and we assume we are given a tree decomposition $T$ of size $O(n)$ for it. Note that since our gates are $2$-local, there can be at most $w^2$ gates per node of the tree decomposition, and thus $m = O(n w^2)$ gates total. Applying Theorem~\ref{thm:appendixmain} samples the remaining qubits in time 
    \[
    O(n d^{2w} + m d^{w}) = O(n 4^w + n w^2 2^w) = O(n 4^w).
    \]
\end{proof}

\section{Example Tensor Network \texorpdfstring{$\mathcal{N}$}{N}}
\label{app:example}

Consider a circuit on eight qudits labeled $a$ through $h$, with diagonal gates on the following subsets of qudits:
$$
\{ a, b\}, \{ a, c \}, \{b \}, \{ b, c \}, \{ b, e \}, \{ b, f \}, \{ b, g \}, \{ b, h \}, \{c, d, e\}, \{ d, e\}, \{ e, g, h \}, \{f, g \}.
$$
We wish to initialize the qudits in state $\chi = \chi_a \otimes \cdots \otimes \chi_{h}$, apply these gates, and measure qudits after applying operators $O_{a}, \ldots, O_{h}$. 

On closer inspection, we notice that $\{ d, e\}$ is a subset of $\{ c, d, e \}$, so the two gates may be combined. Similarly, $\{ b\}$ can be combined with any of the other gates on $b$. Also, suppose $O_h = I$, so we can remove $h$ entirely: sample a classical outcome for $h$ based on $\chi_{h}$ and restrict $U_{egh}$ to that outcome, creating a new gate $U_{eg}$. The other operators, $O_a, \ldots, O_g$, are all nontrivial, so at this point we construct the graph $G$ in Figure~\ref{fig:graph}. The figure shows the remaining qudits, and edges between pairs of qudits that appear together in some gate.

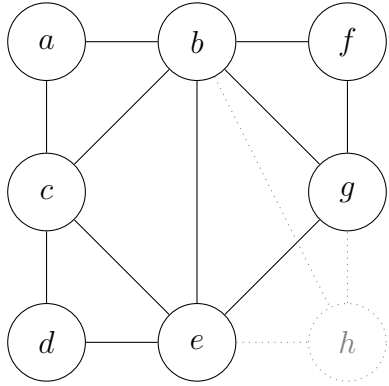
\begin{figure}
\centering
\begin{tikzpicture}
\node[state] (a) at (-2,2) {$a$};
\node[state] (b) at (0,2) {$b$};
\node[state] (c) at (-2,0) {$c$};
\node[state] (d) at (-2,-2) {$d$};
\node[state] (e) at (0,-2) {$e$};
\node[state] (f) at (2,2) {$f$};
\node[state] (g) at (2,0) {$g$};
\node[state,dotted,gray] (h) at (2,-2) {$\color{gray}{h}$};

\draw (a) -- (b) -- (c) -- (a); 
\draw (c) -- (d) -- (e) -- (c);
\draw (b) -- (f) -- (g) -- (b);
\draw (b) -- (e) -- (g);

\draw[dotted,gray] (b) -- (h);

\draw[dotted,gray] (e) -- (h) -- (g);
\end{tikzpicture}
\caption{A graph on qudits $a$ through $g$, where there is an edge $\{ u, v\}$ whenever some gate acts on $u$ and $v$. The grayed out node $h$ and edges $\{ e, h \}$ and $\{ g, h\}$ show what the graph would be if we had not removed $h$, i.e., if $h$ was measured in a different basis.}
\label{fig:graph}
\end{figure}

\begin{figure}
\centering
\begin{quantikz}[transparent,row sep=1mm,column sep=3mm]
	\lstick{$\chi_a$} \qw & 
	\gate[wires=2]{U_{ab}} & \qw & 
	\qw & \qw & 
	\qw & \qw & 
	\qw & \qw & 
	\qw & \qw &
	\gate[wires=3,label style={yshift=5mm}]{U_{ac}} & \gate{O_a} & \qw \\
	\lstick{$\chi_b$} \qw &
	\qw & \qw &
	\gate[wires=2]{U_{bc}} & \qw & 
	\gate[wires=4]{U_{be}} & \qw &
	\gate[wires=5,label style={yshift=5mm}]{U_{bf}} & \qw &
	\gate[wires=6]{U_{bg}} & \qw &
	\linethrough & \gate{O_b} & \qw \\
	\lstick{$\chi_c$} \qw & 
	\gate[wires=3]{U_{cde}} & \qw &
	\qw & \qw &
	\linethrough & \qw &
	\linethrough & \qw &
	\linethrough & \qw &
	\qw & \gate{O_c} & \qw \\
	\lstick{$\chi_d$} \qw &
	\qw & \qw &
	\qw & \qw & 
	\linethrough & \qw &
	\linethrough & \qw &
	\linethrough & \qw &
	\qw & \gate{O_d} & \qw \\
	\lstick{$\chi_e$} \qw &
	\qw & \qw &
	\gate[wires=3,label style={yshift=5mm}]{U_{eg}} & \qw & 
	\qw & \qw &
	\linethrough & \qw &
	\linethrough & \qw &
	\qw & \gate{O_e} & \qw \\
	\lstick{$\chi_f$} \qw & 
	\gate[wires=2]{U_{fg}} & \qw &
	\linethrough & \qw &
	\qw & \qw &
	\qw & \qw & 
	\linethrough & \qw &
	\qw & \gate{O_f} & \qw \\
	\lstick{$\chi_g$} \qw &
	\qw & \qw &
	\qw & \qw & 
	\qw & \qw & 
	\qw & \qw & 
	\qw & \qw &
	\qw & \gate{O_g} & \qw \\
\end{quantikz}	
\caption{A straightforward circuit for the example.}
\label{fig:circuit}
\end{figure}
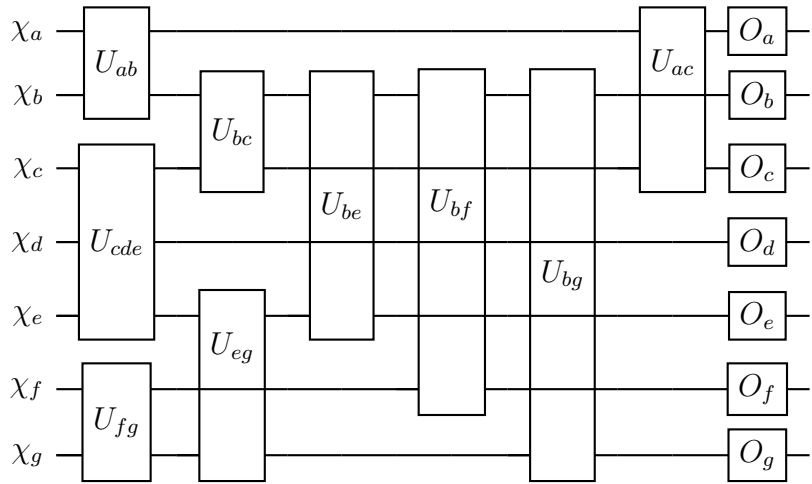

At this point, we \emph{could} sample outcomes for the rest of the qudits with a na\"{i}ve simulation of a circuit like in Figure~\ref{fig:circuit}. That is, construct the full density matrix for $\chi_1 \otimes \cdots \otimes \chi_g$, apply the gates one at a time, and then measure the qudits $a$ through $g$. The problem is that the density matrix grows exponentially with the number of qudits, so it is extremely inefficient to store all the qudits. 

Depending on the graph, we may be able to avoid storing all the qudits at once, by carefully sequencing the gates, and scheduling the introduction/measurement of qudits. For example, separate the graph into two overlapping halves, $\{ a, b, c, d, e \}$ and $\{ b, e, f, g \}$. We can simulate the qudits and gates on $\{ a, b, c, d, e\}$, measure $\{a, c, d\}$ leaving the overlap, $\{ b, e \}$, and then introduce $\{ f, g \}$ and finish the circuit. The density matrix never has more than five qudits by this strategy, but otherwise does the same operations as a full seven qudit simulation. In fact, we can make do with only four qudits at a time, e.g., in groups like this
$$
\{ a, b, c\}, \{ b, c, d, e\}, \{ b, e, g\}, \{ b, f, g \}.
$$
This grouping is a tree decomposition of $G$, and in fact we can do even better with the width-$2$ decomposition in Figure~\ref{fig:treewidth}. 

We turn the tree decomposition into a tensor network, Figure~\ref{fig:network1}. Each node becomes a subnetwork, shown as a dotted box in the figure. Recall that we justify the correctness of the network by pushing merge gadgets to the beginning (left, in Figure~\ref{fig:network1}) of the network, and eliminating them. Without merge tensors, there is one line through the network for each qudit (like a circuit), and it is easy to check that all the gates appear. Correctness of $\mathcal{N}$ follows. 

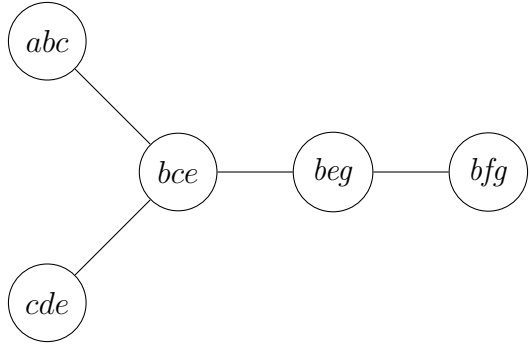
\begin{figure}
\begin{center}
\begin{tikzpicture}
	\node[state] (bce)  {$bce$};
	\node[state,above left=of bce] (abc) {$abc$};	
	\node[state,below left=of bce] (cde) {$cde$};
	\node[state,right=of bce] (beg) {$\mathit{beg}$};
	\node[state,right=of beg] (bfg) {$\mathit{bfg}$};
	
	\draw (abc) -- (bce) -- (cde); 
	\draw (bce) -- (beg) -- (bfg);
\end{tikzpicture}
\end{center}
\caption{A tree decomposition for $G$}
\label{fig:treewidth}
\end{figure}

\begin{figure}
\begin{center}
\input{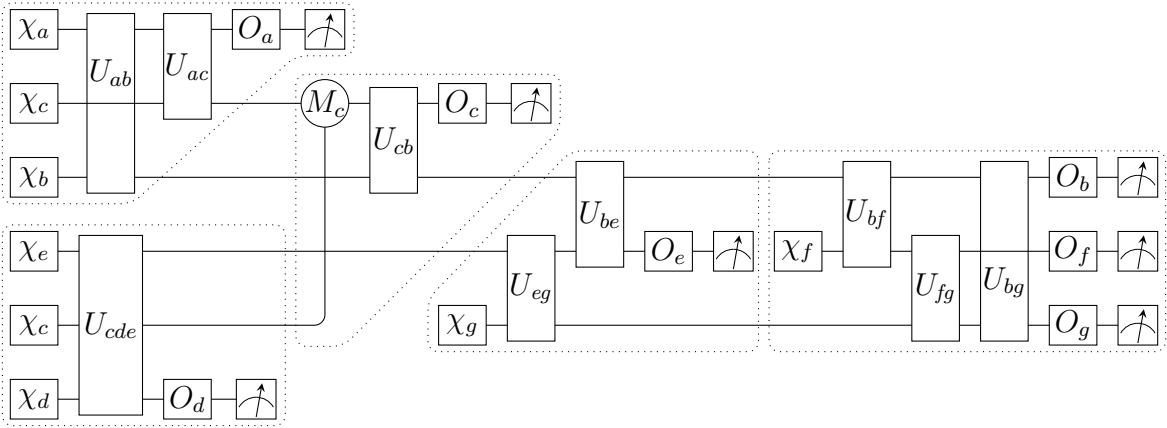}
\end{center}
\caption{A tensor network for the example. For brevity, $\chi_b$ represents the state tensor $\opvec \chi_b$, $U_{xy}$ represents a gate tensor $U_{xy} \otimes \overline{U_{xy}}$, and $O_{z}$ represents the operator tensor $O_{z} \otimes \overline{O_z}$.}
\label{fig:network1}
\end{figure}

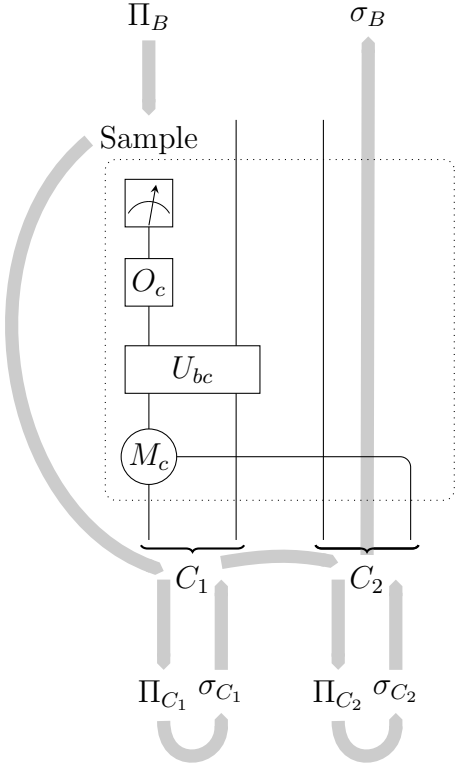
\begin{figure}
	\begin{center}
		 \begin{tikzpicture}[scale=1.500000,x=1pt,y=1pt]
	\coordinate  (BTOP) at (-44,30);
	\coordinate  (ETOP) at (-22,30);

	\coordinate  (BBOT) at (-44,-76);
	\coordinate  (C1BOT) at (-66,-76);
	\coordinate  (EBOT) at (-22,-76);
	\coordinate  (C2BOT) at (-0,-76);
										
	\coordinate (METER) at (-66,9);

	\node [above=5mm of METER] (SAMPLE) {Sample};
	\node [above=of SAMPLE] (START) {$\Pi_B$};
	
	\draw [decorate,decoration = {brace,raise=2pt},thick] (-42,-76) -- node[below=2mm] (C1) {$C_1$} (-68,-76);
	\draw [decorate,decoration = {brace,raise=2pt},thick] (2,-76) -- node[below=2mm] (C2) {$C_2$} (-24,-76);
	
	\node[below=12mm of C1.west] (C1OUT) {$\Pi_{C_1}$};
	\node[below=12mm of C1.east] (C1IN) {$\sigma_{C_1}$};
	\node[below=12mm of C2.west] (C2OUT) {$\Pi_{C_2}$};
	\node[below=12mm of C2.east] (C2IN) {$\sigma_{C_2}$};
	
	\draw (START) edge[-Triangle Cap,line width=5pt,draw=black!20] (SAMPLE);
	\draw (SAMPLE.west) edge[bend right=60,-Triangle Cap,line width=5pt,draw=black!20] (C1);
	\draw (C1.west) edge[-Triangle Cap,line width=5pt,draw=black!20] (C1OUT);
	\draw (C1IN) edge[-Triangle Cap,line width=5pt,draw=black!20] (C1.east);
	\draw (C1.north east) edge[-Triangle Cap,line width=5pt,draw=black!20,bend left=10,transform canvas={yshift=-1mm}] (C2.north west);
	\draw (C2.west) edge[-Triangle Cap,line width=5pt,draw=black!20] (C2OUT);
	\draw (C2IN) edge[-Triangle Cap,line width=5pt,draw=black!20] (C2.east);
	
	\draw (C1OUT.south) edge[-Triangle Cap,line width=5pt,draw=black!20,out=270,in=270,looseness=2.2] (C1IN.south);
	\draw (C2OUT.south) edge[-Triangle Cap,line width=5pt,draw=black!20,out=270,in=270,looseness=2.2] (C2IN.south);
	
	\node[right=21mm of START] (END) {$\sigma_{B}$};
	\draw (C2) edge[-Triangle Cap,line width=5pt,draw=black!20] (END);
	\draw [decorate,decoration = {brace,raise=2pt},thick] (2,-76) -- (-24,-76);
	
	
	\draw[dotted,rounded corners=4pt] (11, 20) rectangle (-77, -66);			
	\draw[color=black] (BTOP) -- (BBOT);
	\draw[color=black] (-66,9) -- (C1BOT);
	\draw[color=black] (ETOP) -- (EBOT);
	
	\draw[fill=white] (-72, -17) rectangle (-60, -5);
	\draw (-66,-11) node {$O_c$};
	
	\draw[fill=white] (-72, 3) rectangle (-60, 15);
	\draw[very thin] (-66, +9.400000) arc (90:150:6pt);
	\draw[very thin] (-66, +9.400000) arc (90:30:6pt);
	\draw[->,>=stealth] (-66, 3.400000) -- +(80:10.392305pt);
	\begin{scope}
		\draw[fill=white] (-55, -33) +(-45:24.041631pt and 8.485281pt) -- +(45:24.041631pt and 8.485281pt) -- +(135:24.041631pt and 8.485281pt) -- +(225:24.041631pt and 8.485281pt) -- cycle;
		\clip (-55, -33) +(-45:24.041631pt and 8.485281pt) -- +(45:24.041631pt and 8.485281pt) -- +(135:24.041631pt and 8.485281pt) -- +(225:24.041631pt and 8.485281pt) -- cycle;
		\draw (-55, -33) node {$\vspace{1em}U_{bc}$};
	\end{scope}
	\draw[rounded corners] (-66,-55) -- (0,-55) -- (C2BOT);
	\begin{scope}
		\draw[fill=white] (-66, -55) circle(7pt);
		\clip (-66, -55) circle(6pt);
		\draw (-66, -55) node {$M_c$};
	\end{scope}

\end{tikzpicture}
	\end{center}
	\caption{The second traversal at a node $B$: (1) $\Pi_B$ is given, (2) sample outcomes for qudits measured in the node, (3) compute $\Pi_{C_1}$, (4) recurse on $C_1$ and get $\sigma_{C_1}$, (5) compute $\Pi_{C_2}$, (6) recurse on $C_2$ and get $\sigma_{C_2}$, (7) compute $\sigma_B$ and return.}
	\label{fig:complicatedtraversal}
\end{figure}

Now we just need to contract the tensor network and sample outcomes, as described in Lemma~\ref{lem:algorithm}. Recall the tensors computed by the algorithm.
\begin{itemize}
	\item $\rho_{B}$ is the contraction of the subtree rooted at $B$ before any measurements, with all qudits traced out except those in $B$ itself.
	\item $\sigma_{B}$ is the contraction of the subtree rooted at $B$ after all measurements in the subtree have been made. 
	\item $\Pi_{B}$ is the contraction of the subtree ``above'' $B$, i.e., everything in the network except the subtree. All qudits not in $B$ are traced out or measured, depending on what parts of the tree the algorithm traversed before $B$. 
\end{itemize}
We compute all the $\rho_{B}$ in a bottom-up traversal, i.e., given $\rho_{C_1}$ and $\rho_{C_2}$ for the children, trace out any qudits not in $B$, contract the tensors with $\mathcal{N}$, and the result is $\rho_{B}$. 

The more interesting half is the second traversal from the top down. Figure~\ref{fig:complicatedtraversal} shows the flow of this phase of the algorithm. We are given $\Pi_{B}$, representing the state of the tree above us. Using $\rho_{B}$ and $\Pi_{B}$, we can get the state of the qudits in $B$, and sample measurement outcomes for them. Then we recurse on the children, but need to provide $\Pi_{C_1}$ and $\Pi_{C_2}$ respectively. So, we use $\Pi_B$, the outcomes we just sampled, and $\rho_{C_2}$ to compute $\Pi_{C_1}$, and recurse on the first child. We get $\sigma_{C_1}$ back, and combine with $\Pi_B$ and the samples to get $\Pi_{C_2}$. We recurse on $C_2$ and get $\sigma_{C_2}$ back. Then we use $\sigma_{C_1}$, $\sigma_{C_2}$, and the outcomes within $B$, to compute $\sigma_B$, the post-measurement state of the subtree (restricted to $B$). This completes the traversal of $B$ (and the subtree $T_B$).

\section{Simplifying level-2 QAOA variational optimization}
\label{app:eliminate_beta_2}

Let $G=(V,E)$ be a graph with $n$ vertices,
\[
B = \sum_{p\in V} X_p \quad \mbox{and} \quad
C=\sum_{\{p,q\}\in E} J_{p,q} Z_p Z_q.
\]
For level-2 QAOA, we would like to set $\beta_1, \beta_2, \gamma_1, \gamma_2$ to maximize the energy
$$
E(\beta_1, \beta_2, \gamma_1, \gamma_2) = \bra{+^n} W^\dag C W \ket{+^n},
$$
where $W= e^{-i\beta_2 B} e^{-i\gamma_2 C} e^{-i\beta_1 B} e^{-i\gamma_1 C}$ is the QAOA circuit with two entangling layers.  

Our goal for this section will be to show that if you fix the values of $\beta_1, \gamma_1$ and $\gamma_2$, then the value of $\beta_2$ that maximizes the energy $E$ can be computed analytically using the energy at three distinct values of $\beta_2$.

To show this, let us first focus on the contribution to the energy by single $Z_s Z_t$ term of $C$.  Since $e^{-i\beta_2 B}$ is a product of single-qubit operations, the conjugation of $Z_s Z_t$ by $e^{-i\beta_2 B}$ only depends on the gates on qubits $s$ and $t$.  We get
\begin{align*}
    e^{i\beta_2 B} Z_s Z_t e^{-i\beta_2 B} &= e^{i\beta_2 (X_s + X_t)} Z_s Z_t e^{-i\beta_2 (X_s + X_t)} = e^{2i \beta_2 (XI + IX)} ZZ \\
    &= \left(\cos^2(2\beta_2)II +i \sin (2\beta_2) \cos(2 \beta_2) (IX+XI) - \sin^2(2\beta_2)XX \right) ZZ \\
    &= \frac{1 + \cos(4\beta_2)}{2} ZZ + \frac{ \sin(4\beta_2)}{2} (ZY+YZ) + \frac{1 - \cos(4\beta_2)}{2} YY.
\end{align*}
By linearity, this implies that we can write the expectation as
$$
E(\beta_1, \beta_2, \gamma_1, \gamma_2) = a \cos(4 \beta_2) + b \sin(4\beta_2) + c
$$
for some real coefficients $a, b,$ and $c$ which are complicated functions of $\beta_1$, $\gamma_1$, and $\gamma_2$. We can compute these coefficients exactly using only three evaluations of the expectation function:
\begin{align*}
    E(\beta_1, \pi/8, \gamma_1, \gamma_2) &= b + c \\
    E(\beta_1, -\pi/8, \gamma_1, \gamma_2) &= - b + c \\
    E(\beta_1, 0, \gamma_1, \gamma_2) &= a + c.
\end{align*}
So, $a$, $b$, and $c$ can be calculated easily by solving the system of equations.  Once the values are known, we can calculate the max energy over $\beta_2$ using some elementary trigonometry:
$$
\max_{\beta_2} E(\beta_1, \beta_2, \gamma_1, \gamma_2) = c + \sqrt{a^2 + b^2}
$$
where the optimal $\beta_2$ is found by solving
\begin{align*}
    \tan(4 \beta_2) &= b/a \\
    a \cos(4 \beta_2) &\ge 0 \\
    b \sin(4 \beta_2) &\ge 0. 
\end{align*}
As a final remark, we note that this calculation also shows that the optimal value for $\beta_2$ is between $[0, \pi/2)$ since $\arctan(b/a) \in [0, 2\pi)$.

\section{Simulation of level-2 QAOA for low-degree graphs}
\label{all:brute_force}

Let $G=(V,E)$ be a graph with $n$ vertices,
\[
B = \sum_{p\in V} X_p \quad \mbox{and} \quad
C=\sum_{\{p,q\}\in E} J_{p,q} Z_p Z_q.
\]
Consider some fixed edge $\{s,t\}\in E$.
Our goal is to compute a quantum mean value 
\be
\label{QAOA2}
\mu = \la +^n|W^\dag Z_s Z_t W|+^n\ra,
\ee
where $W= e^{-i\beta_2 B} e^{-i\gamma_2 C} e^{-i\beta_1 B} e^{-i\gamma_1 C}$ is the QAOA circuit
with two entangling layers. 
In this section we give two classical algorithms $\calA'$ and $\calA''$ that compute the mean value $\mu$
using the Schr\"odinger and the Heisenberg pictures respectively.
The runtime of these algorithms scales exponentially with the size 
of certain local neighborhoods of the vertices $s$ and $t$ defined as follows.
Suppose $M\subseteq V$ is a subset of vertices and $r\ge 0$ is an integer.
Let $\calN_r(M)$ be the set of all vertices $i\in V$ such that the graph distance
between $i$ and $M$ is at most $r$. Let $n_r(M)=|\calN_r(M)|$.
The runtime of our algorithms depends on $n_2(s)$, $n_2(t)$, $n_1(s,t)$,
and $n_2(s,t)$. For example, 
suppose each vertex of $G$ has at most $d$ neighbors.  
Then $n_2(j)\le 1+d^2$ for any vertex $j$,
$n_1(s,t)\le 2d$, and $n_2(s,t)\le 2(d^2-d+1)$.
These bounds are tight if $G$ is a $d$-regular tree.

The unitary operator $W$ can be implemented by a quantum circuit with
two layers of single-qubit $X$-rotations and two entangling layers composed of 
nearest-neighbor $ZZ$-rotations. Such circuit can propagate information
from a qubit $j\in V$ only within the lightcone $\calN_2(j)$.
Accordingly, the mean value $\mu$ can be computed by restricting
the circuit $W$ onto the lightcone $\calN_2(s,t)$. 
The restricted version of $W$ acting on $n_2(s,t)$ qubits
contains $O(n_2(s,t))$ single-qubit gates
and two diagonal unitary operators  describing a
time evolution under the restricted version of $C$.
Such circuit can be simulated using the standard methods
in time $O(n_2(s,t) 2^{n_2(s,t)})$. Once the state $W|+^n\ra$
has been obtained, computing the expected value of $Z_s Z_t$
on this state takes time $O(2^{n_2(s,t)})$. Thus
a naive algorithm for computing $\mu$ 
that exploits the lightcone structure has runtime $O(n_2(s,t) 2^{n_2(s,t)})$.
Below we show how to improve upon this naive algorithm.

\begin{lemma}
\label{lemma:brute_force}
There exist classical algorithms $\calA'$ and $\calA''$ that compute
the  mean value Eq.~(\ref{QAOA2})  in time 

\be
\label{T1T2}
T' = O( n_2(s) 2^{n_2(s)} + n_2(t) 2^{n_2(t)})  
\quad \mbox{and} \quad T'' = O\left( n_1(s,t)4^{n_1(s,t)} + n_2(s,t) 3^{n_1(s,t)}\right)
\ee
respectively.
\end{lemma}
For simplicity here we assume that arithmetic operations with real numbers and evaluation of the standard trigonometric functions can be performed with an infinite precision in a unit time. Under this assumption the algorithms $\calA'$ and $\calA''$ compute the mean value $\mu$ exactly. 
For a given problem instance one  can estimate the runtimes 
using Eq.~(\ref{T1T2})  and pick the algorithm with the smallest runtime. 
In the rest of this section we prove Lemma~\ref{lemma:brute_force} by explicitly describing
the two algorithms. 

\noindent
{\em Algorithm $\calA'$:}
Consider a vertex $j\in \{s,t\}$.
Write $C=C_{\mathsf{loc}}(j) + C_{\mathsf{else}}(j)$, where
$C_{\mathsf{loc}}(j)$ is the sum of all terms $J_{p,q} Z_p Z_q$ with 
$p,q\in \calN_2(j)$ and $C_{\mathsf{else}}(j)$ is the sum of all remaining
terms.  Since all terms in $C$ pairwise commute, the standard lightcone argument shows that
\[
W^\dag Z_j W =W_{\mathsf{loc}}(j)^\dag Z_j W_{\mathsf{loc}}(j)
\]
where
\[
W_{\mathsf{loc}}(j) = e^{-i\beta_2 B}e^{-i\gamma_2 C_{\mathsf{loc}}(j) } e^{-i\beta_1 B} e^{-i\gamma_1 C_{\mathsf{loc}}(j) }.
\]
Note that $W_{\mathsf{loc}}(j)$ acts non-trivially only on $\calN_2(j)$. 
Thus
$W^\dag Z_j W|+^n\ra=|\psi_j\ra \otimes |+^{n-n_2(j)}\ra$,
where 
\[
|\psi_j\ra = W_{\mathsf{loc}}(j)^\dag Z_j W_{\mathsf{loc}}(j)|+^{n_2(j)}\ra
\]
is a state supported in $\calN_2(j)$.
One can compute a lookup table of the function $C_{\mathsf{loc}}(j) \, : \, \{0,1\}^{n_2(j)} \to \mathbb{R}$ in time
$O(n_2(j)2^{n_2(j)})$
by iterating over the set  $\{0,1\}^{n_2(j)}$
using Gray code. Indeed, since $C_{\mathsf{loc}}(j)$ is a quadratic function, updating its
value upon flipping a single bit of the input bit string takes time $O(n_2(j))$.
The quantum circuit $W_{\mathsf{loc}}(j)^\dag Z_j W_{\mathsf{loc}}(j)$
contains $O(n_2(j))$ single-qubit $X$-rotations
and $O(1)$ diagonal unitary operators implementing the time evolution under  $C_{\mathsf{loc}}(j)$.
The latter can be simulated in time $O(2^{n_2(j)})$ since the lookup table
of $C_{\mathsf{loc}}(j)$ is available.  Each single-qubit $X$-rotation can be simulated in time
$O(2^{n_2(j)})$ using a sparse matrix-vector multiplication. 
Thus the full vector specifying the state 
$|\psi_j\ra$ can be computed in time $O(n_2(j) 2^{n_2(j)})$.
We have
\[
\mu = \la +^n| (W^\dag Z_s W) (W^\dag Z_t W)|+^n\ra =  \la \psi'_s|\psi_t'\ra
\]
where
$|\psi'_j\ra$ is a state obtained from $|\psi_j\ra$ by projecting
all qubits $i\notin \calN_2(s)\cap \calN_2(t)$ onto the $|+\ra$ state.
Once a qubit is projected onto the $|+\ra$ state, it is discarded.
The state $|\psi_j'\ra$ can be computed in time $O(2^{n_2(j)})$, as follows
from the proposition below. 
\begin{prop}
Given a state $|\psi\ra\in ({\mathbb C^2})^{\otimes n}$
and an integer $k\in \{1,2,\ldots,n\}$, let
$|\psi'\ra\in ({\mathbb C^2})^{\otimes n-k}$ be
a (possibly unnormalized) state obtained from $|\psi\ra$ by
projecting the first $k$ qubits onto the $|+\ra$ state, that is,
\be
\label{psi2psi'}
\la y|\psi'\ra = 2^{-k/2} \sum_{x\in \{0,1\}^k} \la x,y|\psi\ra,
\qquad y\in \{0,1\}^{n-k}.
\ee
Then one can compute  $|\psi'\ra$ in time $O(2^n)$.
\end{prop}
\begin{proof}
If $k=1$ then compute all amplitudes $\la y|\psi'\ra$ by iterating
over $y\in \{0,1\}^{n-1}$ and using Eq.~(\ref{psi2psi'}).
This takes time $O(2^n)$. Applying the same step inductively
$k$ times takes time
\[
\sum_{i=1}^k O(2^{n-i}) \le O(2^n) \sum_{i=1}^\infty 2^{-i} = O(2^n).
\]
Here we noted that the number of qubits is reduced by one at each step.
\end{proof}
Computing the inner product $\mu=\la \psi'_s|\psi_t'\ra$ takes time $2^{|\calN_2(s)\cap \calN_2(t)|}$
which is negligible compared with the earlier steps of the algorithm.
The overall runtime of algorithm $\calA'$ is therefore $T'=O(n_2(s)2^{n_2(s)}+n_2(t)2^{n_2(t)})$.

\noindent
{\em Algorithm $\calA''$:}
Our starting point is the expression for  the mean value $\mu$ derived in Section~\ref{sec:QAOA_map},  namely
\be
\label{app_mu_eq1}
\mu= \mathrm{Re}
\left[ 2\mu(0000) + 4\mu(0010) + 4\mu(0001)+2\mu(0011)+2\mu(0110)+2\mu(0101)\right],
\ee
where
\be
\label{app_mu_eq2}
\mu(v)=\la v_1 v_2 |U|v_3 v_4\ra \cdot \la +^n|e^{i\gamma_1C} O_1(v)\otimes O_2(v)\otimes \cdots \otimes O_n(v)e^{-i\gamma_1 C}|+^n\ra,
\ee
\be
\label{app_mu_eq3}
U=e^{i\gamma_2 J_{s,t} Z\otimes Z} e^{i\beta_2(X\otimes I + I\otimes X)}
Z\otimes Z e^{-i\beta_2(X\otimes I + I\otimes X)}e^{-i\gamma_2 J_{s,t} Z\otimes Z},
\ee
and  $O_j(v)$ are single-qubit operators defined by
\be
\label{app_mu_eq4}
O_s(v) = e^{i\beta_1 X} |v_1\ra\la v_3| e^{-i \beta_1 X},
\ee
\be
\label{app_mu_eq5}
O_t(v) = e^{i\beta_1 X}|v_2\ra\la v_4| e^{-i \beta_1 X},
\ee
\be
\label{app_mu_eq6}
O_p(v) = e^{i\beta_1 X} e^{i\gamma_2 (J_{s,p} (-1)^{v_1}  + J_{t,p} (-1)^{v_2} - J_{s,p} (-1)^{v_3} - J_{t,p}(-1)^{v_4})Z}
e^{-i\beta_1 X}
\ee
for $p\notin \{s,t\}$. 
Below we use shorthand notations $\calN_r\equiv \calN_r(s,t)$ and $n_r=|\calN_r(s,t)|$,
where $r=1,2$. We begin by observing that 
\be
\label{Op(x)}
O_p(v)=I \quad \mbox{for all $p\notin \calN_1$}.
\ee
Indeed, if $p\notin \calN_1$ then
$J_{s,p}=J_{t,p}=0$ and $O_p(v)=e^{i\beta_1 X} e^{-i\beta_1 X}=I$,
see Eq.~(\ref{app_mu_eq6}). 
Fix a bit string $v\in \{0,1\}^4$ and
let $O_{\calN_1}$ be the tensor product of all operators $O_p(v)$ with $p\in \calN_1$.
From Eqs.~(\ref{app_mu_eq2},\ref{Op(x)})
one gets $\mu(v)=\la v_1 v_2|U|v_3v_4\ra \eta(v)$, where
\be
\label{eta(x)}
\eta(v) = \la +^{n}|e^{i\gamma_1 C}  (O_{\calN_1}\otimes I_{\mathsf{else}})
e^{-i\gamma_1 C} |+^{n}\ra.
\ee
It suffices to show that $\eta(v)$ can be computed in time $T''$ defined in Eq.~(\ref{T1T2}).
The standard lightcone argument shows that  $\eta(v)$ depends only on 
the subgraph $G_{\calN_2}$ induced by $\calN_2$ and the restricted cost function 
that includes only the terms $J_{p,q} Z_p Z_q$ with $p,q\in \calN_2$.
To ease the notations, below we ignore all the remaining terms and assume
that $G=G_{\calN_2}$, $n=n_2$, and $C=\sum_{p,q\in \calN_2} J_{p,q} Z_p Z_q$.
Let us write
\be
C=C_{\mathsf{loc}} + C_{\mathsf{ent}} + C_{\mathsf{else}},
\ee
where $C_{\mathsf{loc}}$ includes all terms $J_{p,q} Z_p Z_q$ with $p,q\in \calN_1$,
$C_{\mathsf{ent}}$  includes all terms $J_{p,q} Z_p Z_q$ with $p\in \calN_1$, $q\notin \calN_1$,
and  $C_{\mathsf{else}}$  includes all the remaining  terms.
Using the fact that all terms in $C$ pairwise commute and that
 $C_{\mathsf{else}}$ commutes with 
$O_{\calN_1}\otimes I_\mathsf{else}$, one can rewrite Eq.~(\ref{eta(x)}) as 
\be
\label{channel1}
\eta(v)=\mathrm{Tr}\left(O_{\calN_1}  e^{-i\gamma_1C_{\mathsf{loc}}} \rho\, e^{i\gamma_1C_{\mathsf{loc}}} \right)
\ee
where $\rho$ is a (mixed) $n_1$-qubit state defined as
\be
\label{channel2}
\rho = \mathrm{Tr}_{\calN_1^c} \left( e^{-i\gamma_1 C_{\mathsf{ent}}}  |+^n\ra\la +^n| e^{i\gamma_1 C_{\mathsf{ent}}} \right).
\ee
Here $\calN_1^c$ is the complement of $\calN_1$
and $\mathrm{Tr}_{\calN_1^c}$ denotes the partial trace over $\calN_1^c$.
\begin{prop}
\label{prop:rho}
There exists a classical algorithm that computes
the matrix of $\rho$ in the standard basis  in time $O(n_14^{n_1}+n3^{n_1})$.
\end{prop}
\begin{proof}
Assume wlog that $\calN_1=\{1,2,\ldots,n_1\}$. Let $m=|\calN_1^c|$ and $p(j)$ be the $j$-th
qubit of $\calN_1^c$, where $j=1,\ldots,m$.
Define $n_1$-qubit states $\rho_0,\rho_1,\ldots,\rho_m$
such that $\rho_0=|+^{n_1}\ra\la +^{n_1}|$ and 
\be
\label{channel3}
\rho_{j}= \mathrm{Tr}_{\mathsf{anc}} \left[  e^{-i\gamma_1 \sum_{q\in \calN_1} J_{p(j),q} Z_q Z_{\mathsf{anc}}}
\left(\rho_{j-1}\otimes  |+\ra\la +| \right)
 e^{i\gamma_1 \sum_{q\in \calN_1} J_{p(j),q} Z_q Z_{\mathsf{anc}}}
\right]
\ee
for all $j=1,2,\ldots,m$. Here we consider a system of $n_1+1$ qubits with the
first $n_1$ qubits describing $\calN_1$ and the last qubit serving as ancilla. 
The notation $\mathrm{Tr}_{\mathsf{anc}}$ means the partial trace over the ancilla.
We claim that $\rho_m=\rho$ is the state defined in Eq.~(\ref{channel2}).
Indeed, one can consider the righthand side of Eq.~(\ref{channel2}) as the output state of a 
quantum channel that takes as input the state $\rho_0=|+^{n_1}\ra\la +^{n_1}|$,
introduces $m=n-n_1$ ancillary qubits initialized in the $|+\ra\la +|$ state,
couples  $\calN_1$ with the ancillary qubits by $e^{-i\gamma_1 C_{\mathsf{ent}}}$, and finally traces out the ancillas.
Since all terms in $C_{\mathsf{ent}}$ pairwise commute, 
the same quantum channel can be implemented sequentially such that the $j$-th step
introduces a single ancilla qubit $p(j)\in \calN_1^c$ initialized in the  state $|+\ra\la +|$,
applies all  interactions in $C_{\mathsf{ent}}$ that act non-trivially on $p(j)$,
and traces out  $p(j)$. This is described by Eq.~(\ref{channel3}).

A simple algebra shows that Eq.~(\ref{channel3}) is equivalent to 
\be
\label{channel4}
\la x|\rho_j|y\ra= \la x|\rho_{j-1}|y\ra \cdot \cos{\left( 2\gamma_1 \sum_{q\in \calN_1} J_{p(j),q} (x_q - y_q)\right)}
\ee
for all $x,y\in \{0,1\}^{n_1}$.
Recalling that $\rho_0=|+^{n_1}\ra\la +^{n_1}|$  and $\rho=\rho_m$ one arrives at
\be
\label{channel5}
\la x|\rho|y\ra = 2^{-n_1} \prod_{j=1}^m  \cos{\left( 2\gamma_1 \sum_{q\in \calN_1} J_{p(j),q} (x_q-y_q)\right)}.
\ee
Define a function $f\, : \, \{-1,0,+1\}^{n_1} \to \mathbb{R}$ such that 
\[
f(z) =  2^{-n_1} \prod_{j=1}^m  \cos{\left( 2\gamma_1 \sum_{q\in \calN_1} J_{p(j),q} z_q \right)}.
\]
From Eq.~(\ref{channel5}) one infers that  $\la x|\rho|y\ra = f(x-y)$ for all $x,y\in \{0,1\}^{n_1}$.
Let us compute the matrix of $\rho$ by iterating over strings
$z\in \{-1,0,+1\}^{n_1}$ using ternary Gray code.
In other words, 
each iteration changes only one
component  of $z$.  We shall maintain a list of values of $m$ linear functions
\[
\ell_j(z)= 2\gamma_1 \sum_{q\in \calN_1} J_{p(j),q} z_q, \qquad j=1,2,\ldots,m.
\]
Updating the value of each function $\ell_j(z)$ upon changing a single component $z_q$
takes a constant time. Thus
updating the value of 
$f(z)=2^{-n_1} \prod_{j=1}^m \cos{(\ell_j(z))}$
takes time $O(m)$.  The number of pairs $x,y\in \{0,1\}^{n_1}$
such that $z=x-y$ is equal to $2^{w(z)}$, where
$w(z)$ is the number of zeros in $z$. 
Identifying all such pairs $x,y$ and recording the value $f(z)$
to the matrix elements $\la x|\rho|y\ra$ takes time  $O(n_1 2^{w(z)})$ for a fixed $z$.
We conclude that computing the full matrix of $\rho$ takes time
\[
O(m 3^{n_1}) + \sum_{z\in  \{-1,0,+1\}^{n_1}} O(n_1 2^{w(z)}) =
O(m3^{n_1} + n_14^{n_1})=O(n3^{n_1} + n_14^{n_1}).
\]
In the last equality we noted that $m=n-n_1\le n$. 
\end{proof}
From  Eq.~(\ref{channel1}) one gets $\eta(v)=\mathrm{Tr}(O_{\calN_1} \tilde{\rho})$, where  
\[
\tilde{\rho} =   e^{-i\gamma_1C_{\mathsf{loc}}} \rho\, e^{i\gamma_1C_{\mathsf{loc}}}.
\]
The next step is to compute the matrix of $\tilde{\rho}$ in the standard basis. We have
\be
\label{channel6}
\la x|\tilde{\rho}|y\ra = \la x|\rho|y\ra \cdot e^{-i\gamma_1\la x|C_{\mathsf{loc}}|x\ra +i\gamma_1 \la y|C_{\mathsf{loc}}|y\ra}.
\ee
A lookup table of the function $\la x|C_{\mathsf{loc}}|x\ra$ with $x\in \{0,1\}^{n_1}$ can be computed in time
$O(n_12^{n_1})$ by iterating over $x$'s using Gray code.
Then one can compute the matrix of $\tilde{\rho}$ using Eq.~(\ref{channel6}) in time $O(4^{n_1})$
by the brute force method. 

It remains to compute  $\eta(v)=\mathrm{Tr}(O_{\calN_1} \tilde{\rho})$.
We claim that this can be done in time $O(4^{n_1})$, assuming that the matrix of $\tilde{\rho}$ is already
available. Indeed, let $\calN_1=AB$, where $A$ and $B$ are subsets of size at most  $(n_1+1)/2$.
Let  $O_A$ and $O_B$ be the tensor products of the single-qubit operators $O_p(v)$
with $p\in A$ and $p\in B$ respectively. Then $O_{\calN_1} = O_A\otimes O_B$.
Compute the matrix of $O_A$ and $O_B$ in the standard basis. 
This takes time $O(|A| 4^{|A|}+|B|4^{|B|})=O(n_12^{n_1})$ if one uses the brute force method.
Finally, we have
\[
\eta(v)=
\mathrm{Tr}(O_{\calN_1} \tilde{\rho})=
\sum_{\alpha,\alpha'\in \{0,1\}^{|A|}}\;
\sum_{\beta,\beta'\in \{0,1\}^{|B|}}\;
\la \alpha' |O_A|\alpha\ra \cdot \la \beta'|O_B|\beta\ra \cdot \la \alpha,\beta|\tilde{\rho}|\alpha',\beta'\ra.
\]
This sum can be computed by the brute force method in time $O(4^{|A|+|B|})=O(4^{n_1})$ since the matrices
of $O_A,O_B$, and $\tilde{\rho}$ are already available. 

Summing up the runtime of each step in the algorithm and recalling that $n=n_2$
after restricting the problem to the subgraph induced by $\calN_2$ gives the overall runtime $T''$  in Eq.~(\ref{T1T2}).

{\em Comment:} Numerical data for RQAOA reported in
Section~\ref{sec:RQAOA}
 was generated by combining three 
algorithms for computing expected values of observables $ZZ$ on the QAOA states:
the forrelation-based algorithm described in Section~\ref{sec:graph_based} and Section~\ref{sec:QAOA_map},
algorithm $\calA'$, and a simplified version of  algorithm $\calA''$.
The latter computes the matrix $\rho$ defined in 
Eq.~(\ref{channel5}) in time $O(n_2(s,t)4^{n_1(s,t)})$ by iterating over bit strings $x,y$ in Eq.~(\ref{channel5})  using the binary
Gray code.
This is only a minor slowdown compared with the algorithm based on ternary Gray code described above
which computes $\rho$ in time 
$O(n_1(s,t)4^{n_1(s,t)} + n_2(s,t) 3^{n_1(s,t)})$.
Algorithms $\calA'$ or $\calA''$ were selected only if
their estimated runtime 
is below a cutoff value of $0.1$ second.
Otherwise, the forrelation-based algorithm was selected.

\section{Tree decomposition algorithm}
\label{app:tree}
\begin{algorithm}[H]
	\caption{Tree decomposition of an outerplanar graph \cite{bodlaender88} \label{alg:outerplanar}}
	\begin{algorithmic}[1]

		\Function{OuterplanarTD}{$G$}
 
		\If{all vertices of $G$ have degree zero}
       \State \Return{$T=(V,\emptyset)$, $\mathcal{B}=\{B_v=\{v\}: v\in V\}$}
		\ElsIf{$G$ has a degree-$1$ vertex}
			\State{Choose any degree-$1$ vertex $v\in V$.}
			\State Let $u\in V$ be the unique neighbor of $v$.
			\Let{$V'$}{$V\setminus \{v\}$}
			\Let{$E'$}{$E\setminus \{v,u\}$}
			\Let{$G'$}{$(V',E')$}
			\Let{$T=(W,F),\mathcal{B}=\{B_i:i\in V_T\}$}{OuterplanarTD($G'$)}
			 \State Find a node $i\in W$ such that $u\in B_i$.
			\Let{$W'$}{$W\cup{\alpha}$}
			\Let{$B_{\alpha}$}{$\{u,v\}$}
			\Let{$F'$}{$F\cup\{i,\alpha\}$}
			\Return{$T'=(W',F')$, $\mathcal{B}'=\mathcal{B}\cup{B_{\alpha}}$}
		\ElsIf{$G$ has a degree-$2$ vertex}
			\State{Choose any degree-$2$ vertex $v\in V$.}
			\State Let $u,w\in V$ be the two distinct neighbors of $v$.
			\Let{$V'$}{$V\setminus \{v\}$}
			\Let{$E'$}{$\left(E\setminus \{\{v,u\},\{v,w\}\}\right)\cup \{u,w\}$}
			\Let{$G'$}{$(V',E')$}
			\Let{$T=(W,F),\mathcal{B}=\{B_i:i\in V_T\}$}{OuterplanarTD($G'$)}
			 \State Find a node $i\in W$ such that $w\in B_i$ and $u\in B_i$.
			\Let{$W'$}{$W\cup{\alpha}$}
			\Let{$B_{\alpha}$}{$\{u,v,w\}$}
			\Let{$F'$}{$F\cup\{i,\alpha\}$}
			\Return{$T'=(W',F')$, $\mathcal{B}'=\mathcal{B}\cup{B_{\alpha}}$}
		\EndIf
 		
		\EndFunction
	\end{algorithmic}
\end{algorithm}

\section{Hardness of approximation for a small relative error}
\label{sec:relative_error}

In the following we consider the problem of approximating a forrelation $\Phi(f,g)$ to within a given \textit{relative} error, where $f,g$ are two-local functions.

Firstly, let us see how to establish hardness of approximation for graph-based forrelations arising from stabilizer states~\cite{aaronson2004improved}. As was shown in~\cite{van2004graphical}, for any $n$-qubit stabilizer state $|\psi\ra$ where exists an $n$-vertex graph $G=(V,E)$ and
single-qubit Clifford operators $C_1,\ldots,C_n$ such that 
\[
|\psi\ra =(C_1\otimes C_2 \otimes \cdots \otimes C_n) U_f H^{\otimes n}|0^n\ra,
\qquad f(x)=\prod_{\{u,v\}\in E} (-1)^{x_u x_v}.
\]
Clearly, $f$ is a two-local function on $G$. Thus the expected value of any tensor product operator
on any $n$-qubit stabilizer state can be expressed as a graph-based forrelation for a suitable $n$-vertex graph and two-local functions $f=g$.
As a corollary, approximating a graph-based forrelation with 
a small relative error is $\#$P-hard, even on a two-dimensional grid graph $G$. This follows from the facts that (a) any output
probability of a quantum circuit can be expressed as
the expected value of a tensor product observable on a 2D
cluster state using measurement-based quantum computing~\cite{raussendorf2001one}, and (b) estimating
the output probabilities of quantum circuits with a small
relative error is $\#$P-hard~\cite{goldberg2017complexity}.

In the remainder of this section we show that this hardness
persists for the standard forrelation problem with $O_j=H$ for all $j$. In the following we write $\mathbb S^1$ for the unit circle in the complex plane.

To this end, we consider so-called IQP circuits \cite{bremner2011classical}. In Ref.\cite{bremner2011classical} it is shown that given an $N$-qubit, $m$-gate quantum circuit $C$ we may efficiently compute $n\leq N+m$ and a two-local function $h:\{0,1\}^n\rightarrow \mathbb S^1$ on a two-dimensional grid graph such that 
\[
\langle 0^n|H^{\otimes n} U_h H^{\otimes n}|0^n\rangle=\sqrt{2}^{n-N} \langle 0^N|C|0^N\rangle.
\]
This is used in Ref.\cite{bremner2011classical} to show that postselected IQP circuits on planar graphs are as powerful as postBQP. Hardness of approximation for IQP circuit amplitudes then follows from the well known fact that a quantum circuit amplitude  $\langle 0^N|C|0^N\rangle$ is \#P-hard to approximate to a given relative error.

\begin{lemma}[\cite{bremner2011classical}]
Given a two-local function $h:\{0,1\}^n\rightarrow \mathbb S^1$ on a planar graph $G$ and $\epsilon>0$, it is \#P-hard to compute an approximation $Z$ such that
\[
(1-\epsilon) \langle 0^n|H^{\otimes n} U_h H^{\otimes n}|0^n\rangle\leq Z\leq (1+\epsilon)\langle 0^n|H^{\otimes n} U_h H^{\otimes n}|0^n\rangle.
\]
\label{lem:planIQP}
\end{lemma}

To establish the following lemma, we show that approximating forrelation of two-local functions on a planar graph is as hard as approximating IQP circuit amplitudes.

\begin{lemma}
Given $\epsilon>0$ and two-local functions $f,g:\{0,1\}^n\rightarrow \mathbb S^1$ on a planar graph $G$, it is \#P-hard to compute an approximation $\tilde{\Phi}$ such that 
\[
(1-\epsilon) \Phi(f,g)\leq \tilde{\Phi}\leq (1+\epsilon) \Phi(f,g).
\]
\label{lem:approxfor}
\end{lemma}
\begin{proof}
Suppose we are given a two-local function $h:\{0,1\}^n\rightarrow \mathbb S^1$ on a planar graph $G$.  Below we (efficiently) construct two-local functions $f,g:\{0,1\}^n\rightarrow \mathbb S^1$ on $G$ such that 
\begin{equation}
\langle0^n|H^{\otimes n} U_f H^{\otimes n} U_g H^{\otimes n}|0^n\rangle=\langle 0^n|H^{\otimes n} U_h H^{\otimes n}|0^n\rangle.
\label{eq:forIQP}
\end{equation}
The claimed hardness of approximation then follows immediately from Lemma \ref{lem:planIQP}. 

To prove Eq.~\eqref{eq:forIQP} we use the one-qubit identity
\begin{equation}
e^{-i\pi/4} SHS=HS^{\dagger}H,
\label{eq:1q}
\end{equation}
where $S=\mathrm{diag}(1,i)$ is the phase gate. Define functions $f,g$ such that $U_g=(S^{\dagger})^{\otimes n}$ and $U_f=e^{i\pi n/4}U_h U_g$. Note that both $f$ and $g$ are two-local on $G$.  Then 
\[
\langle0^n|H^{\otimes n} U_f H^{\otimes n} U_g H^{\otimes n}|0^n\rangle= e^{-i\pi n/4}\langle0^n|H^{\otimes n} U_f S^{\otimes n}H^{\otimes n}|0^n\rangle= \langle 0^n|H^{\otimes n} U_h H^{\otimes n}|0^n\rangle,
\]
where in the first equality we used Eq.~\eqref{eq:1q}. This establishes Eq.~\eqref{eq:forIQP} and completes the proof.

\end{proof}

The graph $G$ in lemmas \ref{lem:planIQP} and \ref{lem:approxfor} may be taken to be a two-dimensional grid graph without loss of generality.


\end{document}